\documentclass[a4paper]{article}

\usepackage{a4wide}

\usepackage[ansinew]{inputenc}
\usepackage{amsfonts}
\usepackage{lmodern}
\usepackage{mathrsfs}
\usepackage[all]{xy}
\usepackage{amsthm, amssymb, amsmath, latexsym}
\usepackage{color}
\usepackage{array, float}
\usepackage{bm}
\usepackage{hyperref}
\usepackage{tikz-cd}
\usepackage{graphicx}
\usepackage{comment}
 \graphicspath{{Images/}}

\usepackage[normalem]{ulem}

\newcommand{\eps}{\varepsilon}
\newcommand{\norm}[2]{{\| #1 \|}_{#2}}
\newcommand{\dt}{\partial_t}
\newcommand{\e}{\varepsilon}

\newcommand{\R}{\mathbb R}
\newcommand{\Real}{\mathbb R}
\newcommand{\nab}[1]{\nabla^{#1}}

\newcommand{\sigm}[1]{\Sigma^{#1}}

\newcommand{\supp}{{\text{supp}}}

\newcommand{\kernel}{{\text{ker}}}

\newcommand{\gec}{{\check g_\eps}}
\newcommand{\N}{\mathbb N} 
\newcommand{\tnab}{{\tilde\nabla}}

\newtheorem{Theorem}{Theorem}[section]
\newtheorem{Corollary}[Theorem]{Corollary}
\newtheorem{Lemma}[Theorem]{Lemma}
\newtheorem{Proposition}[Theorem]{Proposition}
\theoremstyle{definition}
\newtheorem{Definition}[Theorem]{Definition}

\theoremstyle{remark}
\newtheorem{rem}[Theorem]{Remark}

\numberwithin{equation}{section}

\title{{Green Operators in Low Regularity Spacetimes and Quantum Field Theory}}

\author{G\" unther H\" ormann, Yafet Sanchez Sanchez, Christian Spreitzer, James Vickers}

\begin{document}

\date{\today}

\maketitle


\begin{abstract}

\noindent
 In this paper we develop the
  mathematics required in order to provide a description of the
  observables for quantum fields on low-regularity spacetimes. In
  particular we consider the case of a massless scalar field $\phi$ on a
  globally hyperbolic spacetime $M$ with $C^{1,1}$ metric $g$. This
  first entails showing that the (classical) Cauchy problem for the
  wave equation is well-posed for initial data and sources in 
  Sobolev spaces and then constructing low-regularity advanced and
  retarded Green operators as maps between suitable function
  spaces. In specifying the relevant function spaces we need to control
  the norms of both $\phi$ and $\square_g\phi$ in order to ensure that
  $\square_g \circ G^\pm$ and $G^\pm \circ \square_g$ are the identity
  maps on those spaces. The causal propagator $G=G^+-G^-$ is then
  used to define a symplectic form $\omega$ on a normed space
  $V(M)$ which is shown to be isomorphic to $\ker \square_g$. This
  enables one to provide a locally covariant description of the
  quantum fields in terms of the elements of quasi-local
  $C^*$-algebras.

\vskip 1em

 \medskip


\noindent
\emph{Acknowledgements} 
All authors acknowledge financial support
through FWF (Austrian Science Fund) project no. P28770, JAV
acknowledges financial support from STFC in the UK through Grant
No. ST/R00045X/1 and YSS would like to thank the Max Planck Institute
of Mathematics in Bonn and Leibniz University Hannover for their financial support
during the writing of this paper.
\end{abstract}

------------------------------------------------------------------------

\section{Introduction}

This paper is concerned with developing the theory of quantum fields
on low regularity spacetimes. We will follow the algebraic approach to
quantisation as described in \cite{wald} and \cite{algebraic}. In
particular we will draw heavily on the detailed scheme given in the
book \cite{bgp}. The starting point is a smooth Lorentzian manifold
$(M,g)$ and a field equation $P\phi=f$ where $P$ is a normally
hyperbolic differential operator $P$ acting on a vector bundle $F$. In
this paper we will consider the scalar operator $P=\square_g$ as there
are no significant additional mathematical issues in dealing with the
general case. The essence of the algebraic approach as outlined in
\cite{bgp} is to first construct the advanced and retarded Green
operators for $P$ and use these to construct the causal propagator
$G=G^+-G^-$. Note that in order for the Green operators to be unique
we require the spacetime to be globally hyperbolic. The causal
propagator is then used to construct a skew symmetric bilinear map on
the space of smooth functions of compact support by $\tilde
\omega(\phi,\psi):=\langle G(\phi),\psi\rangle_{L^2(M,g)}$. This form is
degenerate but gives rise to a symplectic form $\omega$ on the quotient space
${\mathcal D}(M)/{\rm ker}(G)$ of the test function space.
The next step in the process is to
use $\omega$ to construct representations of the canonical commutation
relations (CCRs) on the space of quasi-local $C^*$-algebras
\cite[Theorem 4.4.11]{bgp} which satisfy the Haag-Kastler axioms
\cite{HK}. Each of these steps may be described in terms of a functor
so that we have a functorial description of how to go to from the
category of globally hyperbolic manifolds equipped with (formally
self-adjoint) normally hyperbolic operators to the space of
quasi-local $C^*$-algebras whose elements describe the observables of
the field (see details below). Indeed this scheme gives rise to a
\emph{locally covariant quantum field} in the sense of
\cite{BFV}. Going from the rather abstract quantisation procedure
described here to the more familiar Fock space representation requires
one to pick out the physically relevant states. For the smooth case a
mathematically appealing criterion (which corresponds to the standard
answer in Minkowski space) is the micro-local spectrum condition of
Radzikowski \cite{rad}. However this is not directly applicable in the
low-regularity case. We return to this point and suggest a suitable
modification in the discussion section at the end.

In generalising the smooth results to the low regularity setting we
need to choose a class of metrics that are sufficiently regular to
establish the results we would like, while being sufficiently general to cover
the cases of physical interest. From this point of view the choice of
$C^{1,1}$ metrics is natural since it allows one to deal with
spacetimes where the curvature remains finite while allowing for
discontinuities in the energy-momentum tensor at, for example, an
interface or the surface of a star. From the mathematical point of
view $C^{1,1}$ is the minimal condition which ensures existence and
uniqueness of geodesics and for which the standard results from smooth
causality theory go through more or less unchanged \cite{KSSV}. It also
ensures that the solutions to the wave equation are in $H^2_{\text{loc}}(M)$
(see Appendix B for details of the function spaces we use) as shown in
Theorem \ref{globalexist} below, which ensures we have enough
regularity to define the quantisation functors we need. Although one
can define solutions to the wave equation for metrics of lower
regularity \cite{SSV1,SSV2} there are difficulties in defining
the corresponding advanced and retarded Green operators for these cases.

In Section \ref{RN} we establish the results we need to prove
existence and uniqueness of solutions to the forward (and backward)
initial value problem for the wave equation on $\Real^{n+1}$ for $C^{1,1}$
metrics. Rather than rework the entire theory of the wave equation for
metrics of low-regularity we use a method of regularising the
coefficients \cite{HS:12}, using the smooth theory to obtain the
corresponding solutions of the Cauchy problem and then using a
compactness argument to show that this converges to a weak
$H_{\text{loc}}^2(\Real^{n+1})$ solution of the original equation. This proceeds
via the theory of Colombeau generalised functions \cite{Col} and is
related to the work of \cite{Ruzhansky} on very weak
solutions. Furthermore by controlling the causal structure of the
regularisation $g_\varepsilon$ by insisting that
$J^{+}_{\varepsilon}(U)\subset J^{+}(U)$ we can ensure that the
forward solution $u^+$ with zero initial data satisfies the causal
support condition $\supp(u^+)\subset J^{+}(\supp(f))$ (cf.\ 
\cite[Theorem 2.6.4]{Ring}).

In Section \ref{GH} we introduce the notion of global hyperbolicity
\cite{saemann} and temporal functions for non-smooth metrics
\cite{conespace} as well as the other results from $C^{1,1}$ causality
theory that we require \cite{KSSV}.
We use the fact that even for
$C^{1,1}$ spacetimes the temporal function can be chosen to be smooth
so we can write $M$ as $\Real \times \Sigma$, where $\Sigma$
is a Cauchy surface,
and define function spaces where we
make a split between space and time. However, as far as possible we
formulate our final results in a way that is independent of the choice
of temporal function, so that the particular choice of space-time
split is not important. The remainder of the section shows how to go
from existence and uniqueness results on $\Real^{n+1}$ to global
results on a globally hyperbolic $C^{1,1}$ spacetime. Our approach to
this closely follows Ringstr\"om \cite{Ring} and the causality results
for $C^{1,1}$ metrics \cite{KSSV}  ensure that the existence proof remains
similar to the smooth case.

The next step is to define appropriate Green operators. In the smooth
case the Green operator takes (compactly supported) smooth functions
to smooth functions. However in the $C^{1,1}$ case it is crucial that the map
is between suitable Sobolev type spaces. In fact, we find precise
conditions on the regularity of the solutions and the causal support
in order to define unique Green operators in globally hyperbolic
spacetimes of limited differentiability. We need to control the
 (local) Sobolev norms of both $\phi$ and $\square_g\phi$ in order to
ensure that $\square_g \circ G^\pm$ and $G^\pm \circ \square_g$ are the
identity maps on the corresponding spaces.  The choice of function space is also
relevant in the definition of ${\rm ker}(G)$ which is used to
construct the factor space for the symplectic map $\omega$. We end the
section by considering the dependence of the construction on the
choice of temporal function.

The passage from the symplectic space defined by $\omega$ to the
canonical commutation relations proceeds almost identically to that of
the smooth case \cite{bgp} so we only sketch out the details although
the analogue of \cite[Theorem 4.5.1]{bgp} requires some work. 

We end the paper with a summary of the results we have obtained and a discussion
of the outstanding issues, including the choice of a Sobolev
micro-local spectrum condition to single out the physical states in the
low-regularity setting.

 Appendix A briefly describes the basic properties of regularisation methods we use while Appendix B gives a brief description of various Sobolev spaces we use in the paper.

{\bf{Notation.}} We denote the derivative of a function $u$ with
respect to $t$ by $u_{t}$ or $\partial_{t}u$ and by $u_{i}$ or
$\partial_{i}u$ if it is with respect to the spatial
$x^i$-coordinate. The space of smooth functions of compact support on a manifold $M$
will be denoted by ${\mathcal D}(M)$. A function $f$ on an open subset $\cal{U}$ of
$\mathbb{R}^{n}$ is said to be Lipschitz if there is some constant $K$
such that for each pair of points $p,q\in {\cal{U}}$, $|f(p) - f(q)| \leq K |p-q|$. We denote by
$C^{k,1}$ those $C^k$ functions where the $k^{\text{th}}$
derivative is a Lipschitz continuous function.  A function on a smooth
manifold is said to be Lipschitz or $C^{k,1}$ if it has this property
upon composition with any smooth coordinate chart. We gather basic notions and
results concerning Sobolev and related function spaces in 
Appendix B.

\section{The smooth setting} 

In this section we briefly review the results in the smooth
setting. The starting point is an existence and uniqueness result for
solutions to a smooth second order hyperbolic equation on
$\Real^{n+1}$ (see e.g. \cite[Theorem 8.6]{Ring}). This is used to
obtain a corresponding result showing the existence of a unique smooth
solution $u \in C^\infty(M)$ to the Cauchy initial value problem for
a smooth normally hyperbolic operator $P$ on a smooth globally
hyperbolic manifold with smooth spacelike Cauchy surface $\Sigma$ and normal vector field $n$ given by
\begin{align*}
Pu&=f \quad \hbox{on $M$ where $f \in {\mathcal D}(M)$} \\
u|_\Sigma&=u_0 \quad \hbox{on $\Sigma$ where $u_0 \in {\mathcal D}(\Sigma)$}\\
\nabla_n u|_\Sigma&=u_1 \quad \hbox{on $\Sigma$ where $u_1 \in {\mathcal D}(\Sigma)$}
\end{align*}
which satisfies the causal support condition
\begin{equation*}
\supp(u) \subset J(\supp(u_0) \cup \supp(u_1) \cup \supp(f)).
\end{equation*}
Note that the globally hyperbolic condition is essential in order to
ensure that the solution is unique.

The next step is to show the existence of advanced and retarded Green operators.
\begin{Definition}
\label{propagators}
  Let $M$ be a time-oriented connected Lorentzian manifold and let $P$
  be a normally hyperbolic operator. An \emph{advanced Green operator}
  $G^+$ is a linear map $G^+:{\mathcal D}(M) \to C^\infty(M)$ such
  that
\begin{enumerate}
\item $P \circ G^+ ={\rm id}_{{\mathcal D}(M)}$
\item $G^+ \circ P|_{{\mathcal D}(M)} ={\rm id}_{{\mathcal D}(M)}$
\item $\supp(G^+\phi) \subset J^+(\supp(\phi))$ 
for all $\phi \in {\mathcal D}(M)$
\end{enumerate}
A \emph{retarded Green operator} $G^-$ satisfies (1) and (2) but (3)
is replaced by $\supp(G^-\phi) \subset J^-(\supp(\phi))$ for all $\phi
\in {\mathcal D}(M)$.
\end{Definition}
Corollary 3.4.3 of \cite{bgp} shows that these exist and
are unique on a globally hyperbolic manifold.

The advanced and retarded Green operators are then used to define the
causal propagator $G:=G^+-G^-$ which maps ${{\mathcal D}(M)}$ to the space
of spatially compact maps $C^\infty_{{\rm sc}}(M)$ i.e.\ the smooth
maps $\phi$ such that there exists a compact subset $K \subset M$ with
$\supp(\phi) \subset J(K)$. If $M$ is globally hyperbolic then one has
the following \emph{exact sequence} \cite[Theorem 3.4.7]{bgp}

\begin{center}
\begin{tikzcd}
0 \arrow{r}  &{\mathcal D}(M) \arrow{r}{P} &{\mathcal D}(M) 
\arrow{r}{G} &C^\infty_{{\rm sc}}(M) \arrow{r}{P} &C^\infty_{{\rm sc}}(M),\\
\end{tikzcd}
\end{center}
in particular,  $\kernel(G)=P({\mathcal D}(M))$.

We want to use $G$ to construct a symplectic vector space to which one
can apply the CCR functor. We first define a skew symmetric bilinear
form on ${\mathcal D}(M)$ by
$\tilde\omega(\phi,\psi):=\langle G(\phi),\psi\rangle_{L^2(M,g)}$.  Unfortunately
the bilinear form is degenerate so it fails to provide the required
symplectic form. However we can rectify this by passing to the
quotient space $V:={\mathcal D}(M)/\kernel(G)$, which by the above is
just ${\mathcal D}(M)/P({\mathcal D}(M))$. Hence $\tilde \omega$
induces a symplectic form $\omega$ on $V$. One can go on to use this to
construct representations of the canonical commutation relations (CCRs) on the space of
quasi-local $C^*$-algebras \cite[Theorem 4.4.11]{bgp} which satisfy
the Haag-Kastler axioms \cite[Theorem 4.5.1]{bgp}.

\section{The Cauchy problem on $\Real^{n+1}$ for $C^{1,1}$
    metrics}\label{RN}

  In this section we establish Theorem \ref{Rn} which gives the
  existence, uniqueness and causal support results we need concerning
  solutions to the Cauchy problem for the wave equation on
  $\Real^{n+1}$ for a $C^{1,1}$ metric.

The proof follows from Lemmas \ref{energy_estimates}, \ref{gensol},
\ref{weaksol} below, which cover a slightly more general version of
the Cauchy problem. The basic technique is to employ a
Chru\'sciel-Grant regularisation of the metric \cite{cg} (see details
in Appendix A) to obtain a family of smooth metrics
$(g_{\varepsilon})_{\varepsilon\in (0,1]}$ which converge to $g$ in
  the $C^1$ topology, have uniformly bounded second derivatives on
  compact sets and satisfy $J^{+}_{\varepsilon}(K)\subset J^{+}(K)$
  for every compact subset $K \subset M$ and $\varepsilon >0$. This
  enables us to obtain the required solution of the wave equation by
  taking a suitable limit of solutions to the smooth equation as
  $\varepsilon \to 0$ while preserving the causal support
  properties. Lemma \ref{energy_estimates} provides a detailed form of
  the energy estimates which proves crucial in the transition to the
  non-smooth setting. The causal support properties follow from Lemma
  \ref{causalsupport_Rn}. Note that although Lemma \ref{gensol} shows
  the existence of a generalised Colombeau solution \cite{Col} to a
  corresponding generalised Cauchy problem the main result of this
  section, Theorem \ref{Rn}, concerns a classical weak solution and does
  not require explicitly referring to the Colombeau solution used in the
  proof.

As a basic setup, we consider a Lorentzian metric $g$ of signature
$(+,-,...,-)$ on $\mathbb R^{n+1}$ with spatial components
$(-h_{ij})_{1\leq i,j \leq n}$ and the corresponding wave operator
\begin{equation}\label{wave_eq}
  Pu:=g_{00}\partial_t^2 u+2\sum_{j=1}^n g_{0j}\partial_{x_j}\partial_t u-\sum_{i,j=1}^n h_{ij}\partial_{x_i}\partial_{x_j} u +\sum_{j=1}^n a_j\partial_{x_j} u + a_0\partial_t u+bu 
\end{equation}
where all coefficients are supposed to be real-valued, smooth, and
bounded in all orders of derivatives. Moreover we assume that there
exist positive constants $\tau_{\text{min}}$, $\tau_{\text{max}}$,
$\lambda_{\text{min}}$, $\lambda_{\text{max}}$ such that

\begin{equation}\label{taumin}
 \tau_{\text{min}}\leq g_{00}(t,x)\leq \tau_{\text{max}} 
\end{equation}
\begin{equation}\label{lambdamin}
 \lambda_{\text{min}}|\xi|^2\leq\sum_{i,j}h_{ij}(t,x)\overline{\xi_i}\xi_j\leq \lambda_{\text{max}}|\xi|^2
 \end{equation}
 for all $(t,x)\in \mathbb R^{n+1} $ and for all $\xi\in \mathbb
 C^{n}$.  It follows from Theorem 23.2.2 in \cite{Hoermander:3} or
 Theorem 2.6 in \cite{BG:07} that for initial data $u_0,u_1\in
 H^{\infty}(\mathbb R^n)$  (see Appendix A) and right-hand side
 $f \in H^\infty((0,T) \times \R^n) \subset C^\infty([0,T] \times \R^n)$
 the Cauchy problem
\begin{equation}\label{cauchy} Pu=f\quad\text{on}\quad \Omega_T:=(0,T)\times \mathbb R^n, 
 \qquad u|_{t=0}=u_0,  \qquad \partial_t u|_{t=0} = u_1 
 \end{equation}
 has a unique solution $u$ which belongs to $
 C([0,T],H^{s+1}(\mathbb R^n)) \cap C^{1}([0,T],H^{s}(\mathbb
 R^n))\cap H^{s+1}((0,T)\times \mathbb R^n)$ for every $s\in \mathbb
 R$. In the following lemma we provide an explicit energy estimate
 with $s=1$.

\begin{Lemma}\label{energy_estimates}
Let $u\in C([0,T],H^{2}(\mathbb R^n)) \cap C^{1}([0,T],H^{1}(\mathbb R^n))\cap H^{2}((0,T)\times \mathbb R^n)$ such that $Pu \in  L^{2}([0,T],H^{1}(\mathbb R^n))$ and 
denote $\tilde u_0:=u(0,\cdot)$ and $\tilde u_1:=\partial_tu(0,\cdot)$. Then  we have
\begin{equation*}
||u||_{C([0,T],H^{2}(\mathbb R^n))}+||u||_{C^1([0,T], H^{1}(\mathbb R^n))} 
\le e^{\beta T} \Big( \norm{\tilde u_0}{H^2(\mathbb R^n)}+ \norm{\tilde u_1}{H^1(\mathbb R^n)}+  ||Pu||_{L^{2}([0,T],H^{1}(\mathbb R^n))}\Big).  
\end{equation*}

\noindent
where $\beta=C_{g',a',b'}/\min(\tau_{\min},\lambda_{\min})$ and
$C_{g',a',b'}$ contains $L^{\infty}(\Omega_T)$-norms of at most
first-order derivatives of the coefficients $g,a,b$.
\end{Lemma}

\begin{proof} 
  We show how to derive the energy estimate in terms of application of
  $P$ to any function $u\in H^{\infty}(\mathbb R^{n+1})$: We may write
\begin{equation}\label{div_form}
 2\text{Re}(Pu \,\dt \overline{u} )= \dt\mathcal T+ \sum\limits_{j=1}^n \partial_{x_j} X_j + Y
\end{equation}
where
\begin{eqnarray*}
\mathcal T(u(t,x);t,x)  
&=& g_{00}(t,x)\left|\dt u(t,x)\right|^2 
+\sum\limits_{i,j=1}^nh_{ij}(t,x)\partial_{x_i}u(t,x)\,\partial_{x_j}\overline{u}(t,x)\\
X_j(u(t,x);t,x) &=& 2g_{0j}(t,x)\left|\partial_t u(t,x)\right|^2  - 2\text{Re}\sum\limits_{i=1}^n   h_{ji}(t,x)\,\partial_{x_i}u(t,x)\,\dt \overline{ u} (t,x)
\end{eqnarray*}
and
\begin{multline*}
Y(u(t,x);t,x) = \Big(- g_{00,0}(t,x) -2 \sum\limits_{j=1}^n g_{0j,j}(t,x) +a_0(t,x)\Big)\left|\partial_t u(t,x)\right|^2 \\
+  \sum\limits_{j=1}^n2\text{Re} \Big(\sum\limits_{i=1}^n   h_{ij,i}(t,x)+a_j(t,x)\Big)\partial_{x_j}u(t,x)\, \dt \overline{u}(t,x) 
 + b(t,x)u(t,x)\dt \overline{u}(t,x).
\end{multline*}
Here we have used the notation $f_{i,\mu}$ for
$\partial_{x^{\mu}}f_{i}$, where Greek indices range from 0 to $n$ and
$x^0:=t$.  Considering time $t$ as a parameter and integrating
equation \eqref{div_form} over the spatial domain $\mathbb R^n$, we
obtain
 \begin{multline}\label{div_form_x_int}
 \frac{d}{dt}\int_{\mathbb R^n} \mathcal T(u(t,x);t,x) d(x_1,...,x_n)=\\ 2\text{Re}\int_{\mathbb R^n} (Pu \dt \overline{u})(t,x) d(x_1,...,x_n)- \int_{\mathbb R^n} Y(u(t,x);t,x) d(x_1,...,x_n)
 \end{multline}
 where we have used that $\int_{\mathbb R^n} \partial_{x_j} X_j
 (u(t,x);t,x)d(x_1,...,x_n)=0$ for all $j=1,...,n$ and for all $t\in
 [0,T]$, since $x\mapsto u(t,x)$ belongs to $H^1(\mathbb R^n)$ and the
 latter possesses $ C^{\infty}_c(\mathbb R^n)$ as a dense subspace. Using the notation $g_{\mu\,\cdot}$ for the $\mu$-th row vector and $\text{div}g_{\mu\,\cdot}$ for its divergence, we
 may estimate $Y(u(t,x);t,x)$ as follows:
\begin{multline*}
|Y(u(t,x);t,x)|\leq  \norm{g_{00,0} +2\, \text{div} \,g_{0\,\cdot} -a_0}{L^{\infty}(\Omega_{T})}\left|\partial_t u(t,x)\right|^2 \\
+  \max_{1\leq i\leq  n} \Big(n\,\norm{\text{div}\, g_{i\,\cdot}}{L^{\infty}(\Omega_{T})} 
 +\norm{a_i}{L^{\infty}(\Omega_{T})}\Big)
 \Big( \sum\limits_{j=1}^n|\partial_{x_j}u(t,x)|^2+ |\dt u(t,x)|^2\Big)\\
 +\frac{1}{2}\norm{b}{L^{\infty}(\Omega_{T})} \Big( |u(t,x)|^2+|\dt u(t,x)|^2\Big)\\
 = \Big( \norm{g_{00,0} +2\, \text{div} \,g_{0\,\cdot} -a_0}{L^{\infty}(\Omega_{T})} +\max_{1\leq i\leq  n} \Big(n\,\norm{\text{div}\, g_{i\,\cdot}}{L^{\infty}(\Omega_{T})} 
 +\norm{a_i}{L^{\infty}(\Omega_{T})}\Big)+\frac{1}{2}\norm{b}{L^{\infty}(\Omega_{T})}\Big) \left|\partial_t u(t,x)\right|^2\\
 +\max_{1\leq i\leq  n} \Big(n\,\norm{\text{div}\, g_{i\,\cdot}}{L^{\infty}(\Omega_{T})} 
 +\norm{a_i}{L^{\infty}(\Omega_{T})}\Big)\sum\limits_{j=1}^n|\partial_{x_j}u(t,x)|^2 + \frac{1}{2}\norm{b}{L^{\infty}(\Omega_{T})} |u(t,x)|^2.
\end{multline*}
We note that 
$$
  \int_{\mathbb R^n} |Y(u(t,x);t,x)| d(x_1,...,x_n)\leq C_{g',a,b} \norm{u(t,\cdot)}{{\tilde H}^1(\mathbb R^n)}^2
$$
 where $C_{g',a,b}$ depends on the the metric $g$ and its first derivatives  (more precisely, only on $g_{0\mu,\mu}$ and $g_{ij,j}$), as well as on $a$ and $b$.  (See Appendix B for the definition of ${\tilde H}^1(M)$.)
 Moreover we have
 \begin{multline*}
C_{\text{min}}\Big(\norm{\dt u(t,\cdot)}{L^2(\mathbb R^n)}^2+\norm{\nabla u(t,\cdot)}{L^2(\mathbb R^n)}^2\Big) \leq \int_{\mathbb R^n} \mathcal T(u(t,x);t,x) d(x_1,...,x_n) \\
 \leq C_{\text{max}}\Big(\norm{\dt u(t,\cdot)}{L^2(\mathbb R^n)}^2+\norm{\nabla u(t,\cdot)}{L^2(\mathbb R^n)}^2\Big)
 \end{multline*}
 where $C_{\text{min}}:=\min(\tau_{\text{min}},\lambda_{\text{min}})$
 and
 $C_{\text{max}}:=\max(\tau_{\text{max}},\lambda_{\text{max}})$. We
 observe that $ |2\text{Re}(Pu \partial_t \overline{u})(t,x) | \leq
 |Pu(t,x)|^2+ |\dt u(t,x)|^2.  $ Using the fundamental theorem of
 calculus we may write
  $$
 |u(t,x)|^2=|u(0,x)|^2+\int\limits_0^{t}\partial_s |u(s,x)|^2 ds \leq 
 |u(0,x)|^2 + \int\limits_0^{t} |\partial_s u(s,x)|^2 ds + 
 \int\limits_0^{t}  |u(s,x)|^2 ds 
 $$
 and  integration over the spatial variables yields
 \begin{multline}\label{minmax}
 C_{\text{min}} \norm{u(t,\cdot)}{L^2(\mathbb R^n)}^2\leq \\
  C_{\text{max}} 
 \left(  \norm{u(0,\cdot)}{L^2(\R^n)}^2 + \int\limits_0^{t} \norm{\partial_s u(s,\cdot)}{L^2(\R^n)}^2 ds + 
 \int\limits_0^{t}  \norm{u(s,\cdot)}{L^2(\R^n)}^2 ds \right)
 \end{multline} 
 where we have deliberately multiplied by $C_{\text{min}} \leq
 C_{\text{max}}$.  Integrating equation \eqref{div_form_x_int} with
 respect to the time variable from 0 to $t$ we obtain
 \begin{multline*}
 \int_{\mathbb R^n}\mathcal T(u(t,\cdot);t,\cdot)dV \leq\\
  \int_{\mathbb R^n}\mathcal T(u(0,\cdot);0,\cdot)dV +2 \text{Re}\int\limits_0^t  \int_{\mathbb R^n} (Pu \dt \overline{u})(s,\cdot) dV ds +\int\limits_0^t  \int_{\mathbb R^n} |Y(u(s,\cdot);s,\cdot)| dVds. 
\end{multline*}
We will from now on write $\tilde u_0:= u(0,\cdot)$ and $\tilde
u_1:= \partial_tu(0,\cdot)$.  Upon adding inequality \eqref{minmax} we
get
\begin{multline*} \int_{\mathbb R^n}\mathcal T(u(t,\cdot);t,\cdot)dV+C_{\text{min}} \norm{u(t,\cdot)}{L^2(\mathbb R^n)}^2 \leq \int_{\mathbb R^n}\mathcal T(u(0,\cdot);0,\cdot)dV+2 \text{Re}\int\limits_0^t  \int_{\mathbb R^n} (Pu \,\dt \overline{u})(s,\cdot) dV ds\\
+\int\limits_0^t  \int_{\mathbb R^n} |Y(u(s,\cdot);s,\cdot)| dVds
+C_{\text{max}} 
 \left(  \norm{\tilde{u}^0}{L^2(\mathbb R^n)}^2 + \int\limits_0^{t} \norm{\partial_s u(s,\cdot)}{L^2(\mathbb R^n)}^2 ds + 
 \int\limits_0^{t}  \norm{u(s,\cdot)}{L^2(\mathbb R^n)}^2 ds \right)
\end{multline*}
Employing the lower bound $ C_{\text{min}}\Big(\norm{\dt
  u(t,\cdot)}{L^2(\mathbb R^n)}^2+\norm{\nabla u(t,\cdot)}{L^2(\mathbb
  R^n)}^2\Big) \leq \int_{\mathbb R^n}\mathcal T(u(t,\cdot);t,\cdot)
dV $ and the upper bound $\int_{\mathbb R^n}\mathcal
T(u(0,\cdot);0,\cdot)dV\leq
C_{\text{max}}\Big(\norm{\tilde{u}^1}{L^2(\mathbb R^n)}^2+\norm{\nabla
  \tilde{u}^0}{L^2(\mathbb R^n)}^2\Big)$ as well as the other
estimates and rearranging terms, we arrive at
 \begin{multline}\label{energy_estimate_1}
 C_{\text{min}}\Big(\norm{u(t,\cdot)}{L^2(\mathbb R^n)}^2+ \norm{\dt u(t,\cdot)}{L^2(\mathbb R^n)}^2 +\norm{\nabla u(t,\cdot)}{L^2(\mathbb R^n)}^2\Big)\\
\leq C_{\text{max}}\Big(\norm{\tilde{u}^0}{L^2(\mathbb R^n)}^2+ \norm{\tilde{u}^1}{L^2(\mathbb R^n)}^2 +\norm{\nabla \tilde{u}^0}{L^2(\mathbb R^n)}^2\Big)
+ \int\limits_0^t\norm{Pu(s,\cdot)}{L^2(\mathbb R^n)}^2 ds\\
+  (1+\max( C_1, C_2, C_3)) \int\limits_0^t \Big(\norm{u(s,\cdot)}{L^2(\mathbb R^n)}^2+ \norm{\partial_s u(s,\cdot)}{L^2(\mathbb R^n)}^2 +\norm{\nabla u(s,\cdot)}{L^2(\mathbb R^n)}^2\Big)ds
  \end{multline}
 where
 \begin{align*}
C_1&= 1+ \norm{g_{00,0} +2\, \text{div} \,g_{0\,\cdot} -a_0}{L^{\infty}(\Omega_{\tau})} +\max_{1\leq i\leq  n} \Big(n\,\norm{\text{div}\, g_{i\,\cdot}}{L^{\infty}(\Omega_{T})} 
 +\norm{a_i}{L^{\infty}(\Omega_{\tau})}\Big)+\frac{\norm{b}{L^{\infty}(\Omega_{T})}}{2}, \\
 C_2 &= \max_{1\leq i\leq  n} \Big(n\,\norm{\text{div}\, g_{i\,\cdot}}{L^{\infty}(\Omega_{T})} 
 +\norm{a_i}{L^{\infty}(\Omega_{\tau})}\Big),\\
 C_3 &= \frac{1}{2}\norm{b}{L^{\infty}(\Omega_{T})}.
  \end{align*}
  Gronwall's inequality \cite[Chapter XVIII, \S 5, Section 2.2]{DL:V5}
 then yields the a-priori energy estimate
 \begin{multline*}
 \sup_{0\leq t \leq T}\Big(\norm{u(t,\cdot)}{H^1(\R^n)}^2+ \norm{\dt u(t,\cdot)}{L^2(\R^n)}^2 \Big)\\
 \leq  C_{\text{min}}^{-1} \Big( C_{\text{max}}
 \Big(\norm{\tilde{u}^0}{L^2(\mathbb R^n)}^2+ \norm{\tilde{u}^1}{L^2(\mathbb R^n)}^2 +\norm{\nabla \tilde{u}^0}{L^2(\mathbb R^n)}^2\Big)+ \norm{Pu}{L^2(\Omega_{T})}^2\Big) \cdot e^{\beta T}
 \end{multline*}
 where $\beta:= C_{\text{min}}^{-1}(1+\max( C_1, C_2, C_3))$.\\

 To obtain similar estimates for spatial derivatives $\nabla
 u=(\partial_{x_1}u,...,\partial_{x_n}u)^{\mathsf{T}}$ of $u$, we
 first divide the equation by $g_{00}$:
\begin{equation*}
		\partial_t^2 u+2\sum_{j=1}^n \frac{g_{0j}}{g_{00}}\partial_{x_j}\partial_t u-\sum_{i,j=1}^n \frac{h_{ij}}{g_{00}}\partial_{x_i}\partial_{x_j} u +\sum_{j=1}^n \frac{a_j}{g_{00}}\partial_{x_j} u + \frac{a_0}{g_{00}}\partial_t u+\frac{b}{g_{00}}u =\frac{Pu}{g_{00}}.
		\end{equation*}
		
Differentiating this equation with respect to $x^k$  yields
 \begin{multline}\label{divided_diff}
		\partial_t^2 \partial_{x_k} u  
		+2\sum_{j=1}^n \frac{g_{0j,k}g_{00}-g_{0j}g_{00,k}}{g_{00}^2}\partial_{x_j}\partial_t u+2\sum_{j=1}^n \frac{g_{0j}}{g_{00}}\partial_{x_j}\partial_t \partial_{x_k}u
		-\sum_{i,j=1}^n \frac{h_{ij,k}g_{00}-h_{ij}g_{00,k}}{g_{00}^2}\partial_{x_i}\partial_{x_j} u \\
		-\sum_{i,j=1}^n \frac{h_{ij}}{g_{00}}\partial_{x_i}\partial_{x_j}\partial_{x_k} u +\sum_{j=1}^n \frac{a_{j,k}g_{00}-a_j g_{00,k}}{g_{00}^2}\partial_{x_j} u+\sum_{j=1}^n \frac{a_{j}}{g_{00}}\partial_{x_j}\partial_{x_k} u + \frac{a_{0,k}g_{00}-a_0g_{00,k}}{g_{00}^2}\partial_t u\\
		+\frac{a_0}{g_{00}}\partial_t \partial_{x_k}u
		+\frac{b_{,k}g_{00}-b g_{00,k}}{g_{00}^2}u+\frac{b}{g_{00}} \partial_{x_k}u =\frac{(Pu)_{,k}g_{00}-(Pu)g_{00,k}}{g_{00}^2}.
		\end{multline}
Multiplying by $g_{00}$ we may write this as
 \begin{multline}\label{equ_x_deriv}
 P\partial_{x_k} u + 2\sum_{j=1}^n \Big(g_{0j,k}-\frac{g_{0j}}{g_{00}}g_{00,k}\Big)\partial_{x_j}\partial_t u - \sum_{i,j=1}^n \Big(h_{ij,k}-\frac{h_{ij}}{g_{00}}g_{00,k}\Big)\partial_{x_i}\partial_{x_j} u+\sum_{j=1}^n\Big(a_{j,k}-\frac{a_{j}}{g_{00}}g_{00,k}\Big)\partial_{x_j} u\\
   =(Pu)_{,k}-\frac{Pu}{g_{00}}g_{00,k}- \Big(a_{0,k}-\frac{a_0}{g_{00}}g_{00,k}\Big)\partial_t u -\Big(b_{,k}-\frac{b}{g_{00}}g_{00,k}\Big)u
		\end{multline}
                where we have brought all terms which do not contain
                spatial derivatives $\partial_{x_k}u$ to the
                right-hand side.  Multiplying by
                $\dt \partial_{x_k}\overline u$ and summing over
                $1\leq k\leq n$, we get
\begin{multline}\label{div_form_Du}
 2\text{Re} \sum_{k=1}^n P(\partial_{x_k} u)\dt \partial_{x_k}\overline u =
  \dt \sum_{k=1}^n \mathcal T(\partial_{x_k} u))+ \sum_{k,j=1}^n \partial_{x_j} X_j(\partial_{x_k} u)) + \sum_{k=1}^nY(\partial_{x_k} u)) 
 + Z(\nabla u)) =\\
   \underbrace{2\sum_{k=1}^n\text{Re}\Bigg(\bigg(\big((Pu)_{,k}-\frac{Pu}{g_{00}}g_{00,k}\big)
 -   \Big(a_{0,k}-\frac{a_0}{g_{00}}g_{00,k}\Big)\partial_t u - \big(b_{,k}-\frac{b}{g_{00}}g_{00,k}\big)u\bigg)\dt \partial_{x_k}\overline u\Bigg)}_{=:F(u(t,x);t,x)}
\end{multline}
where the whole expression depends on $(t,x)$ and $X$ and $Y$ are
defined exactly as before. The additional term $Z$ collects all terms
from \eqref{equ_x_deriv} which are of lower order with respect to
$\partial_{x_k}u$ (that is, they contain spatial derivatives of order
2 at most). We have
\begin{multline*}
Z(\nabla u(t,x);t,x)=2\text{Re}\bigg(2\sum_{j,k=1}^n \Big(g_{0j,k}-\frac{g_{0j}}{g_{00}}g_{00,k}\Big)\partial_{x_j}\partial_t u \,\partial_{x_k}\partial_t \overline  u \\ 
- \sum_{i,j,k=1}^n \Big(h_{ij,k}-\frac{h_{ij}}{g_{00}}g_{00,k}\Big)\partial_{x_i}\partial_{x_j} u \,\partial_{x_k}\partial_t\overline  u 
+\sum_{j,k=1}^n\Big(a_{j,k}-\frac{a_{j}}{g_{00}}g_{00,k}\Big)\partial_{x_j} u \,\partial_{x_k}\partial_t \overline u\bigg)
\end{multline*}
and it is easy to see that $Z$ can be estimated as follows:
\begin{multline*}
|Z(\nabla u(t,x);t,x))|\leq 2n\sum_{i=1}^n |\dt \partial_{x_i} u(t,x)|^2 \max_{1\leq j,k \leq n} \norm{g_{0j,k}- g_{0j}\partial_{x_k}\ln g_{00} }{L^{\infty}(\Omega_T)} \\
+ n^2\sum_{i=1}^n \big(\norm{\partial_{x_i}\nabla u(t,x)}{}^2 +\norm{\dt \nabla u(t,x)}{}^2\big)
  \max_{1\leq j,k \leq n} \norm{h_{ij,k}-h_{ij}\partial_{x_k}\ln g_{00}}{L^{\infty}(\Omega_T)}  \\ 
 +  n\sum_{i=1}^n \big(|\partial_{x_i} u(t,x)|^2+\norm{\dt \nabla u(t,x)}{}^2\big) \max_{1\leq k\leq n}\norm{a_{i,k}- a_{i} \partial_{x_k}\ln g_{00}  }{L^{\infty}(\Omega_T)}  
   \end{multline*}
 and thus
 \begin{eqnarray*}
\int_{\R^n} |Z(\nabla u(t,x);t,x))| dV
 &\leq& C_{g',a'} \Big( \norm{\partial_t \nabla u(t,\cdot)}{L^2(\R^n)}^2+\norm{\nabla u(t,\cdot)}{L^2(\R^n)}^2+ \sum_{j=1}^n\norm{\partial_{x_j}\nabla u(t,\cdot)}{L^2(\R^n)}^2 \Big)\\
 &=& C_{g',a'} \Big( \norm{\partial_t \nabla u(t,\cdot)}{L^2(\R^n)}^2+\norm{\nabla u(t,\cdot)}{H^1(\R^n)}^2\Big)
 \end{eqnarray*}
 where $\norm{\nabla u(t,x)}{}^2:= \sum_{k=1}^n|\partial_{x_k} u(t,x)|^2$.
 Moreover, we have
 \begin{multline*}
 |F(u(t,x);t,x)|\leq  \sum_{k=1}^n\Big(  |(Pu)_{,k}(t,x)|^2 +  |Pu(t,x)|^2 \norm{\partial_{x_k}\ln g_{00}}{L^{\infty}(\Omega_T)} +2|\dt \partial_{x_k}u(t,x)|^2\Big) 
\\ + \max_{1\leq k\leq n} \norm{a_{0,k}- a_0 \partial_{x_k}\ln g_{00}}{L^{\infty}(\Omega_T)} \sum_{k=1}^n \big(|\dt u(t,x)|^2 + |\dt \partial_{x_k}u(t,x)|^2 \big)
 \\
+   \max_{1\leq k\leq n}\norm{b_{,k}-b\, \partial_{x_k}\ln g_{00}}{L^{\infty}(\Omega_T)}
\sum_{k=1}^n\big(|u(t,x)^2|+ |\dt \partial_{x_k}u(t,x)|^2\big)
  \end{multline*}
  and therefore
   \begin{multline*}
\int_{\R^n} |F (u(t,x);t,x))| dV
 \leq\\
  C_{g_{00}'}  \norm{Pu(t,\cdot)}{H^1(\R^n)}^2+C_{g_{00}',a_0',b'}\Big(\norm{ u(t,\cdot)}{L^2(\R^n)}^2 + \norm{ \partial_{t}u(t,\cdot)}{L^2(\R^n)}^2+\norm{ \dt \nabla u(t,\cdot)}{L^2(\R^n)}^2 \Big).
 \end{multline*}
 We may now proceed as in the proof for the basic energy estimate
 \eqref{energy_estimate_1}, following the steps from
 \eqref{div_form_x_int} to \eqref{energy_estimate_1}.  Upon
 integrating equation \eqref{div_form_Du} over the spatial domain
 $\R^n$, we get
\begin{multline*}
 \sum_{k=1}^n \,\int\limits_{\R^n}\mathcal T(\partial_{x_k} u;t,\cdot)dV\leq  \sum_{k=1}^n\,\int\limits_{\R^n}\mathcal T(\partial_{x_k} u;0,\cdot)dV + \sum_{k=1}^n\,\int\limits_0^t\int\limits_{\R^n}|Y(\partial_{x_k} u;t,\cdot)|dVds \\
 + \int\limits_0^t\int\limits_{\R^n} |Z(\nabla u;s,\cdot))|dVds +
 \int\limits_0^t\int\limits_{\R^n} |F(\nabla u;s,\cdot))|dVds
 \end{multline*}
 and by employing the estimates for $Y$, $Z$, and $F$, we arrive at an
 energy estimate for $\nabla u$ (recall that $\int_{\mathbb R^n}
 |Y(\partial_{x_k}u;t,\cdot)| dV\leq C_{g',a,b} \norm{\partial_{x_k}
   u(t,\cdot)}{H^1(\mathbb R^n)}^2$):
 \begin{multline*}
 C_{\text{min}}\Big(\norm{\nabla u(t,\cdot)}{L^2(\mathbb R^n)}^2+ \norm{\dt \nabla u(t,\cdot)}{L^2(\mathbb R^n)}^2 +\sum_{k=1}^n\norm{\partial_{x_k}\nabla u(t,\cdot)}{L^2(\mathbb R^n)}^2\Big)\\
\leq C_{\text{max}}\Big(\norm{\nabla u_0}{L^2(\mathbb R^n)}^2+ \norm{\nabla u_1}{L^2(\mathbb R^n)}^2 +\sum_{k=1}^n\norm{\partial_{x_k}\nabla u_0}{L^2(\mathbb R^n)}^2\Big) +  C_{g',a,b} \int\limits_0^t \sum_{k=1}^n   \norm{\partial_{x_k} u(s,\cdot)}{H^1(\mathbb R^n)}^2 ds\\
+ C_{g',a'}\int\limits_0^t  \Big( \norm{\partial_s \nabla u(s,\cdot)}{L^2(\R^n)}^2+\norm{\nabla u(s,\cdot)}{L^2(\R^n)}^2+\sum_{k=1}^n\norm{\partial_{x_k}\nabla u(t,\cdot)}{L^2(\R^n)}^2\Big) ds
\\
+C_{g_{00}'}\int\limits_0^t\norm{Pu(s,\cdot)}{H^1(\R^n)}^2 ds+C_{g_{00}',a_0',b'}\int\limits_0^t\Big(\norm{ u(s,\cdot)}{L^2(\R^n)}^2 + \norm{ \partial_{s}u(s,\cdot)}{L^2(\R^n)}^2+\norm{ \partial_{s} \nabla u(s,\cdot)}{L^2(\R^n)}^2 \Big)ds
  \end{multline*}
  which we may also write as
   \begin{multline*}
   C_{\text{min}}\Big(\norm{\nabla u(t,\cdot)}{H^1(\mathbb R^n)}^2+ \norm{\dt \nabla u(t,\cdot)}{L^2(\mathbb R^n)}^2\Big) \leq C_{\text{max}}\Big(\norm{\nabla u_0}{H^1(\mathbb R^n)}^2+ \norm{\nabla u_1}{L^2(\mathbb R^n)}^2 \Big) \\
   + C_{g',a',b'} \int\limits_0^t \Big(\norm{\nabla u(s,\cdot)}{H^1(\mathbb R^n)}^2+ \norm{\partial_s \nabla u(s,\cdot)}{L^2(\mathbb R^n)}^2\Big) ds \\
   + C_{g_{00}'}\int\limits_0^T\norm{Pu(s,\cdot)}{H^1(\R^n)}^2 ds+C_{g_{00}',a_0',b'}\int\limits_0^T\Big(\norm{ u(s,\cdot)}{L^2(\R^n)}^2 + \norm{ \partial_{s}u(s,\cdot)}{L^2(\R^n)}^2 \Big)ds
     \end{multline*}
     where $C_{g',a',b'}$ contains $L^{\infty}(\Omega_T)$-norms of at
     most first-order derivatives of the coefficients
     $g,a,b$. Applying Gronwall's inequality then results in an energy
     estimate for the spatial derivative vector $\nabla u$,
\begin{multline}\label{energy_estimate_spatial_derivatives}
 \sup_{0\leq t\leq T}\Big(\norm{\nabla u(t,\cdot)}{H^1(\mathbb R^n)}^2+ \norm{\dt \nabla u(t,\cdot)}{L^2(\mathbb R^n)}^2\Big) 
 \leq  
  C_{\text{min}}^{-1} \Bigg( C_{\text{max}}
 \Big(\norm{\nabla u_0}{H^1(\mathbb R^n)}^2+ \norm{\nabla u_1}{L^2(\mathbb R^n)}^2 \Big)\\
 + C_{g_{00}'}\int\limits_0^T\norm{Pu(s,\cdot)}{H^1(\R^n)}^2 ds+C_{g_{00}',a_0',b'}\int\limits_0^T\Big(\norm{ u(s,\cdot)}{L^2(\R^n)}^2 + \norm{ \partial_{s}u(s,\cdot)}{L^2(\R^n)}^2  \Big) ds \Bigg) e^{\widetilde\beta_1 T}
 \end{multline}
 where $\widetilde\beta_1:= C_{\text{min}}^{-1}C_{g',a',b'}$. Employing the basic energy estimate  \eqref{energy_estimate_1} to estimate $\norm{ u(s,\cdot)}{L^2(\R^n)}^2$  and $ \norm{ \partial_{s}u(s,\cdot)}{L^2(\R^n)}^2$ in terms of  the restrictions $\tilde u_0,\tilde u_1$ and $Pu$, we arrive at
\begin{equation*}
||u||_{C^0([0,T],H^{2}(\mathbb R^n))}+||u||_{C^1([0,T], H^{1}(\mathbb R^n))} 
\le e^{\beta T} \Big( \norm{\tilde{u}^0}{H^2(\mathbb R^n)}+ \norm{\tilde{u}^1}{H^1(\mathbb R^n)}+  ||Pu||_{L^{2}([0,T],H^{1}(\mathbb R^n))}\Big).  
\end{equation*}
\end{proof}

To prepare for the extension to the Colombeau solutions, we proceed by
discussing higher-order energy estimates. Applying the operator
$\partial_{x_l}$ to equation \eqref{divided_diff} and afterwards
multiplying by $g_{00}$ we essentially get a system of the form
$P(\nabla^2 u)=Q^{(2)}(\Sigma^2 u)+R^{(2)}(\Sigma^2 Pu)$, where $Q^2$
is a PDO of order 3, containing time derivatives of at most order 1,
and $R$ is a purely spatial PDO of order 2. Here we have used the
notation
 $$
\begin{array}{rclrcl}
\nab{1} u
&:=&
(\partial_{x_1}u,...,\partial_{x_n}u)^{\mathsf{T}},\\
 \nab{r+1} u
&:=& \nab{1} \nab{r} u,\\
 \sigm{r} u
&:=& 
(u,...,u)^{{\mathsf{T}}}_{n^r}
 \,.
\end{array}
$$
Thus the terms produced by $Q^{(2)}(\Sigma^2 u)$ are either
lower-order terms in $\nabla^2 u$ or terms with less than two spatial
derivatives of $u$. In the first case, they can be dealt with just as
the generic lower-order terms appearing in $P(\nabla^2 u)$. In the
second case, they can be interpreted as source terms on the right-hand
side, just as the term $R^{(2)}(\Sigma^2 Pu)$. More generally, we
obtain
  \begin{equation}\label{higher_order_equ}
  P(\nabla^{r} u)=Q^{(r)}(\Sigma^r u)+R^{(r)}(\Sigma^r Pu)
  \end{equation}
  where $Q^{(r)}$ is a PDO of order $r+1$, containing time derivatives
  of at most order 1, and $R^{(r)}$ is a purely spatial PDO of order
  $r$. Moreover, the coefficients of $Q^{(r)}$ and $R^{(r)}$ depend on
  spatial derivatives of the metric $g$ of at most order $r$. It is
  important to note that all principal coefficients of $Q^{(r)}$
  depend only on $g_{\mu\nu},a_{\mu},b$ and on derivatives
  $g_{\mu\nu,\rho}$ ($0\leq \rho \leq n$). In particular, they can be
  viewed as lower-order terms of an operator $P_r$, which has the same
  principal symbol as $P$. The energy estimate following from
  \eqref{higher_order_equ} is of the form
  \begin{multline}\label{higher_order_energy_estimate_spatial}
 \sup_{0\leq t\leq T}\Big(\norm{\nabla^r u(t,\cdot)}{H^1(\mathbb R^n)}^2+ \norm{\dt \nabla^r u(t,\cdot)}{L^2(\mathbb R^n)}^2\Big) 
 \leq  
  C_{\text{min}}^{-1} \Bigg( C_{\text{max}}
 \Big(\norm{\nabla^r u_0}{H^1(\mathbb R^n)}^2+ \norm{\nabla^r u_1}{L^2(\mathbb R^n)}^2 \Big)\\
 + C_{g_{00}'}\int\limits_0^T\norm{Pu(s,\cdot)}{H^r(\R^n)}^2 ds+C_{g^{(r)},a^{(r)},b^{(r)}}\int\limits_0^T \Big(
 \norm{u(s,\cdot)}{H^{r-1}(\R^n)}^2+\norm{\partial_s u(s,\cdot)}{H^{r-1}(\R^n)} 
 \Big)  ds \Bigg) e^{\beta_r T}
 \end{multline}
 where $\beta_r$ depends only on $g_{\mu\nu},a_{\mu},b$ and on
 derivatives $g_{\mu\nu,\rho}$. 

     Summing up, we obtain energy estimates for $ \nabla^r
 u$ and $\dt \nabla^r u$ for $r\in \mathbb N_0$. Equation
 \eqref{wave_eq} itself then immediately provides an estimate for
 $\dt^2 u$ as well. Similarly, equation \eqref{divided_diff} allows us
 to estimate $\dt^2 \partial_{x_k} u$ in terms of (mixed) derivatives
 of order $\leq$ 1 in the time variable. More generally, equation
 \eqref{higher_order_equ} directly provides estimates for $\dt^2
 \nabla^r u$ ($r\in \mathbb N$) in terms of (mixed) derivatives of
 order $\leq$ 1 in the time variable. To get estimates for mixed
 derivatives $\dt^m \nabla^r u$ of order $m \geq 3$ in the time
 variable, we may apply the operator $\dt^l$ to equation
 \eqref{higher_order_equ}, yielding
 \begin{equation}\label{mixed_derivatives_estimate}
  P(\nabla^{r}\dt^l u)=Q^{(r)}(\Sigma^r \dt^l u)+R^{(r)}(\Sigma^r\dt^l Pu)+ \widetilde{Q}^{(r,l)}(\Sigma^r  u)+\widetilde{R}^{(r,l)}(\Sigma^r Pu)
\end{equation}
where $\widetilde{Q}^{(r,l)}$ is a PDO of order $r+l$ with time
derivatives only up to order $l-1$ and $\widetilde{R}^{(r,l)}$ is a
PDO of order $r+l-1$ with time derivatives up to order $l-1$. The
first term in $P(\nabla^{r}\dt^l u)$ is $g_{00}\nabla^{r}\dt^{l+2}u$.
Thus, starting with $l=1$, we immediately obtain energy estimates for
all mixed derivatives of the form $\nabla^{r}\dt^3 u$. In the next
step, we can show energy estimates for $\nabla^{r}\dt^4u$, and proceed
iteratively to get energy estimates for all mixed derivatives
$\nabla^{r}\dt^m u$ as well. These estimates are direct consequences
of \eqref{mixed_derivatives_estimate}, once estimates for the terms
with time derivatives of lower order have been established (there is
no need for a Gronwall argument). However, a perhaps more elegant
viewpoint is that for any $r,l\in \mathbb N$, equation
\eqref{mixed_derivatives_estimate} also implies $H^1$-estimates for
$\nabla^{r}\dt^l u$ in terms of initial data and source terms via the
energy estimate \eqref{energy_estimate_1} for the operator $P$.
Again, it is important to note that the principal coefficients of
$\widetilde{Q}^{(r,l)}$ and $\widetilde{R}^{(r,l)}$ only depend on
$g_{\mu\nu},a_{j},b$ and on first derivatives
$g_{\mu\nu,\rho}$. Higher-order derivatives of the coefficients will
only enter in the lower order terms and thus do not show up in the
exponent after applying Gronwall's inequality. In the following lemma we refer to Colombeau theoretic notions which are summarised in Appendix A.

\begin{Lemma}\label{gensol}

  We consider the initial value problem \eqref{cauchy} with
  generalised functions as coefficients and data and assume that
\begin{itemize}
\item[(i)] the components of  $g$ as well as the lower-order coefficients $a_{\mu}$ and $b$ belong to $\mathcal G_{L^{\infty}}(\R^{n+1})$,
\item[(ii)] there exist  constants $\tau_{\text{min}}>0$ and  $\lambda_{\text{min}}>0$ such that  $ g_{00}^{\e}(t,x)\geq \tau_{\text{min}}(\log(1/\e))^{-1}$  and $\sum_{i,j}h_{ij}^{\e}(t,x)\overline{\xi_i}\xi_j\geq \lambda_{\text{min}}(\log(1/\e))^{-1}|\xi|^2$ for all $(t,x)\in \mathbb R^{n+1},\; \xi\in \mathbb C^{n}$  and for all $\e<\e_0$,  
\item[(iii)] all coefficients $g_{\mu\nu}$, $a_{\mu},b$ as well as all derivatives $g_{\mu\nu,\rho}$  are of $L^{\infty}$-log-type.
\end{itemize}
\noindent
Then, given  initial data $u_0,u_1\in \mathcal G_{L^{2}}(\R^{n}) $ and right-hand side $f\in \mathcal G_{L^{2}}(\Omega_T)$, the Cauchy problem  \eqref{cauchy} has a unique solution $u\in \mathcal G_{L^{2}}(\Omega_T)$.

\end{Lemma}

\begin{proof}
  We fix a symmetric representative of $g$ and representatives of all
  lower-order coefficients, initial data, and right-hand side.  As
  noted above, we have smooth solutions $u_{\e}\in
  C^{\infty}([0,T],H^{\infty}(\mathbb R^n))$ to the corresponding
  classical initial value problems for each $\e< \e_0$.  To show
  moderateness, we apply Lemma \ref{energy_estimates} and obtain
$$
||u_{\eps}||_{C^0([0,T],H^{2}(\mathbb R^n))}+||u_{\eps}||_{C^1([0,T], H^{1}(\mathbb R^n))} 
\le e^{\beta_{\eps} T} \Big( \norm{ {u_0}_{\eps}}{H^1(\mathbb R^n)}+ \norm{ {u_1}_{\eps}}{L^2(\mathbb R^n)}+  ||f_{\eps}||_{L^{2}([0,T],H^{1}(\mathbb R^n))}\Big).  
$$
where $(e^{\beta^{\e}})$ is moderate thanks to the log-type condition
on the coefficients and their first derivatives as well as the
positivity condition (ii). Thus the estimate shows that
$(\norm{u_{\e}}{H^1(\Omega_T)})$ and $(\norm{\nabla
  u_{\e}}{H^1(\Omega_T)})$ is moderate.  For $r\geq 2$, the
higher-order estimates \eqref{higher_order_energy_estimate_spatial}
for the spatial derivative vector $\nabla^ru_{\e}$ then imply that
$(\norm{\nabla^r u_{\e}}{H^1(\Omega_T)})$ is also moderate (for any
$r\in \mathbb N$), since all principal coefficients of the operators
on the right-hand-side of equation \eqref{higher_order_equ} only
depend on $(g_{\mu\nu}^{\e}),(a_{\mu}^{\e}),(b^{\e})$ and on first
derivatives $(g^{\e}_{\mu\nu,\rho})$, all of which are of
$L^{\infty}-$log-type.

An iterative application of the higher-order estimates for spatial
derivatives \eqref{higher_order_energy_estimate_spatial},
  \begin{multline*}
 \sup_{0\leq t\leq T}\Big(\norm{\nabla^r u_{\e}(t,\cdot)}{H^1(\mathbb R^n)}^2+ \norm{\dt \nabla^r u_{\e}(t,\cdot)}{L^2(\mathbb R^n)}^2\Big) 
 \leq  
   \Bigg( C^{\e}_{\text{max}}
 \Big(\norm{\nabla^r {u_0}_{\e}}{H^1(\mathbb R^n)}^2+ \norm{\nabla^r {u_1}_{\e}}{L^2(\mathbb R^n)}^2 \Big)\\
 + C_{g^{\e '}_{00}}\int\limits_0^T\norm{f_{\e}(s,\cdot)}{H^r(\R^n)}^2 ds+C_{g^{(r)}_{\e},a^{(r)}_{\e},b^{(r)}_{\e}}\int\limits_0^T \Big(
 \norm{u_{\e}(s,\cdot)}{H^{r-1}(\R^n)}^2+\norm{\partial_s u_{\e}(s,\cdot)}{H^{r-1}(\R^n)} 
 \Big)  ds \Bigg) \frac{e^{\beta_r^{\e} T}}{C_{\text{min}}^{\e}},   
 \end{multline*}
 starting with $r\geq 2$ then establishes moderateness of
 $\norm{\nabla^r u_{\e}}{H^1(\Omega_T)}$ for all $r\in \mathbb N_0$,
 since $1/C_{\text{min}}^{\e}$ is of logarithmic growth in $\e$
 (and thus moderate), $ C^{\e}_{\text{max}}$, $C_{g'}^{\e}$,
 $C_{g^{(r)}_{\e},a^{(r)}_{\e},b^{(r)}_{\e}}$ are of moderate growth
 and $\beta_r^{\e}$ is of $L^{\infty}$-log-type (where
 $\beta_r^{\e}$ depends on $C_{\text{min}}^{\e}$ and (first
 derivatives of) $g_{\e},a_{\e}$, and $b_{\e}$).  Finally, the
 corresponding energy estimates for mixed derivatives following from
 equation \eqref{mixed_derivatives_estimate} yield moderateness of
 $\norm{\nabla^r \dt^l u_{\e}}{H^1(\Omega_T)}$ as well. Note that only the principal coefficients will be exponentiated in the Gronwall inequality. Thus the 
 important observation in all these higher-order estimates is that the
 principal coefficients on both sides of equation
 \eqref{mixed_derivatives_estimate} only depend on
 $g_{\mu\nu}^{\e},a_{\mu}^{\e},b^{\e}$ and on derivatives
 $g_{\mu\nu,\rho}^{\e}$, all of which are of
 $L^{\infty}$-log-type.

 In total this implies that for all $r,l\in\mathbb N_0$ there exists
 $m\in\mathbb N_0$ such that $\norm{\nabla^r
   \dt^lu_{\e}}{L^2(\Omega_T)}^2=O(\e^{-m})$ as $\e\to 0$ and thus
 $[(u_{\e})_{\e>0}]\in \mathcal G_{L^2}(\Omega_T)$. To show uniqueness
 of the generalised solution in $\mathcal G_{L^2}(\Omega_T)$, we
 assume negligible initial data $(\widetilde {u_0}_{\e}),(\widetilde
 {u_1}_{\e})\in \mathcal N_{L^2}(\mathbb R^n)$ and right-hand side
 $(\widetilde f_{\e})\in \mathcal N_{L^2}(\Omega_T)$.  Then the same
 energy estimates we used for moderateness yield negligibility of
 the solution $[(\widetilde u_{\e})_{\e>0}]$.
 \end{proof}

In the next lemma we want to identify conditions on the coefficients
and data of a low-regularity Cauchy problem \eqref{cauchy} such that
the corresponding generalised Cauchy problem, obtained via
regularisation, has a unique solution $u\in \mathcal
G_{L^2}(\Omega_T)$ and, moreover, this solution admits a
distributional shadow.

\begin{Lemma} \label{weaksol}
		
  We consider the Cauchy problem \eqref{cauchy} where all coefficients
  belong to $W^{1,\infty}(\R^{n+1})$ and the bounds \eqref{taumin} and
  \eqref{lambdamin} are satisfied.  Then, given initial data
  $(u_0,u_1)\in H^2(\R^n)\times H^{1}(\R^{n}) $ and right-hand side
  $f\in L^2([0,T],H^1(\R^n))$, we consider the corresponding
  generalised Cauchy problem obtained via convolution regularisation
  of all coefficients and data. Then the unique generalised according
  to Lemma \ref{gensol} has a distributional shadow $\widetilde u\in
  C^0([0,T],H^2(\R^n)) \cap C^1([0,T],H^1(\R^n)) \cap H^2(\Omega_T)$
  which is also the unique weak solution.


\end{Lemma}

\begin{proof}
  We note that the regularised coefficients and data satisfy the
  assumptions in Lemma \ref{gensol} and we obtain a unique generalised
  solution. We aim at showing that any representative $(u_{\e})$ of
  the solution is a Cauchy net. To this end we apply the variant of
  the energy estimate in Lemma \ref{energy_estimates}, implicit in the
  proof (see also \eqref{energy_estimate_1}), with order of spatial
  Sobolev norms reduced by one (except for the initial data): Applying
  the operator $P_{\e}$ to the difference $u_{\e}-u_{\tilde \e}$, we
  obtain
\begin{multline}\label{Cauchy_estimate}
||u_{\e}-u_{\tilde \e}||_{C^0([0,T],H^{1}(\mathbb R^n))}+||u_{\e}-u_{\tilde \e}||_{C^1([0,T], L^{2}(\mathbb R^n))} 
\\
\le e^{\beta_{\e} T} \Big( \norm{{u_0}_{\e}-{u_0}_{\tilde\eps}}{H^1(\mathbb R^n)}+ \norm{{u_0}_{\e}-{u_0}_{\tilde \e}}{L^2(\mathbb R^n)}+  ||P(u_{\e}-u_{\tilde \e})||_{L^{2}(\Omega_T)}\Big),  
\end{multline}
 where  $\beta^{\e}$ is bounded uniformly in $\e$ thanks to the hypotheses on the non-smooth coefficients. We may write 
 \begin{eqnarray*}
 P_{\e}(u_{\e}-u_{\tilde\e}) &=& f_{\e}- P_{\e}u_{\tilde\e} = f_{\e}-f_{\tilde\e}- (P_{\e}-P_{\tilde\e})u_{\tilde\e}.
 \end{eqnarray*}
 Considering the regularity of the initial data and right-hand side,
 it is easy to see that $(u_{\e})$ is a Cauchy net, if for all
 $\eta>0$ there exists $\e_0$ such that
 $\norm{(P_{\e}-P_{\tilde\e})u_{\tilde\e}}{L^2(\Omega_T)}<\eta$ for
 all $\tilde\e<\e\leq \e_0$. We have
 \begin{multline*}
 \norm{(P_{\e}-P_{\tilde\e})u_{\tilde\e}}{L^2(\Omega_T)}
		\leq \\
		\norm{g^{\e}_{00}-g^{\tilde\e}_{00}}{L^{\infty}(\Omega_T)}
		\norm{\partial_t^2 u_{\tilde\e}}{L^2(\Omega_T)}
		+2\sum_{j=1}^n \norm{g^{\e}_{0j}-g^{\tilde \e}_{0j}}{L^{\infty}(\Omega_T)}\norm{\partial_{x_j}\partial_t u_{\tilde\e}}{L^2(\Omega_T)}
		\\+\sum_{i,j=1}^n \norm{h^{\e}_{ij}-h^{\tilde\e}_{ij}}{L^{\infty}(\Omega_T)}\norm{\partial_{x_i}\partial_{x_j} u_{\tilde\e}}{L^2(\Omega_T)} 
		+\sum_{j=1}^n \norm{a^{\e}_j-a_j^{\tilde\e}}{L^{\infty}(\Omega_T)}\norm{\partial_{x_j} u_{\tilde\e}}{L^2(\Omega_T)}\\
		 + \norm{a_0^{\e}-a_0^{\tilde\e}}{L^{\infty}(\Omega_T)}\norm{\partial_t u_{\tilde\e}}{L^2(\Omega_T)}
		+\norm{b^{\e}-b^{\tilde\e}}{L^{\infty}(\Omega_T)}
		\norm{u_{\tilde\e}}{L^2(\Omega_T)} 
		\leq\\
		 \sum_{\mu,\nu=0}^n \norm{g^{\e}_{\mu\nu}-g^{\tilde\e}_{\mu\nu}}{L^{\infty}(\Omega_T)}\norm{u_{\tilde \e}}{H^2(\Omega_T)}
		+ \sum_{\mu=0}^n 
		\norm{a^{\e}_{\mu}-a_{\mu}^{\tilde\e}}{L^{\infty}(\Omega_T)} \norm{u_{\tilde \e}}{H^1(\Omega_T)}\\
		 + \norm{b^{\e}-b^{\tilde\e}}{L^{\infty}(\Omega_T)} \norm{u_{\tilde \e}}{L^2(\Omega_T)}.
\end{multline*}
Since all coefficients belong to $W^{1,\infty}$ and thus converge in
the $L^{\infty}$-norm, it suffices to show that $\norm{u_{\tilde
    \e}}{H^2(\Omega_T)}$ is bounded uniformly in $\e$. First, observe
that the energy estimate \eqref{energy_estimate_1} implies that
$\norm{u_{\tilde \e}}{H^1(\Omega_T)}=O(1)$ as $\e\to 0$. The energy
estimate for $\nabla u$, inequality
\eqref{energy_estimate_spatial_derivatives},
\begin{multline*}
 \sup_{0\leq t\leq T}\Big(\norm{\nabla u_{\tilde \e}(t,\cdot)}{H^1(\mathbb R^n)}^2+ \norm{\dt \nabla u_{\tilde \e}(t,\cdot)}{L^2(\mathbb R^n)}^2\Big) 
 \leq  
  C_{\text{min}}^{-1} \Bigg( C_{\text{max}}
 \Big(\norm{\nabla u_{\tilde \e}^0}{H^1(\mathbb R^n)}^2+ \norm{\nabla u_{\tilde \e}^1}{L^2(\mathbb R^n)}^2 \Big)\\
 + C_{g_{00}'}\int\limits_0^T\norm{f_{\tilde \e}(s,\cdot)}{H^1(\R^n)}^2 ds+C_{g_{00}',a_0',b'}\int\limits_0^T\Big(\norm{ u_{\tilde \e}(s,\cdot)}{L^2(\R^n)}^2 + \norm{ \partial_{s}u_{\tilde \e}(s,\cdot)}{L^2(\R^n)}^2  \Big) ds \Bigg) e^{\widetilde\beta_1^{\tilde \e} T}
 \end{multline*}
 then implies that $\norm{u_{\tilde \e}}{H^2(\Omega_T)}=O(1)$ as
 $\e\to 0$ as well. Going back to \eqref{Cauchy_estimate}, we thus
 have shown that $(u_{\e})$ is indeed a Cauchy net in the norm
 $L^{\infty}([0,T],H^1(\R^n))$ and $(\dt u_{\e})$ is a Cauchy net in
 the norm $L^{\infty}([0,T],L^2(\R^n))$. Hence the unique generalised
 solution $u\in \mathcal G_{L^2}(\Omega_T)$ has a distributional limit
 $\widetilde u\in C ([0,T],H^1(\R^n))\cap
 C^1([0,T],L^2(\R^n))$. However, we can deduce even better regularity
 from the properties of the net $(u_{\e})$. For any $t\in [0,T]$ and
 any $\varphi\in\mathcal D(\R^n)$, we have
$$ |\langle \partial_{x_j}\partial_{x_k} \widetilde u(t,\cdot),\varphi\rangle |=
|\lim_{\e\rightarrow 0}\langle \partial_{x_j}\partial_{x_k}
u_{\e}(t,\cdot),\varphi\rangle|\leq \lim_{\e\rightarrow
  0}|\langle \partial_{x_j}\partial_{x_k}
u_{\e}(t,\cdot),\varphi\rangle|\leq C \norm{\varphi}{L^2(\mathbb
  R^n)}$$ 
since $\norm{u_{\e}(t,\cdot)}{H^2(\R^n)}=O(1)$ as $\e\to 0$
and we already have uniform convergence of $(u_{\e}(t,\cdot))$ in
$H^1(\R^n)$ by the Cauchy net estimate and hence as a distribution.
It follows that $\partial_{x_j}\partial_{x_k}u_0(t,\cdot )\in
L^2(\R^n)$ and therefore $u_0\in C^0([0,T],H^2(\R^n))$. Moreover we
have
$$ |\langle \partial_{x_j}\dt \widetilde u(t,\cdot),\varphi\rangle |=
|\lim_{\e\rightarrow 0}\langle \partial_{x_j}\dt u_{\e}(t,\cdot),\varphi\rangle|\leq
 \lim_{\e\rightarrow 0}|\langle \partial_{x_j}\dt u_{\e}(t,\cdot),\varphi\rangle|\leq
  \widetilde C \norm{\varphi}{L^2(\mathbb R^n)}$$
and thus $\widetilde u\in C^1([0,T],H^1(\R^n))$. Similarly we can
  show that $\partial_{x_{\mu}}\partial_{x_k} \widetilde u \in
  L^2(\Omega_T)$ for $0\leq\mu\leq n$ and $1\leq k\leq n$ In addition,
  multiplying the equation $P_{\e}u_{\e}=f_{\e}$ by
  $(g_{00}^{\e})^{-1}$, it is easy to see that $\norm{\dt^2
    u_{\e}}{L^2(\Omega_T)}=O(1)$ as $\e\to 0$, implying that $\dt^2
  \widetilde u\in L^2(\Omega_T)$ as well. Summing up, we have obtained
  a distributional shadow $u_0$ of the generalised solution with
  \begin{equation*}
  \widetilde u\in  C^0([0,T],H^2(\R^n)) \cap  C^1([0,T],H^1(\R^n)) \cap H^2(\Omega_T).
  \end{equation*}

  In the last part of the proof we show that the distributional shadow
  $\widetilde u$ of the generalised solution is the unique weak
  solution to the Cauchy problem. The proof follows the line of
  arguments in the proof of \cite[Corollary 4.6]{HS:12}.
  First note that both $\widetilde u$ and $\dt \widetilde u$ are
  continuous and thus, by construction of $\widetilde u$, the initial
  conditions are satisfied. The Cauchy net estimate
  \eqref{Cauchy_estimate} implies that $u_{\e}\to\widetilde u$ as
  $\e\to 0$ in the norm $H^1(\Omega_T)$. Our aim is to prove that
  $P\,\widetilde u = f$ in a suitable weak sense. Since the
  coefficients belong to $ W^{1,\infty}(\R^{n+1})$, we immediately get
  $L^2$-convergence of all first-order terms and $H^1$-convergence of
  all zero-order terms:
	$$
	 \sum_{\mu=0}^n a^{\e}_{\mu}\partial_{x_{\mu}} u_{\e} \overset{\e\to 0}{\longrightarrow}
	 \sum_{\mu=0}^n a_{\mu}\partial_{x_{\mu}} \widetilde u \quad \text{in}\:\, L^2(\Omega_T)\qquad \text{and}\quad b^{\e}u_{\e} \overset{\e\to 0}{\longrightarrow} b\,\widetilde u \quad \text{in}\:\,   H^1(\Omega_T)
$$
We claim that 
\begin{multline*}
P^{(2)}_{\e}u_{\e}:=g_{00}^{\e}\partial_t^2 u_{\e}+2\sum_{j=1}^n g_{0j}^{\e}\partial_{x_j}\partial_t u_{\e}-\sum_{i,j=1}^n h^{\e}_{ij}\partial_{x_i}\partial_{x_j} u_{\e}\\
\overset{\e\to 0}{\longrightarrow} g_{00}\partial_t^2 \widetilde u+2\sum_{j=1}^n g_{0j}\partial_{x_j}\partial_t \widetilde u-\sum_{i,j=1}^n h_{ij}\partial_{x_i}\partial_{x_j} \widetilde u=:  P^{(2)}\widetilde u
\end{multline*}
weakly* in $L^2(\Omega_T)$. Since there exists $C>0$ such that
$\norm{u_{\e}}{H^2(\Omega_T)}\leq C$ for all $\e<\e_0$, we know that
the set $\{ u_{1/n}|\,n\in\mathbb N \}$ is bounded in
$H^2(\Omega_T)$. By the weak compactness theorem, this implies that
there exists a weakly$\ast$ convergent subsequence
$(u_{1/n_k})_{k\in\mathbb N}$ with limit $\widetilde v\in
H^2(\Omega_T)$ (cf.\ \cite[Theorem 6.64]{RR:03}) and indeed
$\widetilde v = \widetilde u$ since we already know from the Cauchy
net estimate that $u_{1/n_k}\to \widetilde u$ in $H^1(\Omega_T)$ as
$k\to \infty$. Therefore $\partial_{\mu}\partial_{\nu}u_{1/n_k}$
converges weakly* to $\partial_{\mu}\partial_{\nu}\widetilde u$ in
$L^2(\Omega_T)$ and we have for any $\varphi \in L^2(\Omega_T)$:
 \begin{eqnarray*}
|\langle P^{(2)}_{1/n_k}u_{1/n_k},\varphi\rangle- \langle  P^{(2)}\widetilde u,\varphi\rangle| &=&
|\langle (P^{(2)}_{1/n_k}- \widetilde P^{(2)})u_{1/n_k},\varphi\rangle- \langle  P^{(2)}(\widetilde u-   u_{1/n_k}),\varphi\rangle|\\
&\leq & |\langle (P^{(2)}_{1/n_k}-  P^{(2)}) u_{1/n_k},\varphi\rangle| + |\langle  P^{(2)}(\widetilde u-   u_{1/n_k}),\varphi\rangle| \to 0
\end{eqnarray*}
as $k \to \infty$. To see this, observe that the first term goes to
zero because the coefficients of $P_{1/n_k}$ converge to those of
$\widetilde P$ in $L^{\infty}(\Omega_T)$ as $k\to \infty$ and
$(u_{1/n_k})_{k\in \mathbb N}$ is bounded in $H^2(\Omega_T)$; the
second term vanishes as well in the limit $k\to \infty$ since
$u_{1/n_k}$ converges to $\widetilde u$ in $H^2(\Omega_T)$ as $k\to
\infty$. We provide the explicit calculation for the $g_{00}$-term
(the others can be treated similarly):
  \begin{multline*}
|\langle g_{00}^{1/n_k}\dt^2 u_{1/n_k},\varphi\rangle- \langle  g_{00}\,\dt^2\widetilde u,\varphi\rangle|=
|\langle (g_{00}^{1/n_k}- g_{00})\dt^2 u_{1/n_k},\varphi\rangle- \langle g_{00}(\dt^2\widetilde u-\dt^2u_{1/n_k}),\varphi\rangle|\\
\leq  \norm{g_{00}^{1/n_k}- g_{00}}{L^{\infty}(\Omega_T)}\norm{u_{1/n_k}}{L^2(\Omega_T)} \norm{\varphi}{L^2(\Omega_T)}
\\+ 
\norm{ g_{00}}{L^{\infty}(\Omega_T)}\norm{\dt^2\widetilde u - \dt^2 u_{1/n_k}}{L^2(\Omega_T)}\norm{\varphi}{L^2(\Omega_T)} \to 0 \quad (k\to \infty).
\end{multline*}
This shows that for any $\varphi\in L^2(\Omega_T)$,
$$ \langle  P \, \widetilde u, \varphi\rangle = \lim_{\e\to 0} \langle  P_{\e} \,  u_{\e}, \varphi\rangle  =  \lim_{\e\to 0} \langle f_{\e}, \varphi\rangle  =\langle  f,\varphi \rangle
$$
and thus $\widetilde u$ is indeed a weak solution of the initial value problem. \\ 

To show uniqueness of the weak solution, we suppose that there exists
another solution $\widetilde w \in H^2(\Omega_T)\cap
C([0,T],H^1(\R^n))\cap C^1([0,T],L^2(\R^n))$ such that $P\,\widetilde
w= f$ and $(\widetilde w,\dt \widetilde w)|_{t=0}=(u_0,u_1)$. We may
regularise this solution so that $w_{\e}\rightarrow \widetilde w$ as
$\e\to 0$ in $H^2(\Omega_T))$ and $(w_{\e},\dt w_{\e})|_{t=0}\to
(u_0,u_1)$ in $H^1(\R^n)\times L^2(\R^n)$. The following estimate then
shows that $P_{\e}w_{\e}\to P\,\widetilde w=f$ in $L^2(\Omega_T)$ as
$\e\to 0$:
      \begin{multline*}
 \norm{P_{\e}w_{\e}- P \,\widetilde w}{L^2(\Omega_T)} \\= \norm{P_{\e}w_{\e}-P_{\e}\,\widetilde  w + P_{\e}\,\widetilde  w
  - P\, \widetilde w}{L^2(\Omega_T)} 
  \leq  \norm{P_{\e}(w_{\e}-\widetilde w)}{L^2(\Omega_T)} +
    \norm{(P_{\e}- P)\widetilde w}{L^2(\Omega_T)} 
 \\
  \leq C \norm{w_{\e}-\widetilde w}{H^2(\Omega_T)}+\sum_{\mu,\nu=0}^n\norm{g_{\mu\nu}^{\e}- g_{\mu\nu}}{L^{\infty}(\Omega_T)} \norm{\widetilde w}{H^2(\Omega_T)}
\\  +
  \sum_{\mu=0}^n \norm{a_{\mu}^{\e}- a_{\mu}}{L^{\infty}(\Omega_T)} \norm{\widetilde w}{H^1(\Omega_T)}  + \norm{b^{\e}-b}{L^{\infty}(\Omega_T)}\norm{\widetilde w}{L^2(\Omega_T)} 
  \to 0\quad (\e\to 0).
  \end{multline*}
  Here we have only used the $H^2(\Omega_T)$-convergence of $w_{\e}$
  to $\widetilde w$ as $\e\to 0$.  Denoting by $(u_{\e})_{\e}$ a
  representative of generalised solution and applying the basic energy
  estimate \eqref{energy_estimate_1} to the difference $u_{\e}-w_{\e}$
  then yields
  \begin{multline*}
\frac{1}{C} \sup_{0\leq t \leq T}\Big(\norm{u_{\e}(t,\cdot)-w_{\e}(t,\cdot)}{H^1(\R^n)}^2+ \norm{\dt u_{\e}(t,\cdot) -\dt w_{\e}(t,\cdot)}{L^2(\R^n)}^2 \Big)
 \leq  \\ 
 \norm{(u_{\e}-w_{\e})|_{t=0}}{L^2(\mathbb R^n)}^2+ \norm{(\dt u_{\e}-\dt w_{\e} )|_{t=0}}{L^2(\mathbb R^n)}^2 +\norm{\nabla (u_{\e}-w_{\e})|_{t=0}}{L^2(\mathbb R^n)}^2 + \norm{f_{\e}-P_{\e} \, w_{\e}}{L^2(\Omega_{T})}^2. \end{multline*}
Letting $\e\rightarrow 0$ we obtain
 $\widetilde u=\widetilde w$ in $  C([0,T],H^1(\mathbb R^n))\cap C^1([0,T],L^2(\mathbb R^n))$. 
\end{proof}

\begin{rem}
  The required conditions for the existence of a distributional shadow
  (and weak solution) are weaker than those that would be required in
  a similar result based on transforming the equation \eqref{wave_eq}
  into a first-order system as in \cite{HHSS:13}, since the
  lower-order coefficients of this system would contain derivatives of
  the principal coefficients of equation \eqref{wave_eq} and therefore
  $g_{\mu\nu}\in W^{2,\infty}(\R^{n+1})$ would be necessary (instead
  of $ g_{\mu\nu}\in W^{1,\infty}(\R^{n+1})$).  However, for a wave
  equation derived from the Laplace-Beltrami operator of a Lorentzian
  metric $g$,
 \begin{equation*}
 \Box_gu =\sum_{\mu,\nu=0}^n |\det g|^{-\frac{1}{2}} \partial_{\mu}\big( |\det g|^{\frac{1}{2}} g^{\mu\nu}\partial_{\nu}u\big)=g^{\mu\nu}(\partial_{\mu}\partial_{\nu}u+\Gamma^{\rho}_{\mu\nu}\partial_{\rho}u)=f,
\end{equation*}
where 
$ \Gamma^{\rho}_{\mu\nu}=\frac{1}{2}g^{\rho\sigma}(g_{\sigma\mu,\nu}+g_{\sigma\nu,\mu}-g_{\mu\nu,\sigma})
$, the lower-order coefficients  contain derivatives of the metric and thus the metric  has to be $W^{2,\infty}$ ($C^{1,1}$) anyway in order to obtain a distributional shadow of the generalised solution.

\end{rem}

\begin{Lemma}\label{causalsupport_Rn}
  Consider the Cauchy problem \eqref{cauchy} where all coefficients
  belong to $W^{1,\infty}(\R^{n+1})$ and the bounds \eqref{taumin} and
  \eqref{lambdamin} are satisfied.  Then, given vanishing initial data and right-hand side
  $f\in L^2([0,T]),H^1(\R^n))$ with $\supp(f) \subset J^+(K)$ for some compact set $K$, the unique distributional weak
  solution $\bar u^+$ of the advanced problem given by Lemma
  \ref{weaksol} satisfies the causal support condition
$$
{\supp}(\bar u^+) \subset J^+({\supp}(f))
$$
\end{Lemma}

\begin{proof}

Let  $u^+_\eps$ be the unique solution of the corresponding advanced problem for  the regularised metric $g_\eps$ where the right-hand side $f$ is not necessarily  smooth. Note that upon taking $K$ slightly larger, we may assume
that $\supp(f) \subset J^+_\varepsilon(K)$ for $\varepsilon$ sufficiently small.

By Lemmas \ref{energy_estimates}, \ref{gensol}  we can obtain $u_{\varepsilon}^+$ as a limit $\alpha\to 0$ of $u_{\varepsilon,\alpha}^+$ where $(u_{\varepsilon,\alpha}^+)_{\alpha > 0}$ is a Colombeau representative of the unique generalised solution with fixed smooth $g_{\varepsilon}$-coefficients and right-hand side being the class of a convolution regularisation $(f_{\alpha})_{\alpha > 0}$. By the smooth theory we have for every $\varepsilon$ and every $\alpha$ that 
\begin{equation}\label{suppalpha}
\text{supp}(u_{\varepsilon,\alpha}^+)\subseteq J_{\varepsilon}^+(\text{supp}(f_{\alpha})).
\end{equation}
 At fixed  $\varepsilon$ and as $\alpha\to 0$ we have
in terms of monotonically decreasing sets $\bigcap_{\alpha > 0}\text{supp}(f_{\alpha})=\supp(f)$, i.e.\ $\supp(f_{\alpha})\searrow \supp(f)$. By closedness of the causal relation \cite[Lemma A.5.5]{bgp} we obtain that 
\begin{equation}\label{suppalphalim}
J_{\varepsilon}^+(\text{supp}(f_{\alpha})) \searrow J_{\varepsilon}^+(\text{supp}(f)).
\end{equation}

 Let $\psi$ be a test function such that $\supp(\psi)\cap J_{\varepsilon}^+(\supp(f))=\emptyset$. By \eqref{suppalpha} and \eqref{suppalphalim} there exists some $\alpha_0>0$ such that $\supp(\psi)\cap J_{\varepsilon}^+(\text{supp}(f_{\alpha})) =\emptyset $ for all $\alpha<\alpha_0$. Therefore $\langle u_{\varepsilon,\alpha}^+,\psi \rangle=0$ for every $\alpha<\alpha_0$; taking the limit $\alpha \to 0$ we obtain $\langle u_{\varepsilon}^+,\psi \rangle=0$ and therefore 
 \begin{equation*}
 \supp(u_{\varepsilon}^+)\subseteq J_{\varepsilon}^+(\text{supp}(f)). 
 \end{equation*}
  Now let $\psi$ have support disjoint from $J^+(\supp(f))$. Then, by the causal properties of the Chru\'{s}ciel-Grant regularisation, we have that $\supp(\psi)$ is also  disjoint from $J_{\varepsilon}^+(\supp(f))$ for every $\varepsilon >0$.

 Therefore, 
  $\langle u_{\varepsilon}^+, \psi \rangle=0 $ for every such $\psi$  and for every $\varepsilon >0$, showing that $\supp(u^+)\subseteq J_{\varepsilon}^+(\supp(f))$ for every $\varepsilon >0$, hence $\supp(\bar u^+)\subset J^{+}(\supp(f))$.
\end{proof}

\begin{rem}\label{initial}  More generally for non-zero initial data one has 
$$
\supp(\bar u^+) \subset J^+(\supp(f) \cup \supp(u_0) \cup \supp(u_1)).
$$
The proof of this follows from the above by recasting the smoothed
version of the problem as an equivalent inhomogeneous problem with
non-zero initial data.
\end{rem}

\begin{Theorem}\label{Rn}
  Consider the Cauchy problem \eqref{cauchy} where all coefficients
  belong to $W^{1,\infty}(\R^{n+1})$ and the bounds \eqref{taumin} and
  \eqref{lambdamin} are satisfied.  Then, given initial data
  $(u_0,u_1) \in H^2(\Real^n) \times H^1(\Real^n)$ and $f\in
  L^2([0,T]),H^1(\R^n)$, there exists a unique distributional weak
  solution $u\in C([0,T],H^2(\R^n))$ $\cap C^1([0,T],H^1(\R^n))\cap
  H^2(\Omega_T)$. Furthermore $u$ satisfies the causal support
  condition
$$
\supp(u) \subset J(\supp(f) \cup \supp(u_0) \cup \supp(u_1)).
$$
\end{Theorem}

\begin{proof}
  The existence and uniqueness of the weak solution follow from Lemma
  \ref{weaksol} and the causal support condition follows from Lemma
  \ref{causalsupport_Rn} together with Remark \ref{initial} applied to
  both the past and future.
\end{proof}


\section{The Cauchy problem for $C^{1,1}$ globally hyperbolic
  spacetimes}\label{GH}

\subsection{Causality results for $C^{1,1,}$ spacetimes}

In this paper we will be considering solutions of the wave equation on
orientable spacetimes $(M,g)$ endowed with a $C^{1,1}$ metric. Note
that although the metric is only $C^{1,1}$ we will always assume that the
manifold has a smooth structure. The concept of global
hyperbolicity (for smooth metrics) was introduced by Leray
\cite{Leray} as a condition to ensure the existence of unique
solutions to hyperbolic equations and in particular the Cauchy problem
for the wave equation is well-posed for smooth globally hyperbolic
spacetimes \cite[Theorem 3.2.11 p. 84ff]{bgp}. For our situation it is
therefore natural to consider globally hyperbolic spacetimes with
$C^{1,1}$ metrics. Global hyperbolicity is the strongest of the
conditions in the causal hierarchy of spacetimes \cite{ladder} and
recently there has been considerable interest in looking at the causal
properties of low-regularity spacetimes \cite{conespace}
\cite{KSSV}. It was shown explicitly by Chrus\'ciel \cite{crusciel}
that essentially all of causality theory for smooth spacetimes goes
through to the $C^2$ case. However in the proofs of these results an
important role is played by the existence of totally normal (convex)
neighbourhoods and the Gauss Lemma whose existence is not automatic in
the $C^{1,1}$ case. However it was shown in \cite{KSSV}, \cite{KSS}
\cite{mingnormal} that such neighbourhoods do exist and the
exponential map gives a local lipeomorphism. This enables essentially
the whole of the results of standard smooth causality theory to go
through to the $C^{1,1}$ case (see \cite{KSSV} and \cite{mingnormal}
for details).

For spacetimes with a $C^2$ metric there are four equivalent notions
of global hyperbolicity (see for example \cite[Section 3.11
p. 340ff.]{ladder}). These are:
\begin{enumerate}
\item compactness of the causal diamonds and causality\footnote{As shown by Bernal and S\'anchez the
    requirement of strong causality in the classical definition of
    global hyperbolicity can be weakened to only require causality
    \cite{bernalsanchez}.},
\item compactness of the space of causal curves connecting two points
 and causality \cite{Leray}, 
\item existence of a Cauchy hypersurface,
\item the metric splitting of the spacetime.
\end{enumerate}
For $C^{1,1}$ spacetimes we will adopt the first definition. However
for non-totally imprisoned \cite{nontotallyimp} $C^0$ spacetimes (and
hence in particular for globally hyperbolic $C^{1,1}$ spacetimes) these
four definitions remain equivalent \cite{saemann}. See also
\cite{conespace} Theorem 2.45 for a more general notion formulated in
terms of closed cone structures.

In our constructions below we will make use of time functions and
temporal functions. A \emph{time function} is a function that is
strictly increasing along every causal curve while a \emph{temporal
  function} has the additional property that its gradient is
everywhere past-directed and timelike. It is shown by Minguzzi
\cite[Theorem 2.30]{conespace} (see also \cite{FathiSiconolfi}) that
for a stably causal closed cone structure (and hence in particular for
a $C^{1,1}$ globally hyperbolic spacetime) there exists a
\emph{smooth} temporal function $t:M\rightarrow
\mathbb{R}$. Furthermore every globally hyperbolic closed cone
structure is the domain of dependence of a stable Cauchy surface (see
definition below) $\Sigma$, so that $M=D(\Sigma)$ and that $M$ is the
topological product $\Real \times \Sigma$ where the first projection is $t$
and the level surfaces $\Sigma_\tau=\{x \in M: t(x)=\tau\}$ are
\emph{diffeomorphic} to $\Sigma$ \cite[Theorem 2.42]{conespace}.  So
that although the metric is only $C^{1,1}$ the topological splitting
remains smooth.

In the case of a smooth metric Bernal and Sanchez
  \cite{sanchezcauchy} show that given a smooth spacelike Cauchy
  hypersurface $\Sigma$ there exists a smooth temporal function $t$
  such that $\Sigma=t^{-1}(0)$. However in the case of a non-smooth
  metric the temporal function they construct will not be smooth. To
  generalise the results of \cite{sanchezcauchy} to the non-smooth
  case we need the concept of a \emph{stable Cauchy hypersurface}
  introduced by Minguzzi in \cite{conespace}. These are Cauchy
  hypersurfaces which are also Cauchy hypersurfaces for some metric
  $g' \succ g$ with strictly wider lightcones than $g$. 

  Bernard and Suhr \cite[Corollary 2.4]{BernardSuhr} show that a
  smooth spacelike Cauchy hypersurface is a stable Cauchy hypersurface
  and that furthermore one can construct a smooth temporal function
  such that $\Sigma=t^{-1}(0)$ \cite[Theorem 1]{BernardSuhr}. A full
  discussion of this issue is given in the paper by Minguzzi
  \cite{MinguzziAddenda}. The approach in \cite{MinguzziAddenda} is
  complementary to that in \cite{BernardSuhr} and consists of using
  topological arguments to show that the causal cones can be widened
  while preserving the Cauchy property of the hypersurface. One may
  then use the methods of Bernal Sanchez \cite{sanchezcauchy} to
  construct a smooth time function with
  $\Sigma=t^{-1}(0)$ which as shown in \cite{conespace} is a smooth
  temporal function for the original spacetime $(M,g)$.  See
  \cite[Theorem 2.22]{MinguzziAddenda} for details.  Indeed given two
  smooth spacelike Cauchy hypersurfaces $\Sigma_0$ and $\Sigma_1$ with
  $\Sigma_1 \subset J^+(\Sigma_0) \setminus \Sigma_0$ one can find a smooth
  temporal function $t$ that interpolates between them so that
  $\Sigma_0 \subset t^{-1}(0)$ and $\Sigma_1 \subset t^{-1}(1)$
  \cite[Theorem 2.23]{MinguzziAddenda}.

\subsection{Existence and Uniqueness}

In this section we extend the results of Section \ref{RN} to a
globally hyperbolic $C^{1,1}$ spacetime $(M,g)$. The main result is
Theorem \ref{globalexist} which establishes the existence and
uniqueness of $H^2_{\text{loc}}(M)$ solutions to the wave equation for a
globally hyperbolic $C^{1,1}$ spacetime. We start by obtaining an
energy inequality which we use to establish uniqueness and the causal
support properties of solutions to the wave equation.

\begin{Lemma}\label{HE7.4.4}{\rm (Energy inequality) \cite[Lemma 7.4.4]{HandE}}\\
  Let $(M,g)$ be an $(n+1)$-dimensional $C^{1,1}$ globally hyperbolic
  spacetime with $\Sigma$ a smooth spacelike $n$-dimensional Cauchy surface
  and $t$ a smooth temporal function such that $\Sigma=t^{-1}(0)$. Let $U \subset M$ be an open
  set with compact closure and let $U^+:=U \cap J^+(\Sigma)$ be such that
  $\partial U \cap {\bar U}^+$ is achronal. Then if $u$ is a (weak)
  $H_{\text{loc}}^2(M)$ solution of $\square_g u=f$ where $f \in
  L^2_{\text{loc}}(M)$ then
\begin{equation}\label{enegyin}
||u||_{\tilde H^1(\Sigma_\tau\cap U^+)} \leq K \left(||u||_{\tilde H^1(\Sigma_0\cap U^+)} 
+||f||_{L^2(U_\tau)}\right)
\end{equation} 
where $U_\tau=\left\{ q \in U: 0 \leq t(q) \leq \tau\right\}$. 
\end{Lemma}

\begin{center}
\includegraphics[scale=0.7]{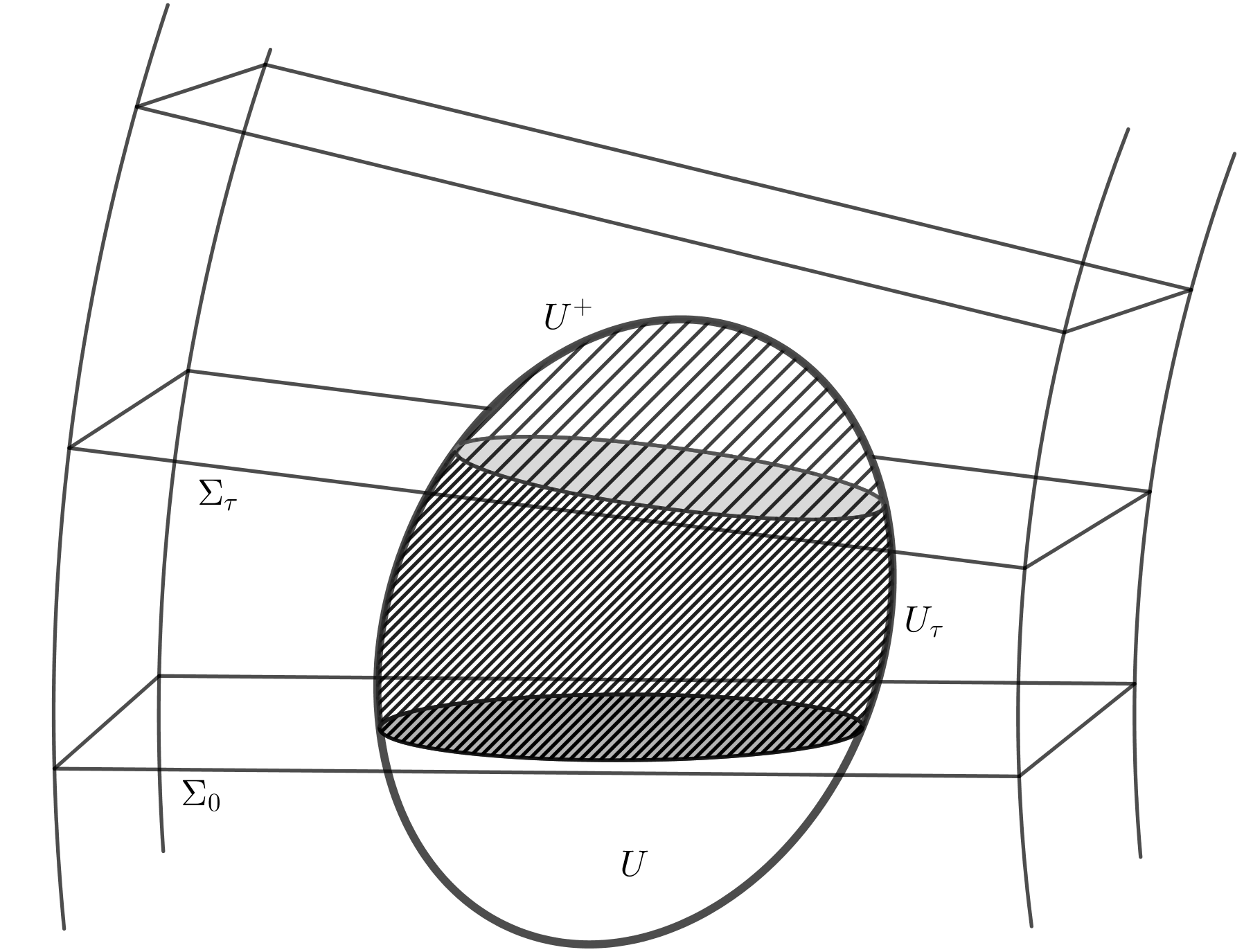}
\end{center}

\begin{proof} 
  To establish the energy inequality we follow Hawking and Ellis by
  applying the divergence theorem to an enhanced energy-momentum
  tensor \cite[Lemma 7.4.4]{HandE}. Let
\begin{equation*}
\overset{(1)}{T}_{\alpha\beta}:=\nabla_{\alpha}u\nabla_{\beta}u-\frac{1}{2}\left(g^{\rho\sigma}\nabla_{\rho}u\nabla_{\sigma}u\right)g_{\alpha\beta}.
\end{equation*}
be the energy momentum tensor of the scalar field. Then
$\overset{(1)}{T}_{\alpha\beta}$ has vanishing divergence and
satisfies the dominant energy condition. We now follow \cite{HandE}
and modify this by adding on the term
\begin{equation*}
\overset{(0)}{T}_{\alpha\beta}:=-\frac{1}{2}g_{\alpha\beta}u^2
\end{equation*}
to obtain  
\begin{equation*}
S_{\alpha\beta}=\overset{(0)}{T}_{\alpha\beta}+\overset{(1)}{T}_{\alpha\beta}
\end{equation*}
which still satisfies the dominant energy condition. Let
$\xi_\alpha=\nabla_\alpha t$ then we obtain the required inequality by
applying the divergence theorem to $S^{\alpha\beta}\xi_\alpha$ over
the region $U_\tau=\left\{ q \in U: 0 \leq t(q) \leq \tau\right\}$.
In order to do this we require that ${\rm
  div}(S^{\alpha\beta}\xi_\alpha)$ should be integrable with respect
the volume form $\nu_g$, and this is guaranteed by the compactness of
$\bar U$ and the fact that our solution is in $H^2_{\text{loc}}(M)$. In fact
it is enough that the weak solutions have two derivatives in
$L^{2}_{\text{loc}}(M,g)$ if the metric and the timelike vector field are in the
space $C^{0,1}$ (see \cite{crusdiv}). The boundary of $U_{\tau}$
consists of three parts; the level surface $\Sigma_\tau\cap {\bar U}^+$,
the level surface $\Sigma_0\cap {\bar U}^+$ and the remainder which we
denote ${\mathcal H}$. Because of the dominant energy condition and
the fact that $\partial U \cap {\bar U}^+$ is achronal, the
contribution to the surface integral from ${\mathcal H}$ is
positive. We therefore obtain the following inequality:
\begin{equation}\label{div}
\int_{\Sigma_\tau \cap {\bar U}^+}S^{\alpha\beta}\xi_\alpha\xi_\beta\mu_{\tau}-\int_{\Sigma_0 \cap {\bar U}^+}S^{\alpha\beta}\xi_\alpha\xi_\beta\mu_{0} \leq \int_{U_\tau} \nabla_\alpha(S^{\alpha\beta}\xi_\beta)\nu_g
\end{equation}
where $\nu_g$ is the volume form on $U_\tau$ given by $g$, and
$\mu_{\tau}$ is the volume form induced by $g$ on the $\Sigma_\tau$. 

We now define an energy type integral 
\begin{equation*}
E(\tau)=\int_{\Sigma_\tau \cap {\bar U}^+}S^{\alpha\beta}\xi_\alpha\xi_\beta\mu_{\tau}.
\end{equation*}
Then on $\bar U$ this is equivalent \cite{wilson} to the restricted Sobolev norm
(\ref{restricted})
\begin{equation*}
C_1||u||_{\tilde H^1(\Sigma_\tau \cap {\bar U}^+)} \leq E(\tau) \leq C_2||u||_{\tilde H^1(\Sigma_\tau \cap {\bar U}^+)}.
\end{equation*}
Note that since the solution $u$ is in $H^2_{\text{loc}}(M)$ we have well-defined traces in ${\tilde H}^1(\Sigma)$. 
In terms of the energy norm we may write (\ref{div}) in the form
\begin{equation*}
E(\tau) \leq E(0) + \int_{U_\tau}((\nabla_{\alpha}S^{\alpha\beta})\xi_\beta+S^{\alpha\beta}\nabla_\alpha\xi_\beta)\nu_g.
\end{equation*}
Repeated application of the Cauchy-Schwartz inequality and $\square_g u=f$ then gives \cite{wilson}.
\begin{equation}\label{energyin}
E(\tau) \leq E(0) + C_1||f||^2_{L^2(U_\tau)} + C_2 \int_0^\tau E(s)ds
\end{equation} 
which on applying Gronwall's inequality gives
\begin{equation*}
E(\tau) \leq \left(E(0) + C_1||f||^2_{L^2(U_\tau)} \right)e^{C_2\tau}.
\end{equation*}
In terms of the Sobolev type norms this gives 
\begin{equation*}
||u||_{\tilde H^1(\Sigma_\tau\cap U^+)} \leq K \left(||u||_{\tilde H^1(\Sigma_0\cap U^+)} 
+||f||_{L^2(U_\tau)}\right).
\end{equation*} 
\end{proof}

We now use the energy inequality (\ref{enegyin}) to prove uniqueness
of the solution as well as the causal support properties of the
solution to the Cauchy problem.

\begin{Proposition}\label{uniqueness}{\rm (Uniqueness)} \\
  Let $(M,g)$ be a (connected, oriented, time oriented,) globally
  hyperbolic $(n+1)$-dimensional Lorentzian manifold with $C^{1,1}$
  metric and $\Sigma$ a smooth spacelike $n$-dimensional spacelike Cauchy
  surface. Let $u$ be a (weak) $H_{\text{loc}}^2(M)$
  solution of
$$
\square_g u=f
$$
with $f \in H^1_{\text{loc}}(M)$, which satisfies the initial conditions 
\begin{align*}
u|_\Sigma&=u_0,\\
\nabla_nu|_\Sigma&=u_1, \quad \hbox{where $n$ is the unit normal to $\Sigma$}
\end{align*}
with $(u_0,u_1) \in H^2(\Sigma) \times H^1(\Sigma)$. Then $u$ is unique.
\end{Proposition}

\begin{proof}
  Let $q \in M$ and without loss of generality suppose that $q \in I^+(\Sigma)$. Then since our
  spacetime is globally hyperbolic we may find a $p$ such that $q \in
  I^-(p) \cap I^+(\Sigma):=U$ where by $C^{1,1}$ causality theory $U$ has
  compact closure \cite{KSSV}. Now suppose there exist two solutions
  $u$ and $\tilde u$ to the above initial value problem. Then applying
  Lemma \ref{HE7.4.4} to $\hat u:=u-\tilde u$ over the region $U^{+}$
  gives
\begin{equation*}
||\hat u||_{\tilde H^1(\Sigma_\tau\cap U^+)} \leq K ||\hat u||_{\tilde H^1(\Sigma_0\cap U^+)}=0 
\end{equation*} 
Hence $||\hat u||_{\tilde H^1(\Sigma_\tau \cap U^+)}=0$ so that $\hat u$
and $\nabla \hat u$ vanish in $U^{+}$. Since $q$ is arbitrary the solution is unique in $D^+(\Sigma)$. A
similar result applies to $D^-(\Sigma)$, so we have uniqueness in the whole
of $M=D(\Sigma)$.
\end{proof}

\begin{Proposition}\label{causalsupport}{\rm (Causal Support)}\\
  Let $(M,g)$ be a connected, oriented, time oriented, globally
  hyperbolic $(n+1)$-dimensional Lorentzian manifold with $C^{1,1}$
  metric, $\Sigma_0$ a smooth  spacelike $n$-dimensional Cauchy
  surface and $t$ a smooth temporal function such that $t^{-1}(0)=\Sigma_0$. 
 Let 
$u \in C^0(\R, H^2(\Sigma_t))\cap C^1(\R, H^1(\Sigma_t)) \cap H_{\text{loc}}^2(M)$ 
be a (weak) solution to the initial value problem
\begin{align*}
\square_g u&=f \quad \hbox{on $M$,}\\
u&=u_0 \quad \hbox{on $\Sigma$,}\\
\nabla_n&=u_1 \quad \hbox{on $\Sigma$.}
\end{align*}
Then $\supp(u) \subset J(\supp(u_0) \cup \supp(u_1) \cup \supp(f))$.
\end{Proposition}

\begin{proof}
  We prove the result for $J^+$. A similar proof holds for $J^-$.  Let
  $V=J^+(\supp(u_0) \cup \supp(u_1) \cup \supp(f))$ and suppose $q \in
  M \setminus V$. Then we may find a point $p \in D^+(\Sigma)$ such that $q
  \in I^-(p) \cap I^+(\Sigma)$ and $u$ and $\nabla u$ vanish on $J^-(p)
  \cap \Sigma$ and $f$ vanishes on $J^-(p) \cap J^+(\Sigma)$. Now let $U=I^-(p)
  \cap I^+(\Sigma)$ and apply (\ref{energyin}) on this region to obtain
\begin{equation*}
||u||_{\tilde H^1(\Sigma_\tau \cap J^+(\Sigma_0) \cap I^-(p))} \leq K 
\left(||u||_{\tilde H^1(\Sigma_0 \cap J^-(p))} +||f||_{L^2(J^+(\Sigma)\cap J^-(p))}\right), \quad \hbox{for $0 \leq \tau \leq t(p)$}.
\end{equation*}

But by the choice of $p$ the right hand side vanishes, so $u$ must
also vanish in the region $U$ with $\tau$ in
the range $0 \leq \tau \leq t(p)$. So that $u$ vanishes on a neighbourhood of q. Since, $q$ was arbitrary $u$ vanishes on $M
\setminus V$ which proves the result.
\end{proof}

\bigskip

To establish existence on $M$ we need the following two Lemmas from
Ringstr\"om \cite{Ring}. In both cases the proof given in \cite{Ring}
for the smooth case goes through to that of a $C^{1,1}$ metric
unchanged.

\begin{Lemma}\label{12.5}{\rm (Ringstr\"om \cite[Lemma 12.5]{Ring})}\\
  Let $(M,g)$ be an $(n+1)$-dimensional Lorentzian manifold with
  $C^{1,1}$ metric and $\Sigma$ a smooth spacelike $n$-dimensional
  submanifold. If $p \in S$ there is a chart $(U,x)$ with $p \in U$
  and $x=(x^0,x^1,\ldots,x^n)$ such that $q \in U \cap \Sigma$ if and only
  if $q \in U$ and $x^0(q)=0$. Furthermore we may choose $x$ so that
  $\frac{\partial}{\partial x^0}$ is the future directed unit normal
  to $\Sigma$ for $q \in \Sigma \cap U$. 

  If we fix $\epsilon >0$ and let
  $g_{\mu\nu}:=g(\frac{\partial}{\partial
    x^\mu},\frac{\partial}{\partial x^\nu})$, then we can assume $U$
  to be such that $|g_{0i}| \leq \epsilon$ $i=1,\ldots,n$ on $U$. If
  we let $a=g_{00}(p)$ and $b>0$ be such that $g_{ij}(p)$ regarded as
  a positive definite matrix is bounded below by $b$ (i.e.
  $g_{ij}(p)\xi^i\xi^j>b|\xi|^2$) then we may assume that
  $g_{00}(q)<a/2$ and $g_{ij}(q)$ regarded as a positive definite
  matrix is bounded below by $b/2$ for $q \in U$.
\end{Lemma}

\begin{Lemma}\label{12.16}{\rm (Ringstr\"om \cite[Lemma 12.16]{Ring})}\\
  Let $(M,g)$ be a (connected, oriented, time oriented,) globally
  hyperbolic $(n+1)$-dimensional Lorentzian manifold with $C^{1,1}$
  metric and $\Sigma$ a smooth spacelike $n$-dimensional Cauchy
  surface. Let $t$ be a (smooth) temporal function (as given by
  \cite{BernardSuhr} \cite{MinguzziAddenda}) with $t^{-1}(0)=\Sigma_{t_0}$. If $p \in \Sigma_{t_0}$ there is an
  $\epsilon>0$ and open neighbourhoods $U$, $W$ of $p$ such that
\begin{enumerate}
\item the closure of $W$ is compact and contained in $U$;
\item if $q \in W$ and $\tau \in [t_0-\epsilon, t_0+\epsilon]$, then
  $J^+(\Sigma_\tau) \cap J^-(q)$ is compact and contained in $U$;
\item there is a chart $(U,\phi)$ with $\phi=(x^0,\ldots,x^n)$ and $x^0=t$,
  such that there exist $a,b>0$ with $g_{00}(q)<-a$ and
  $g_{ij}(q)\xi^i\xi^j\geq b\delta_{ij}\xi^i\xi^j$ for $q \in U$;
\item for any compact $K \subset U$ there is a  $ C^{1,1}$ matrix valued
  function $h$ on $\Real^{n+1}$ such the $h_{\mu\nu}=g_{\mu,\nu}\circ
  \phi^{-1}$ on $\phi^{-1}(K)$ and such that there are positive
  constants $a_1$, $b_1$ $c_1$ with $h_{00} \leq -a_1$,
  $h_{ij}\xi^i\xi^j \geq b_1\delta_{ij}\xi^i\xi^j$ and $|h_{\mu\nu}|
  \leq c_1$ on all of $\Real^{n+1}$.
\end{enumerate}
\end{Lemma}
Note that that although the proof is identical to that in \cite{Ring} it relies on the $C^{1,1}$  causality results of \cite{KSSV}  and \cite[Theorem 1]{BernardSuhr}. Note also that in point (4) above the matrix valued function is only
$C^{1,1}$ rather than smooth as it is in \cite[Lemma 12.16]{Ring}. 

We are now in a position to establish existence.

\begin{Proposition}\label{ex1}{\rm (Existence for compactly supported source and initial data)}\\
  Let $(M,g)$ be a time oriented $(n+1)$-dimensional Lorentzian
  manifold with $C^{1,1}$ metric and $\Sigma$ a smooth spacelike
  $n$-dimensional hypersurface. Let $t$ be a smooth temporal function with $t^{-1}(0)=\Sigma$ and let $n$ be the future directed
  timelike unit normal to $\Sigma$. 

  Given initial data $(u_0,u_1) \in H^2_{\text{comp}}(\Sigma) \times
  H^1_{\text{comp}}(\Sigma)$ and source $f \in H^1_{\text{comp}}(M)$,
  then there exists a (weak) solution $u \in C^0(\Real,
  H^2(\Sigma_t))\cap C^1(\Real, H^1(\Sigma_t)) \cap
  H_{\text{loc}}^2(M)$ to the initial value problem
\begin{align*}
\square_g u&=f \quad \hbox{on $M$,}\\
u&=u_0 \quad \hbox{on $\Sigma$,}\\
\nabla_nu&=u_1 \quad \hbox{on $\Sigma$.}
\end{align*}

\end{Proposition}

\begin{proof}
  We again closely follow Ringstr\"om \cite[Theorem 12.17]{Ring}.\\
Let $K_1 \subset \Sigma$ be a compact set such that $\supp(u_0) \cup  \supp(u_1) \subset K_1$ and $K_2 \subset M$ a compact set such that $\supp(f) \subset K_2$. 
  Let $t_1>0$ and define $R_{t_1}$ to be the set of $q$ such that $0
  \leq t(q) \leq t_1$. Then $R_{t_1}$ is closed and $K_3=K_2 \cap
  R_{t_1}$ is compact. The union of $I^+(p)$ for $p \in I^-(\Sigma)$ is an
  open cover of $K_1 \cup K_3$ so there is a finite number of points
  $p_1,\ldots,p_\ell$ such that the $I^+(p_i)$ are finite subcover of
  $K_1 \cup K_3$. Note that the set $F=\bigcup_{i=1}^\ell J^+(p_i)
  \cap J^-(\Sigma_{t_1})$ is compact and that if there is a solution in
  $R_{t_1}$ it has to be zero in $R_{t_1} \setminus F$ by Proposition
  \ref{causalsupport}. We now show that there is a solution in the compact set
  $R_{t_1} \cap F$.  

  Let $F_\tau:=F \cap \Sigma_\tau$. Let $0 \leq \tau <t_1$ and assume we
  have a solution in the function space specified in the proposition up to time $\tau$,
  i.e. on $R_\tau$ or for every $R_s$ with $0 \leq s <\tau$. We now
  extend the solution to the future of $\tau$. For every $p \in F_\tau$
  there are neighbourhoods $U_p$, $W_p$ and $\epsilon_p$ with the
  properties of Lemma \ref{12.16}. By compactness there is a finite
  number of points $\tilde p_i,\ldots, \tilde p_N$ such that the
  $W_{\tilde p_i}$ cover $F_\tau$. Let $0 <\epsilon \leq
  \min\left\{\epsilon_{\tilde p_1},\ldots ,\epsilon_{\tilde
      p_N}\right\}$ be such that
  \begin{equation}
  F_s \subset \bigcup_{i=1}^NW_{\tilde p_i}
  \end{equation}
  for all $s \in [\tau-\epsilon, \tau+\epsilon]$. Now let $s_1 \in  [\tau-\epsilon, \tau]$ be such that there is a solution, in the
  function space specified in the proposition 
  up to and including $s_1$ and let $p \in F_s$ for any $s \in [s_1,
  \tau+\epsilon]$. Then $K_p:=J^-(p)\cap J^+(\Sigma_{s_1})$ is compact and
  contained in one of the charts, say $(U_{\tilde p_k}, \phi)$. Let
  $\chi \in C^\infty_0(U_{{\tilde p}_k})$ be such that $\chi(q)=1$ for
  all $q \in K_p$. Then we use our solution up to time $s_1$ to define
  new initial data on $\Sigma_{s_1}$ given by $\tilde u_0:=(\chi
  u)|_{\Sigma_{s_1}} \in H^2(\Sigma_{s_1})$ and $\tilde
  u_1:=(\chi\nabla_nu)|_{S_{s_1}} \in H^1(\Sigma_{s_1})$ and source $\tilde
  f:=\chi f$.  These all have their support within $U_{\tilde p_k}$ so
  we may use the chart $(U_{\tilde p_k}, \phi)$ to regard these are
  data and source on the whole of $\Real^n$ and $\Real^{n+1}$
  respectively. We may also extend the Lorentz metric $g_{\mu \nu}
  \circ \phi^{-1}$ to a Lorentz matrix-valued function $h_{\mu\nu}$ on
  the whole of $\Real^{n+1}$ which coincides with $g_{\mu\nu} \circ
  \phi^{-1}$ on $K_p$. We may therefore regard the tildered version as
  an initial value problem on $\Real^{n+1}$. The third condition of
  Lemma \ref{12.16} and the fact that the solution is in $C^0(\Real,
  H^2(\Sigma_t))\cap C^1(\Real, H^1(\Sigma_t)) \cap H_{\text{loc}}^2(M)$ ensures that
  we may apply Proposition \ref{weaksol} to obtain a solution on
  $\Real^{n+1}$ which on $\phi(U_{\tilde p_k})$ may be transferred back
  to give a solution on $K_p$.
In the region $V_p:=I^-(p) \cap J^+(\Sigma_{s_1})$ we
    define $u$ to be this solution.
In the region $V_p \cap V_q$ then uniqueness
  ensures that the two potential solutions coincide. We now define
  $O_1$ to be the union of the $V_p$ for $p \in F_s$ , $s \in
  [s_1,\tau+\epsilon]$ then the above construction defines a unique
  solution in $O_1$. Note that the interior of $O_1$ contains $F_s$
  for all $s \in (s_1,\tau+\epsilon)$. Now define $O_2$ to be the set
  of points for which $s_1 \leq t(q) <\tau+\epsilon$ for which $q
  \notin F$. We want to define the solution to be zero in this set,
  however we need to check that there is no contradiction for points
  in both $O_1$ and $O_2$. If $q \in O_2 \cap O_1$ with $t(q)>s_1$
  then both $u$ and $\nabla u$ vanish at $J^-(q) \cap S_{s_1}$ and $f$
  vanishes in $J^-(q) \cap J^+(\Sigma_{s_1})$. Furthermore there is an
  $r$ such that $q \in V_r \subset O_1$. So by uniqueness the solution
  defined on $O_1$ has to vanish for a sufficiently small neighbourhood at $q$.


  In summary we have shown that if there exists a solution
   for all $s<\tau$, or up to time $\tau$, in the
  required function space we get a solution in the same space on the
  larger region $R_{\tau+\epsilon}$ for some $\epsilon >0$.  Let
  $\mathcal A$ be the set of $s \in [0,\infty)$ such that there is a
    solution up to time $s$.  Taking $\tau=0$ in the
      above we have a solution on $R_{\epsilon}$ so $\mathcal A$ is
      not empty. We have also shown that if $\tau \in {\mathcal A}$ then a solution exists for $[0,\tau+\epsilon)$, so that for any $\tau>0$ we may find an open interval containing $\tau$ in which a solution exists. Thus $\mathcal A$ is open in the  relative topology of $ [0,\infty)$. Finally we note that by definition  $\mathcal A$ is also closed in $[0,\infty)$ because it
      contains its limit points. Then this set is open, closed
    and non-empty so must be the whole of $[0,\infty)$ and we have a
      solution for all future times. By time
      reversal we also have a solution for all past times and hence on  the whole of $M$.
\end{proof}

\begin{Theorem}\label{globalexist}{\rm (Global Existence and Uniqueness)}\\
  Let $(M,g)$ be a connected, oriented, time oriented
  $(n+1)$-dimensional Lorentzian globally hyperbolic manifold with
  $C^{1,1}$ metric and $\Sigma$ a smooth spacelike $n$-dimensional
  Cauchy hypersurface. Let $t$ be a smooth temporal function with
  $t^{-1}(0)=\Sigma$ and let $n$ be the future directed timelike unit
  normal to $\Sigma$.

  Given initial data $(u_0,u_1) \in H^2(\Sigma) \times H^1(\Sigma)$ and 
source  $f \in C^0(\R,H^1(\Sigma_t)) \cap H^1_{\text{loc}}(M)$
then there exists a unique (weak) solution $u \in C^0(\Real,
    H^2(\Sigma_t))\cap C^1(\Real, H^1(\Sigma_t) \cap H_{\text{loc}}^2(M)$ to the initial
    value problem
\begin{align*}
\square_g u&=f \quad \hbox{on $M$,}\\
u&=u_0 \quad \hbox{on $\Sigma$,}\\
\nabla_nu&=u_1 \quad \hbox{on $\Sigma$.}
\end{align*}
Moreover $\supp(u) \subset J(\supp(u_0) \cup \supp(u_1) \cup \supp(f))$.

\end{Theorem}

\begin{proof}
  Let $p$ be any point to the future of $\Sigma$. Then $K_p=J^-(p) \cap
  J^+(\Sigma)$ is a compact set.  Let $\chi \in C^\infty_0(M)$ be such that
  $\chi(q)=1$ for all $q \in K_p$.  Now define $f'=\chi f$, $u_0'=\chi
  u_0$ and $u_1'=\chi u_1$.  Then by Proposition \ref{ex1} there is a
  unique solution $u'$ to the primed initial value problem. Now set
  $u=u'$ in $I^-(p) \cap J^+(\Sigma)$. If now  $r \in K_p \cap K_q$ then by
  uniqueness the two potential solutions agree, so there is no
  contradiction. Thus we have a solution on the whole of  $D^+(\Sigma)$. A similar argument gives us a solutions on $D^-(\Sigma)$ and hence on the whole of $M$. The
  solution is unique by Proposition (\ref{uniqueness}) and satisfies
  the causal support condition by Proposition (\ref{causalsupport}).

 \end{proof}

We also want to show that the initial value problem is
  well-posed. For solutions in $H^1_{\text{loc}}(M)$ this follows immediately
  from the energy estimate (\ref{enegyin}). But well-posedness in
  $H^2_{\text{loc}}(M)$ requires a higher order estimate which we now establish.

\begin{Lemma}\label{tlocallemma}
  For each compact subset $K \subset M$ there exists a $\delta > 0$  with the following property: If
   $(u_0,u_1) \in H^2(\Sigma) \times H^1(\Sigma)$
  with $\supp(u_j) \subset K \cap \Sigma$ ($j=1,2$) and $f \in H^1(M)$ with
  $\supp(f) \subset K$, then  the (weak) solution $u \in C^0(\Real,
    H^2(\Sigma_t))\cap C^1(\Real, H^1(\Sigma_t)) \cap H_{\text{loc}}^2(M)$ 
of  $\square_g u=f$ with initial data $(u_0,u_1)$ satisfies the energy inequality 
\begin{equation}\label{tlocaleqn}
||u||_{C^0([0,\delta],H^{2}(\Sigma_t))}+||u||_{C^1([0,\delta], H^{1}(\Sigma_t))} 
\le C \Big( \norm{ u_0}{H^2(\Sigma)}+ \norm{u_1}{H^1(\Sigma)}+  \norm{f}{H^1(M)}\Big).  
\end{equation}

\end{Lemma}

\begin{proof} We use a similar approach to that in the existence proof given in \cite[Theorem 3.2.11]{bgp}. \\
  For every $p \in K$ there are neighbourhoods $U_p$, $W_p$ and
  $\epsilon_p > 0$ with the properties of Lemma \ref{12.16}. By
  compactness there is a finite number of points $p_1,\cdots, p_N$
  such that the corresponding $W_{p_j}$ cover $K$. Now let
  $\{\chi_j\}_{j=1}^N$ be a partition of unity of $K$ subordinate to
  the $W_{p_j}$.

  We now define
\begin{equation*}
u_{0,j}:=\chi_ju_0, \qquad u_{1,j}:=\chi_ju_1, \qquad f_j:=\chi_j f.
\end{equation*}
So that 
\begin{equation*}
u_0=\sum_{j=1}^N u_{0,j}, \qquad u_1=\sum_{j=1}^N u_{1,j}, \qquad f=\sum_{j=1}^N f_j.
\end{equation*}
We also define
\begin{equation*}
K_j:=\supp(u_{0,j}) \cup \supp(u_{1,j}) \cup \supp(f_j) \subset K \cap \bar \supp(\chi_j) \subset W_{p_j}.
\end{equation*} 
Let $u_j$ be the (weak) solution of the IVP
\begin{equation*}
\square u_j=f_j, \qquad u_j|_{\Sigma} = u_{0,j}, \qquad \nabla_n u_j|_{\Sigma}=u_{1,j}.
\end{equation*}

We will employ (an implicit choice of a temporal function in terms of)
the diffeomorphism $M \cong \R \times \Sigma$ and slightly abuse
notation from now on by considering all functions $u$, $u_j$ etc.\ to
be defined already on products $I \times \Sigma$, where $I$ is some
open real interval, thus suppressing the transfers of functions via
restrictions of the underlying global diffeomorphism.

By point two of Lemma \ref{12.16} there exists an $\epsilon_{p_j}>0$ such that 
\begin{equation}\label{chart}
\left((-2\epsilon_{p_j},2\epsilon_{p_j}) \times \Sigma \right) \cap J(K_j) \subset U_{p_j}.
\end{equation}

Let $\delta=\min\{\epsilon_{p_1},\cdots, \epsilon_{p_N}\}$. Given the solutions
$u_j$ of the local problem we may extend them by zero on all of
$(-2\delta,2\delta) \times \Sigma$ and sum them to give our \emph{unique} solution
\begin{equation*}
u=\sum_{j=1}^N u_j, \hbox{ on } (-2\delta,2\delta) \times \Sigma.
\end{equation*}
Since by (\ref{chart}) each of the $u_j$ lie entirely within some chart $(U_{p_j},
\phi_{p_j})$ we may regard the initial value problem as one on $I \times
\Real^n$ where $I$ is  an interval chosen sufficiently large such that the images of $(-2\delta,2\delta) \times \Sigma$ under all the $\phi_{p_j}$ are contained in $I \times \Real^n$.
 
Then the third condition of Lemma \ref{12.16} enables us to transfer the basic energy estimate according to Lemma \ref{energy_estimates} from $I \times \Real^n$ to ones for $u_j$ on $U_j \subset M$ to give
\begin{equation*}
||u_j||_{C^0([0,\delta],H^{2}(\Sigma))}+||u_j||_{C^1([0,\delta], H^{1}(\Sigma))} 
 \le C_j\left( \norm{u_{0,j}}{H^2(\Sigma)}+ \norm{u_{1,j}}{H^1(\Sigma)}+  \norm{f_j}{L^2([0,\delta], 
H^1(\Sigma))}\right).
\end{equation*}
Now $u=\sum_ju_j$ so that
\begin{align*}
||u||_{C^0([0,\delta],H^{2}(\Sigma))}+||u||_{C^1([0,\delta], H^{1}(\Sigma))}
&\le \sum_{j=1}^N\left\{||u_j||_{C^0([0,\delta],H^{2}(\Sigma))}+||u_j||_{C^1([0,\delta], H^{1}(\Sigma))} \right\}\\
& \le \sum_{j=1}^N C_j\left( \norm{ u_{0,j}}{H^2(\Sigma)}+ \norm{u_{1,j}}{H^1(\Sigma)}+  \norm{f_j}{L^2([0,\delta], H^1(\Sigma))}\right)\\
& \le \sum_{j=1}^N C_j\left( \norm{ u_{0}}{H^2(\Sigma)}+ \norm{u_{1}}{H^1(\Sigma)}+  \norm{f}{L^2([0,\delta], H^1(\Sigma))}\right)\\
& \le C \left( \norm{ u_{0}}{H^2(\Sigma)}+ \norm{u_{1}}{H^1(\Sigma)}+  \norm{f}{L^2([0,\delta], 
H^1(\Sigma))}\right),
\end{align*}
where we may replace $\norm{f}{L^2([0,\delta], H^1(\Sigma))}$ by the larger value $\norm{f}{H^1(M)}$, since  $f  \in H^1_{\text{comp}}(M)$.
\end{proof}

We remark that in the above proof (and formulation of the result) we have replaced the norm $\norm{f}{L^2([0,\delta], H^1(\Sigma))}$ by $\norm{f}{H^1(M)}$, which is valid for $f  \in H^1_{\text{comp}}(M)$, to avoid the need to specify a particular choice of a temporal function.

\begin{Proposition}\label{globalv2}{\rm (Global higher energy estimates)}\\
  Let $(M,g)$ be a time oriented $(n+1)$-dimensional Lorentzian
  manifold with $C^{1,1}$ metric and $\Sigma$ a smooth spacelike
  $n$-dimensional hypersurface. Let $t$ be a smooth temporal function with $\Sigma=t^{-1}(0)$ and let $n$ be the future directed
  timelike unit normal to $\Sigma$. 

Given initial data  $(u_0,u_1) \in H^2_{\text{comp}}(\Sigma) \times H^1_{\text{comp}}(\Sigma)$
and source $f
    \in H^{1}_{\text{comp}}(M))$, then
    the (weak) solution $u \in C^0(\Real,
    H^2(\Sigma_t))\cap C^1(\Real, H^1(\Sigma_t)) \cap H_{\text{loc}}^2(M)$ satisfies

\begin{equation*}
||u||_{C^0([0,T],H^{2}(\Sigma))}+||u||_{C^1([0,T], H^{1}(\Sigma))} 
\le C \Big( \norm{ u_0}{H^2(\Sigma)}+ \norm{u_1}{H^1(\Sigma)}+  ||f||_{H^{1}(M)}\Big)  
\end{equation*}
for any  interval $[0,T]$.

\end{Proposition}

\begin{proof}
  We first use Lemma \ref{tlocallemma} to obtain an estimate for the
  data $\hat u_0:=u|_{\Sigma_\delta}$ and $\hat u_1 :=\nabla _n
  u|_{\Sigma_\delta}$ induced by $u$ on $\Sigma_\delta$. It follows
  from (\ref{tlocaleqn}) and the fact that that $u \in
  C^0(\Real,H^2(\Sigma_t))\cap C^1(\Real, H^1(\Sigma_t)) \cap
  H_{\text{loc}}^2(M)$ that
\begin{equation*}
\norm{\hat u_0}{H^2(\Sigma)}+ \norm{\hat u_1}{H^1(\Sigma)}
\le \tilde C \left(\norm{ u_0}{H^2(\Sigma)}+ \norm{u_1}{H^1(\Sigma)}+  ||f||_{H^{1}([0,\delta] \times \Sigma)} \right). 
\end{equation*}
Now applying Lemma \ref{tlocallemma} to the
initial surface $\Sigma_\delta$ we obtain a $\hat \delta >\delta$ such
that
\begin{multline*}
||u||_{C^0([\delta,\hat\delta],H^{2}(\Sigma))}+||u||_{C^1([\delta,\hat\delta], H^{1}(\Sigma))} 
\le \hat C \left( \norm{\hat u_0}{H^2(\Sigma)}+ \norm{\hat u_1}{H^1(\Sigma)}+  ||f||_{H^{1}([\delta, \hat\delta] \times \Sigma)}\right) \nonumber \\
\le \hat C \left\{\tilde C \left(\norm{ u_0}{H^2(\Sigma)}+ \norm{u_1}{H^1(\Sigma)}+  ||f||_{H^{1}([0,\delta] \times \Sigma)} \right)
+ ||f||_{H^{1}([\delta, \hat\delta] \times \Sigma)}\right\} \nonumber \\
\le C_1 \left(\norm{ u_0}{H^2(\Sigma)}+ \norm{u_1}{H^1(\Sigma)}+  ||f||_{H^{1}([0,\hat\delta] \times \Sigma)} \right).
\end{multline*}

Combining the two energy inequalities on $[0,\delta]$ and $[\delta, \hat\delta]$ we have
\begin{multline*}
||u||_{C^0([0,\hat\delta],H^{2}(\Sigma))}+||u||_{C^1([0,\hat\delta], H^{1}(\Sigma))} \\
 \le ||u||_{C^0([0,\delta],H^{2}(\Sigma))}+||u||_{C^1([0,\delta], H^{1}(\Sigma))} +
||u||_{C^0([\delta ,\hat\delta],H^{2}(\Sigma))}+||u||_{C^1([\delta,\hat\delta], H^{1}(\Sigma))} \\
\le C \left( \norm{ u_0}{H^2(\Sigma)}+ \norm{\tilde u_1}{H^1(\Sigma)}+  ||f||_{H^{1}([0, \delta] \times \Sigma)}\right) + C_1 \left(\norm{ u_0}{H^2(\Sigma)}+ \norm{u_1}{H^1(\Sigma)}+  ||f||_{H^{1}([0,\hat\delta] \times \Sigma)} \right) \\
\le C_2 \left( \norm{ u_0}{H^2(\Sigma)}+ \norm{\tilde u_1}{H^1(\Sigma)}+  ||f||_{H^{1}([0, \hat\delta] \times \Sigma)}\right).
\end{multline*}
This shows that we may extend the energy inequality from $[0, \delta]$
to the larger time interval $[0,\hat\delta]$  by repeatedly applying Lemma \ref{tlocallemma}. 

Similarly, for every  $\tau\in [0,T]$ we may find  a $\delta(\tau)>0$ such that
Lemma \ref{tlocallemma} applies to the time interval 
$(\tau-\delta(\tau), \tau+\delta(\tau))$  with initial data given on $\Sigma_\tau$. 
By compactness, finitely many intervals 
$(\tau_k-\delta(\tau_k),\tau_k+\delta(\tau_k))$ ($k=1,...,m$) cover $[0,T]$ and the energy inequalities on these may be combined.
\end{proof}

We now use the above proposition to obtain a spacetime energy inequality

\begin{Proposition}(Higher order spacetime energy estimates) \\
  Let $(M,g)$ be a time oriented $(n+1)$-dimensional Lorentzian
  manifold with $C^{1,1}$ metric and $\Sigma$ a smooth spacelike
   $n$-dimensional hypersurface. Given initial data $(u_0,u_1) \in
  H^2_{\text{comp}}(\Sigma) \times H^1_{\text{comp}}(\Sigma)$  and source $f \in H^{1}_{\text{comp}}(M)$, then the (weak)
  solution $u$ satisfies the energy-inequality
\begin{equation}\label{stee}
||u||_{H^{2}(K)}\le C \Big( \norm{ u_0}{H^2(\Sigma)}+ \norm{u_1}{H^1(\Sigma)}+  
\norm{f}{H^{1}(M)}\Big).  
\end{equation}
for any compact $K\subset M$.
\end{Proposition}

\begin{proof}
Without loss of generality we may assume that $K\subset[0,T]\times \Sigma$. 
Due to the regularity of $u \in C^0(\Real, H^2(\Sigma_t))\cap
  C^1(\Real, H^1(\Sigma_t))$ we have control of the second order
  spatial derivatives, the second order mixed derivatives and the lower order terms
\begin{align}
\max(||\partial_i \partial_j u||_{L^{2}(K)},||\partial_j u||_{L^{2}(K)})&\le C \left(||u||_{C^0([0,T],H^{2}(\Sigma))}\right),\label{h2estimatelow1} \\
\max(||\partial_i \partial_t u||_{L^{2}(K)},||\partial_t u||_{L^{2}(K)})&\le C \left(||u||_{C^1([0,T],H^{1}(\Sigma))}\right),
\label{h2estimate2}
\end{align}

In order to obtain the required estimate we also need to control the $\partial_{tt}u$ in the $L^{2}(K)$ norm.
Using $\square_g u=f$ we have
\begin{equation*}
||g^{00}\partial_{tt}u||_{L^{2}(K)}=||f+(-g^{ti}\partial_t\partial_i-g^{ij}\partial_j\partial_i+g^{\alpha\beta}\Gamma^\gamma_{\alpha\beta}\partial_\gamma) u||_{L^{2}(K)}   
\end{equation*}

From \eqref{h2estimatelow1},\eqref{h2estimate2}, the regularities of
$f$ and $g$, we obtain an $L^2$ estimate for $\partial_{tt}u$,
\begin{align*}
||\partial_{tt}u||_{L^{2}(K)}&=C_1||f+(-g^{ti}\partial_t\partial_i-g^{ij}\partial_j\partial_i+g^{\alpha\beta}\Gamma^\gamma_{\alpha\beta}\partial_\gamma) u ||_{L^{2}(K)}  \\
 &\le C_2\left(  ||u||_{C^0([0,T],H^{2}(\Sigma))}+||u||_{C^1([0,T], H^{1}(\Sigma))} +||f||_{L^{2}(K)}\right).
\end{align*}
Combining the above with Proposition \ref{globalv2} completes the proof.
\end{proof}

Equation (\ref{stee}) implies the following result.
\begin{Corollary}\label{sol}
  The solution to the Cauchy problem described in Theorem
  \ref{globalexist} is well-posed in the sense that the
  \emph{solution map} 
$$ 
\text{Sol}\colon H^2(\Sigma) \times H^1(\Sigma) \times H^1_{\text{comp}}(M) \to H^2_{\text{loc}}(M), 
(u_0, u_1,  f) \mapsto  u,
$$
is continuous  in the topologies coming from the respective Sobolev spaces (see Appendix B).
\end{Corollary}


\section{Green operators for $C^{1,1,}$ spacetimes}

In this section we will define Green operators for $\square_g$ on
globally hyperbolic manifolds $M$ with $C^{1,1}$ metrics. We will show
existence and uniqueness of Green operators via the existence of
solutions to the wave equation with appropriate regularity and causal
support. We define below the notion of generalised
hyperbolicity which will give us the required conditions in this situation.

\begin{Definition} \label{genhyp2} {\bf (Generalised hyperbolicity)} A
  spacetime $(M,g_{ab})$ is said to satisfy the condition of
  {\emph{generalised hyperbolicity}} if the inhomogeneous wave equation for zero
  Cauchy data is
  well-posed and causal. \\

  The precise choice of function spaces in the definition of
  well-posedness depends upon the regularity of the metric. In our
  case we require the following conditions:
     
  \noindent
  {\bf Existence, uniqueness and support of solutions} \label{existence}
  For every $f\in H_{\text{comp}}^{1}(M)$ there exists a unique future solution $u^+\in H_{\text{loc}}^{2}(M,g)$  such that 

\begin{equation*} 
  \square_{g}u^+=f \text{ on } M
 \end{equation*} 
 which satisfies the causal support condition $\supp(u^+) \subset J^+(\supp(f))$.  Note that this condition implies that on
   a Cauchy surface to the past of $\supp(f)$ one must have zero initial
   data. However the particular choice of such Cauchy surface makes no
   difference to the solution (see proof of Theorem
   \ref{gws} below for more details).

We also require that for every $f\in H_{\text{comp}}^{1}(M)$   there exists a unique past solution $u^-\in H_{\text{loc}}^{2}(M,g)$  such that 
\begin{equation*} 
  \square_{g}u^-=f \text{ on } M
 \end{equation*} 
 which satisfies $\supp(u^-) \subset J^-(\supp(f))$ where we can
 choose any Cauchy surface to the future of $\supp(f)$ and solve the
 Cauchy problem going back in time.

Moreover, we require that the maps $f \mapsto u^+$  and  $f \mapsto u^-$ are continuous maps from $H^1_{\text{comp}}(M) \to H^2_{\text{loc}} (M)$ equipped with suitable topologies.

\end{Definition}

\begin{Theorem}\label{gws}
  Let $(M,g)$ be a time oriented $(n+1)$-dimensional Lorentzian
  manifold with $C^{1,1}$ metric and $\Sigma$ a smooth spacelike
  $n$-dimensional hypersurface. Then $(M,g)$ satisfies the condition
  of generalised hyperbolicity.
\end{Theorem}

\begin{proof}
  Theorem \ref{globalexist}  shows that a globally hyperbolic $C^{1,1}$
  spacetime satisfies the condition of generalised hyperbolicity to the future by considering the forward initial value problem 
\begin{equation*}
  \square_{g}u^+=f \text{ on } M, \:\:
 {u^+}(\Sigma_{+})=0,   \:\:
  \nabla_n{{u^+}}(\Sigma_{+})=0,
\end{equation*} 
where $f\in H_{\text{comp}}^{1}(M)$  and $\Sigma_{+}$ is a smooth spacelike Cauchy hypersurface such that $J^{+}(\supp(f))\cap \Sigma_{+}=\emptyset$. 

{\bf Note:} If we were to choose some other smooth spacelike Cauchy
hypersurface ${\tilde \Sigma}_+$, which also satisfies $J^+(\supp(f))
\cap {\tilde \Sigma}_+=\emptyset$, then the corresponding solution is
the same, since the divergence theorem arguments used in Lemma
\ref{HE7.4.4} apply and yield that the solution must vanish in the
region between $\Sigma_+$ and ${\tilde \Sigma}_+$.

Similarly, Theorem \ref{globalexist} shows it satisfies the condition of generalised hyperbolicity to the past by 
considering the backwards initial value problem
$$  
  \square_{g}u^-=f \text{ on } M,\:\:
{u^-}(\Sigma_{-})=0, \:\:
  \nabla_n{{u^-}}(\Sigma_{-})=0.
$$
where  $f\in H_{\text{comp}}^{1}(M)$  and $\Sigma_{-}$ is a smooth spacelike Cauchy hypersurface such that $J^{-}(\supp(f))\cap \Sigma_{-}=\emptyset$. Again the solution is independent of the choice of Cauchy surface as long as it satisfies the causal support condition. 
\end{proof}

\subsection{Green operators}

The definition of the Green operators in the non-smooth setting will
require us to choose suitable spaces of functions as domain and range
(see Theorem \ref{exact}). We therefore define the following spaces:
\begin{align}
V_0=&\{\phi\in H_{\text{comp}}^{2}(M) \text{ s.t. } \square_g\phi\in  H^{1}_{\text{comp}}(M) \} \nonumber\\
U_0=&H^{1}_{\text{comp}}(M)\\\nonumber
V_{sc}=&\{\phi\in  H_{\text{loc}}^{2}(M) \text{ s.t. } \square_g\phi\in H^{1}_{\text{loc}}(M) \\
&\text{ and } \text{supp}(\phi) \subset J(K) \text{ where K is a compact subset of } M\}\nonumber
\end{align}

\begin{rem}
Note that none of the spaces defined above depend upon the choice of background metric used in the definition of the Sobolev spaces.
\end{rem}

\begin{Definition} \label{Green} 
 A linear map $$G^{+}:H_{\text{comp}}^{1}(M)\rightarrow H_{\text{loc}}^{2}(M)$$ satisfying the  properties
 \begin{enumerate}
   \item $\square_{g}G^{+}=\text{id}_{H_{\text{comp}}^{1}(M)}$, 
   \item $ G^{+}\square_{g}|_{V_0}=\text{id}_{V_0}$,
   \item $\supp(G^{+}(f))\subset J^{+}(\supp(f))$ for all $f\in H_{\text{comp}}^{1}(M)$,
 \end{enumerate}
 \noindent
 is called an \emph{advanced  Green operator} for $\square_{g}$. A \emph{retarded 
 Green operator} $G^{-}$  is defined similarly.

\end{Definition}

\begin{rem}\label{weakremark}

(i)  Clearly,  the regularity condition in the definition of the space $V_0$ was chosen to guarantee that  $\square_g f$, given $f\in V_0$, belongs to the domain $H_{\text{comp}}^{1}(M)$ of the Green operators. 

(ii)
In several proofs below we will show the identities in Properties 1 and 2 to hold weakly, i.e.,  when evaluated on test functions in ${\mathcal D}(M)$. However it can be shown using the results of
H\"ormander \cite[Theorem 1.25]{Hoermander:1} that if two such Sobolev functions have the same effect on test functions then they are actually equal as Sobolev functions.  

(iii) The function spaces $H^2_{\text{loc}}(M)$ and
$H^1_{\text{comp}}(M)$ used as target space and domain for the Green
operators are in perfect accordance with the theory of so-called
\emph{regular fundamental solutions} for hyperbolic operators with
constant coefficients
\footnote{Of course, in case of the wave operator
  we even have explicit representations for the advanced and retarded
  fundamental solutions $E_+$ and $E_-$, e.g., in \cite[Sections 6.2
    and 7.4]{Hoermander:1}, but these are not required here.}
(cf.~\cite[Section 12.5]{Hoermander:2}) as we sketch briefly in the
following: Let $M$ be $(n+1)$-dimensional Minkowski space so that
$\square$ has the symbol $p(\tau,\xi) = \tau^2 -|\xi|^2$, which is
hyperbolic with respect to the directional vectors $(\pm 1,0)$
and produces the temperate weight function
    $\tilde{p}(\tau,\xi) := \sqrt{\sum_{|\alpha|
        \geq 0} |\partial^\alpha p(\tau,\xi)|^2} = \sqrt{(\tau^2 -
      \xi^2)^2 + 4(\tau^2 + \xi^2 +2)} \geq \sqrt{1 + \tau^2 + \xi^2}
    =: w_1(\tau,\xi)$. The unique fundamental solution $E_\pm$ with
    support in the half space where $\pm t \geq 0$
    belongs to
          $B^{\text{loc}}_{\infty,\tilde{p}}$, i.e., for every test
          function $\phi$ the Fourier transform $\mathcal{F}(\phi
          E_\pm)$ times $\tilde{p}$ is measurable and bounded. The
          advanced and retarded Green operators are then given by
          convolution $G^\pm f = E_\pm \ast f$ for every $f \in
          H^1_{\text{comp}}(\R^{n+1}) = \mathcal{E}'(\R^{n+1}) \cap
          B_{2,w1}$, where $B_{2,w1} = \{u \in \mathcal{S}' \mid w_1
          \cdot \mathcal{F}u \in L^2\} = H^1(\R^{n+1})$. Finally, we
          may apply \cite[Theorem 12.5.3 or Theorem
            10.1.24]{Hoermander:2} to obtain $G^\pm f \in
          B^{\text{loc}}_{2, \tilde{p} \cdot w_1} \subseteq
          B^{\text{loc}}_{2, w_1^2} = H^2_{\text{loc}}(\R^{n+1})$,
          since $\tilde{p}(\tau,\xi) w_1(\tau,\xi) \geq
          w_1^2(\tau,\xi) = 1 + \tau^2 + \xi^2$.

\end{rem}

We next show that the advanced and retarded Green operators are adjoints of one another. To do this we use the following Lemma
\begin{Lemma}\label{selfadj}
Given $\chi, \varphi\in H_{\text{loc}}^{2}(M,g)$ and $\supp(\chi)\cap supp(\varphi)$ compact. Then, we have 
\begin{equation}
\int_M \square_g \chi \varphi \nu_g =\int_M \chi \square_g\varphi \nu_g 
\end{equation}
\end{Lemma}
\noindent
The proof of the Lemma follows from using integration by parts twice and the support properties given in the hypothesis. Note that the specified regularity of the metric $g$ and of the functions is needed in order to use the $L^2$ inner product.
 We may now prove the following theorem.

\begin{Theorem}\label{greenadj}
Given Green operators satisfying conditions  1 and 3 of Definition \ref{Green} and $\chi,\varphi\in   H_{\text{comp}}^{1}(M) $ we have that

\begin{equation*}
\int_M G^{+}(\chi) \varphi \nu_g =\int_M \chi G^-(\varphi) \nu_g 
\end{equation*}
\end{Theorem}

\begin{proof}

First, notice that if $\chi,\varphi\in H_{\text{comp}}^{1}(M)$ we have that $
G^{+}(\chi), G^-(\varphi)\in H_{\text{loc}}^{2}(M,g)$. Moreover,
$G^{+}(\chi)\cap G^-(\varphi)\subset J^{+}(\supp(\chi))\cap
J^{-}(\supp(\varphi))$ is compact by the global hyperbolicity
condition.

Hence,
\begin{multline*}
\int_M G^{+}(\chi) \varphi \nu_g =(G^{+}(\chi),  \varphi )_{L^{2}(M,g)}
=(G^{+}(\chi),\square_g G^-(\varphi) )_{L^{2}(M,g)}\\
=(\square_g G^{+}(\chi), G^-(\varphi) )_{L^{2}(M,g)}
=(\chi, G^-(\varphi) )_{L^{2}(M,g)}
=\int_M \chi G^-(\varphi) \nu_g.
\end{multline*}
\end{proof}

We are now in a position to prove the main result about existence of Green
operators.

\begin{Theorem}\label{gf}
Let $(M, g)$ be a spacetime that satisfies the definition of
generalised hyperbolicity (Definition \ref{genhyp2}). Then there exist
unique continuous advanced and retarded Green operators for $\square_g$ on $M$.
\end{Theorem}

\begin{proof}
We will only discuss the advanced Green operator, the existence and the properties of the retarded Green operator follow from time reversal.

\medskip
{\emph{Existence:}} We define the linear map  
$$G^{+}:H_{\text{comp}}^{1}(M)\rightarrow   H_{\text{loc}}^{2}(M)$$
 which sends a source function $f$ to the (unique) advanced weak solution
 $u^{+}$. That such a $ u^{+}$ exists and is unique is a consequence of
 generalised hyperbolicity. Property 1 in Definition \ref{Green} is immediate. In addition, the energy estimate \eqref{stee} shows that $G^{+}$  is a continuous operator. It remains to prove Properties 2 and 3 in Definition \ref{Green}.

Property 2:  Let  $f\in V_0$ and $v\in \mathcal{D}(M)$, then
\begin{equation*}
(G^{+}(\square_{g}f), v)_{L^{2}(M,g)}=(\square_{g}f,G^{-}(v))_{L^{2}(M,g)}
=(f, \square_{g} G^{-}(v))_{L^{2}(M,g)}
=(f, v)_{L^{2}(M,g)},
\end{equation*}
where we have used Theorem \ref{greenadj} and Property 1 for the retarded Green operator $G^-$. Thus the weak form of the required identity holds, which implies $G^{+}\square_{g}(f)=f$ for every $f\in V_{0}$ (see Remark \ref{weakremark}(ii)).

Property 3 follows because $\supp(G^+)=\supp(u^{+})\subset J^{+}(\supp(f))$ by  Lemma \ref{causalsupport}.

\bigskip

{\emph{Uniqueness:}} Let $\tilde{G}^{+}$ be another linear operator satisfying Definition \ref{Green}.  Given $f\in H^{1}_{\text{comp}}(M)$ we have  that $v:=\tilde{G}^{+}(f)$ satisfies $\square_g v=f$ and $\supp(v)\subset J^{+}(\supp(f))$. Since $f\in H^{1}_{\text{comp}}(M)$, $\supp(v)\subset J^{+}(\supp(f))$ and $M$ is globally hyperbolic, there is a smooth timelike Cauchy surface  $\Sigma$ to the past of the support of $f$  where the Cauchy data vanishes, i.e., $v= 0$ and $\nabla_n v=0 $.  Hence, $v$ is a solution to the zero initial data forward Cauchy problem on $\Sigma$. By uniqueness we must have $v=u^{+}$ so we can conclude that $\tilde{G}^{+}(f)={G}^{+}(f)$ for all $f\in H^{1}_{\text{comp}}(M)$.
\end{proof}

We now show that the low-regularity Green operators satisfy an exact sequence result similar to that in the smooth case \cite[Theorem 3.4.7]{bgp}.

\begin{Theorem} \label{exact}

Let $M$ be a connected time-oriented globally hyperbolic Lorentzian manifold that satisfies Definition \ref{genhyp2}. Define the causal propagator as  $$G=G^{+}-G^{-}:  H_{\text{comp}}^{1}(M)\rightarrow  H_{\text{loc}}^{2}(M)$$ 

Then the image of $G$ is contained in $V_{sc}$ and the following complex is exact:
 
\begin{center}
\begin{tikzcd}
    0\arrow{r}& V_0 \arrow{r}{\square_{g}}& U_0\arrow{r}{G} &V_{\text{sc}} \arrow{r}{\square_g} & H^{1}_{\text{loc}} (M).
   \end{tikzcd}
\end{center}
\end{Theorem}

\emph{Proof of Theorem \ref{exact}:}  First we show that the sequence is a complex: We have from the definitions $G^{+}\square_g|_{V_0}=G^{-}\square_g|_{V_0}= \text{id}_{V_0}$ and $\square_g G^{+}=\square_g G^{-}= \text{id}_{ H_{\text{comp}}^{1}(M)}$, therefore  $G \square_g\phi=0$ for all $\phi\in V_0$ and $\square_g G\psi =0$ for all $\psi\in U_0$.

\begin{itemize}

\item Exactness at $V_0$, i.e.,  injectivity of $\square_g$: 
Let $\phi\in V_0$ be such that $\square _g \phi=0$. By compactness of the support there
is a smooth spacelike Cauchy hypersurface $\Sigma$ such that $\phi=0$ and $\nabla_n\phi=0$ on $\Sigma$. Therefore, $\phi$ is a solution to the Cauchy problem with vanishing  initial data and source. Uniqueness of the solution implies $\phi=0$.

\item Exactness at $U_0$: Let $\phi\in U_0$ be such that $\phi\in \ker(G)$, i.e., $G^{+}(\phi)=G^{-}(\phi)$. Define $\psi:=G^{+}(\phi)=G^{-}(\phi)$, hence  $\psi\in  H_{\text{loc}}^{2}(M)$ and $ \square_{g}\psi=\phi$.
Moreover, $\psi$ is compactly supported in $M$ because
$\text{supp}(\psi)\subset
\text{supp}(G^{+}(\phi))\cap\text{supp}(G^{-}(\phi))\subset
J^{+}(\supp(\phi))\cap J^{-}(\supp(\phi))$ and the latter is compact due to
global hyperbolicity. Thus there exists  $\psi\in
H_{\text{loc}}^{2}(M)$ such that $\square_g \psi=\phi \in U_0$ and
$\text{supp}(\psi)$ is compact. Hence, $\psi\in V_{0}$ 
  and $\phi \in \text{im} (G)$.

\item Exactness at $V_{\text{sc}}$: 
Let $\phi\in \text{ker} (\square_g)$ and $\phi\in V_{sc}$. Without loss of
  generality we may assume that $\supp(\phi)\subset I^{+}(K)\cup
  I^{-}(K)$ for some compact set\footnote{We may take
    $I^{+}(K)$ rather than $J^{+}(K)$ by replacing an initial choice
    of compact set $\tilde K$ with a slightly larger $K$ for which
    $J^+(\tilde K) \subset I^+(K)$. Specifically we can take $K=\bar
    U$ where $U$ is any open set containing $\tilde K$ with compact
    closure. Similar remarks apply to $I^-(K)$} $K$ of $M$. 
   Using a partition of unity $\{\chi_-,\chi_+\}$ subordinate to $\{I^{-}(K),
  I^{+}(K)\}$ we let $\phi_1=\chi_- \phi$ and $\phi_2=\chi_+\phi$, thus 
  $\phi=\phi_{1}+\phi_{2}$. Then $\supp(\phi_{1})\subset J^{-}(K)$ and
  $\supp(\phi_{2})\subset J^{+}(K)$, hence $\phi_1, \phi_2 \in
  V_{sc}$.

  Define $\psi:=-\square_g \phi_1=\square_g \phi_2 $. Then
  $\supp(\psi)$ is compact because $\supp(\psi)\subset J^{-}(K)\cap
  J^{+}(K)$. Moreover, $\psi\in
  H^{1}_{\text{loc}}(M)$ since $\phi\in V_{\text{sc}}$. Combining these two observations, we conclude that
  $\psi\in H_{\text{comp}}^{1}(M)$ and therefore $G^{+}(\psi)$ is defined.

For arbitrary $\chi\in \mathcal D(M)$ we have
\begin{equation*}
(\chi,G^{+}(\psi))_{L^{2}(M,g)}=(\chi,G^{+}(\square_g \phi_2))_{L^{2}(M,g)}
=(\square_g G^{-}\chi, \phi_2)_{L^{2}(M,g)}
=(\chi, \phi_2)_{L^{2}(M,g)}
\end{equation*}

which shows $G^{+}(\psi)=\phi_2$,  where we have made use of the fact
that the supports of $G^-\chi$ and $\phi_2$ intersect in a compact set due
to global hyperbolicity. Similarly, $G^{-}(\psi)=-\phi_{1}$ and therefore
$G(\psi)=G^{+}(\psi)-G^{-}(\psi)=\phi_{2}+\phi_{1}=\phi$. In summary, we
may conclude that there exists $\psi\in U_0$ satisfying
$G(\psi)=\phi$.
\end{itemize}
\  \hfill $\Box$

\subsection{Restrictions}

We briefly discuss the restriction of Green operators to \emph{causally compatible} subsets $\Omega \subset M$, that is, sets such that
$$J_{\Omega}(x)=J_{M}(x)\cap \Omega \quad \forall x\in \Omega.$$
 
We have the following theorem (cf.\ \cite[Proposition 3.5.1]{bgp}).
\begin{Theorem}\label{res}
Let $M$ be a time oriented connected globally hyperbolic manifold with
a $C^{1,1}$ Lorentzian metric, $G^{+}$ be the advanced Green operator
for $\square_g$ and $\Omega\subset M$ be a causally compatible open
subset. Then we may define an advanced Green operator for the
restriction of $\square_g$ to $\Omega$ by
$$
{\tilde{G}}^{+}(\varphi):=G^{+}(\varphi_{\text{ext}})|_{\Omega}, \quad \hbox{for $\varphi\in  H_{\text{comp}}^{1}(M)$ with $\supp(\varphi)\subseteq {\Omega}$}
$$
where $\varphi_{\text{ext}}$ denotes the extension of $\phi$ by zero. Similar results hold for $G^{-}$.
\end{Theorem}

\begin{rem}
  We denote the restriction of $\square_g$ to $\Omega$ by ${\widetilde
    \square}_g$. Notice that for all $u\in H_{\text{loc}}^{2}(M)$ we have
  ${\widetilde \square}_g(u|_{\Omega})=\square_g|_{\Omega}(u|_{\Omega})=(\square_gu)|_{\Omega}$
  and for all $u\in H^{2}(\Omega)$ with $\supp(u)\subseteq\Omega$ we
  have $({\widetilde\square}_gu)_{\text{ext}}=\square_g (u_{\text{ext}})$.
\end{rem}

\emph{Proof of Theorem \ref{res}:}

Property 1: 
Let $f\in H^{1}_{\text{comp}}(M)$ with $\supp(f)\subseteq {\Omega}$, then 
\begin{equation*}
{\widetilde\square}_{g}\tilde{G}^{+}(f)={\widetilde\square}_{g}({G}^{+}(f_{\text{ext}})|_{\Omega})
={\square_{g}}({G}^{+}(f_{\text{ext}}))|_{\Omega}
=f_{\text{ext}}|_{\Omega}
= f.
\end{equation*}

Property 2:
Let $f\in V_0$ with $\supp(f)\subset {\Omega}$, then 
\begin{equation*}
\tilde{G}^{+}({\widetilde{\square}}_{g}f)=({G}^{+}(({\widetilde{\square}}_{g}f)_{\text{ext}})|_{\Omega})\\
=({G}^{+}({\square_{g}}f_{\text{ext}}))|_{\Omega}\\
=f_{\text{ext}}|_{\Omega}\\
= f.
\end{equation*}

Property 3: For $f \in H^1_{\text{comp}}(M)$ with $\supp(f)\subseteq {\Omega}$ we have
\begin{multline*}
\supp(\tilde{G}^{+}(f))=\supp\left({G}^{+}(f_{\text{ext}})|_{\Omega}\right)
=\supp({G}^{+}(f_{\text{ext}}))\cap {\Omega}
\subset J_{M}^{+}(\supp(f_{\text{ext}}))\cap \Omega \\
= J_{M}^{+}(\supp(f))\cap \Omega
=  J_{\Omega}^{+}(\supp(f)).
\end{multline*}

\ \hfill $\Box$

\section{Quantisation functors}

In this section we discuss suitable categories and functors as in the smooth case that will allow us to construct the algebra of observables of the quantum theory.

\subsection{The functor {\rm SYMPL} and the categories {\rm GENHYP} and {\rm SYMPLVECT}}

This subsection defines a category based on the analytic results in the previous sections and a functor assigning to each object a symplectic space.

\begin{Definition}\label{genhyp}
Let {\rm GENHYP} denote the category whose objects are 3-tuples $(M,
G^{+}, G^{-})$ where $M$ is a time oriented connected globally
hyperbolic manifold as in Definition \ref{genhyp2} and
$G^{+},G^{-}$ are the unique Green operators of $\square_g$.  Let $X=(M_1, G_1^{+}, G_1^{-})$ and $Y=(M_2,
G_2^{+}, G_2^{-})$ be two objects in {\rm GENHYP}, then $\text{Mor}(X,Y)$
consists of all smooth maps $\iota:M_1\rightarrow M_2$ which are
time-orientation preserving isometric embeddings such that
$\iota(M_1)\subset M_2$ is a causally compatible open subset.
\end{Definition}

\begin{rem}
Theorem \ref{gws} shows that time oriented globally hyperbolic
spacetimes with $C^{1,1}$ metrics are objects in this category.
\end{rem}

Before considering quantisation we prove the following result on compatibility of Green operators.

\begin{Theorem}
Let $M_1$ and $M_2$ be as in Definition \ref{genhyp}, then
  the following diagram commutes:

\begin{center}
\begin{tikzcd}
H_{\text{comp}}^{1}(M_1)\arrow{r}{ext} \arrow{d}[swap]{G_{1}^{\pm}} 
& H_{\text{comp}}^{1}(M_2) \arrow{d}{G_{2}^{\pm}}\\ 
H^{2}_{\text{loc}}(M_1) \arrow[leftarrow]{r}{res} & H^{2}_{\text{loc}}(M_2)
\end{tikzcd}
\end{center}

\end{Theorem}

\begin{proof}
Theorem \ref{res} shows that $\tilde{G}^{\pm}(\phi):=G_{2}^{\pm}(\phi_{\text{ext}})|_{M_{1}}$ is a Green
operator.  By uniqueness, this operator has to be equal to $G^{\pm}_{1}$ and the result follows.
\end{proof}

\begin{rem}
In the smooth setting \cite{bgp} the category {\rm LORFUND} is defined as the
category with objects being $5$-tuples $(M,F,G^{+}, G^-,P)$, where $M$
is a Lorentzian manifold, $F$ is real vector bundle over $M$ with
non-degenerate inner product, $P$ is a formally self-adjoint normally
hyperbolic operator acting on sections in $F$ and $G^{+}$, $G^-$ are the
advanced and retarded Green operators for $P$. The morphisms consist of maps $\iota$ such that
$\iota:M_1\rightarrow M_2$ is a time-orientation preserving
isometric embedding such that $\iota(M_1)\subset M_2$ is a causally
compatible open subset \cite{bgp}.  Moreover, given the condition of
globally hyperbolicity one can form the category {\rm GLOBHYP} where
objects are $3$-tuples $(M,F,P)$, where $M$ is a Lorentzian manifold,
$F$ is real vector bundle over $M$ with non-degenerate inner product,
$P$ is a formally self-adjoint normally hyperbolic operator acting on
sections in $F$. The morphisms are then given by maps
$\iota$ such that $\iota:M_1\rightarrow M_2$ is a time-orientation
preserving isometric embedding such that $\iota(M_1)\subset M_2$ is a
causally compatible open subset. The existence and uniqueness of Green
operators allow us to form a functor from {\rm GLOBHYP} to {\rm LORFUND}
\cite{bgp}.
\end{rem}

We now use the Green operators in order to construct a symplectic vector space. Let $(M, G^{+}, G^{-})$ be an object of {\rm GENHYP} and define $$\tilde{\omega}:H^1_{\text{comp}}(M) \times H^1_{\text{comp}}(M) \rightarrow \mathbb{R}$$
by
$$\tilde{\omega}(\phi,\psi)=\int_{M}G(\phi)\psi \nu_g$$

where $G=G^{+}-G^{-}$ is the causal propagator (see Theorem \ref{exact}). Then $\tilde \omega$ is bilinear and skew-symmetric by Theorem  \ref{greenadj}. However, $\tilde{\omega}$ is degenerate because $\ker(G)$ is nontrivial.
Moreover, using Theorem \ref{exact} we have that $$\ker(G)=\square_{g}V_{0}.$$

Therefore on the quotient space $V(M)=U_0/ \ker(G)=U_0/
\square_{g}V_{0}$  the degenerate form $\tilde{\omega}$ induces a
symplectic form which we denote by $\omega$.

\begin{rem}
It follows from Corollary \ref{sol} that $G$ is continuous so that
$\ker G$ is a closed subspace and hence $V(M)$ is a normed space (and in particular, Hausdorff). See the Discussion section for more details on this point. 
\end{rem}

Finally, we need a functor ${\rm
  SYMPL}:{\rm GENHYP} \rightarrow {\rm SYMPLVECT}$, where {\rm
  SYMPLVECT} is the category whose objects are symplectic vector
spaces   with morphisms given by symplectic maps, i.e., linear maps $A$ such that $\omega_1(f,g)=\omega_2(A f,Ag)$. The following theorem shows the existence of such a functor.

\begin{Theorem}

Let $X=(M_1, G_1^{+}, G_1^{-})$ and  $Y=(M_2, G_2^{+}, G_2^{-})$ be two
objects in {\rm GENHYP} and $f\in Mor(X,Y)$ be a morphism. Then $\text{ext} \colon H^1_{\text{comp}}(M_1) \rightarrow H^1_{\text{comp}}(M_2)$ maps the null space $\ker(G_{1})$ into the
null space $\ker(G_{2})$ and hence induces a continuous
symplectic linear map $V(M_{1})\rightarrow V(M_{2}).$
\end{Theorem}

\begin{proof}

Let $\phi\in \ker(G_{1})$ then $\phi=\square_{g_{1}}\psi$ for some
$\psi\in V_0(M_1)$ where we have used Theorem \ref{exact}.

From the fact that
$G_{2}(\phi_{\text{ext}})=G_{2}((\square_{g_{1}}\psi)_{\text{ext}})=G_{2}\square_{g_{2}}\psi_{\text{ext}}=0$
we see that $\text{ext}(\ker(G_{1}))\subset \ker(G_{2})$. Hence, $\text{ext}$
induces a linear map from $V(M_{1})\rightarrow V(M_{2}) $. Moreover,
for  $\phi,\psi \in H^1_{\text{comp}}(M_1)$ we have on taking representatives
\begin{equation*}
{\omega}_{1}(\phi, \psi)=\int_{M_{1}}G_{1}(\phi)\psi \nu_{g_1}
=\int_{M_{1}}G_{2}(\phi_{\text{ext}})|_{M_{1}}\psi\nu_{g_1}
=\int_{M_{2}}G_{2}(\phi_{\text{ext}}) \psi_{\text{ext}}\nu_{g_2}
={\omega}_{2}(\phi_{\text{ext}}, \psi_{\text{ext}}).
\end{equation*}

Therefore, $\text{ext}$ induces a symplectic map from $V(M_{1})$ to
$V(M_{2})$. This induced map is also continuous (with the respective
quotient topologies), because it is the composition $\pi_2\circ E$ of
two continuous maps, namely $E \colon V(M_1) = H^1_{\text{comp}}(M_1)
/ \ker(G_1) \to H^1_{\text{comp}}(M_2)$, the factor map of
$\text{ext}$, and $\pi_2$, the quotient map $H^1_{\text{comp}}(M_2)
\to H^1_{\text{comp}}(M_2) / \ker(G_2) = V(M_2)$.
\end{proof}

\subsection{The functor {\rm CCR} and the categories $C^{*}$-{\rm ALG} and {\rm QUASILOCALALG}}

In this section we closely follow \cite{bgp} and define the algebraic structures that will be required to represent the observables of the quantum theory. 
The definitions are algebraic in
nature and do not require any further
analytical considerations with respect to the regularity of solutions to the Cauchy problem. Nevertheless, the $C^{1,1}$ causality theory is required and will be mentioned below when it is used.
Another modification with respect the smooth case is that
when considering the symplectic space $(V,\omega)$ in the smooth
theory one has $V(M)={\mathcal D}(M) / {\ker G}$ where
$G=G^+-G^-$ is a map $G: {\mathcal D}(M) \to C^\infty_{sc}(M)$. Employing the short-hand notation $[f]$ for the class $f + \ker(G)$ in $U_0 / \ker(G)$, the
symplectic form is given by
$\omega([f],[h])=(f,Gh)_{L^{2}(M,g)}$ whereas in this
section we will have $V(M)=H^1_{\text{comp}}(M) / {\ker G}$ where now
$G \colon H^1_{\text{comp}}(M) \to V_{\text{sc}}$ and the symplectic form is given by
$\omega([f],[h])=(f,Gh)_{L^{2}(M, g)}$.

\medskip

We now introduce the definition of a Weyl system and a {\rm CCR}-representation of $(V,\omega)$.
\begin{Definition}
A \emph{Weyl system} of the symplectic vector space ($V,\omega$)
consists of a $C^{*}$-algebra ${\cal{A}}$ with unit and a map
$W:V\rightarrow {\cal{A}}$ such that for all $\varphi,\psi\in V$,
\begin{enumerate}
  \item $W(0)=1$,
  \item $W(-\varphi)=W(\varphi)^{*}$,
  \item $W(\varphi)\cdot W(\psi)=e^{-i\omega(\varphi,\psi)/2}W(\varphi+\psi)$.
\end{enumerate}
A Weyl system $({\cal{A}},W )$ of a symplectic vector space
$(V,\omega)$ is called a \emph{CCR-representation} of
$(V,\omega)$ if ${\cal{A}}$ is generated as a $C^{*}$- algebra by the
elements $W(\varphi)$, $\varphi\in V$. In this case we call
${\cal{A}}$ a {\rm CCR}- algebra of $(V,\omega)$ and write it as
$CCR(V,\omega)$.
\end{Definition}

It is always possible to construct a CCR-representation $(\text{CCR}(V,\omega),W)$ for any symplectic vector space $(V,\omega).$ (See \cite[Example 4.2.2]{bgp} ). Moreover, the construction is categorical in the sense that if $(V_{1},\omega_{1})$ and $(V_{2},\omega_{2})$ are two symplectic vector spaces and  $S:V_{1}\rightarrow V_{2}$ is a symplectic linear map. Then, there exist a unique injective $*$-morphism $  \text{CCR}(S):\text{CCR}(V_{1},\omega_{1})\rightarrow \text{CCR}(V_{2},\omega_{2})$


The proof can be found in Corollary 4.2.11 in \cite{bgp}.

From uniqueness of the map $\text{CCR}(S)$   it is possible to define  a functor
$${\rm CCR}:{\rm SYMPL} \rightarrow C^{*}{\rm-ALG}$$
where $C^{*}{\rm-ALG}$ is the category whose objects are $C^{*}$-algebras and whose morphisms are injective unit preserving *-morphisms.

A set  $I$ is called a directed set with orthogonality relation, if it carries a partial order $\le$ and a symmetric relation $\perp$ between its  elements such that:

\begin{enumerate}
\item for all $\alpha,\beta\in I$ there exists a $\gamma\in I$ with $\alpha\le\gamma$ and $\beta\le \gamma$,
\item for every $\alpha\in I$ there is a $\beta\in I$ with $\alpha\perp\beta$,
\item $\alpha\le \beta$ and $\beta\perp\gamma$, then $\alpha\perp\gamma$,
\item if $\alpha\perp\beta$ and $\alpha\perp\gamma$, then there exists a $\delta\in I$ such that $\beta\le\delta,\gamma\le\delta$ and $\alpha\perp\delta$.
\end{enumerate}

Sets of this type allow to define the objects and morphisms of the category QUASILOCALALG.

\begin{Definition}

The objects of the category QUASILOCALALG are bosonic quasi-local $C^{*}$-algebras which are pairs $({\cal{U}},\{{\cal{U}_{\alpha}}\}_{\alpha\in I})$ of a $C^{*}$- algebra ${\cal{U}}$ and a family $\{ {\cal{U}_{\alpha}}\}_{\alpha\in I}$ of $C^{*}$-subalgebras, where $I$ is a directed set with orthogonality relation such that the following holds:
\begin{enumerate}
\item$ {\cal{U}_{\alpha}}\subset {\cal{U}_{\beta}}$ whenever $\alpha\le\beta$,
\item ${\cal{U}}=\overline{\bigcup_{\alpha}{\cal{U}_{\alpha}}}$,
\item The algebras ${\cal{U}_{\alpha}}$ have a common unit 1,
\item If  $\alpha\perp\beta$, then the commutators of elements from ${\cal{U}_{\alpha}}$ with those of ${\cal{U}_{\beta}}$ are trivial.
\end{enumerate}

A morphism  between two quasi-local $C^{*}$-algebras  $({\cal{U}},\{{\cal{U}_{\alpha}}\}_{\alpha\in I})$ and  $({\cal{V}},\{{\cal{V}_{\beta}}\}_{\beta\in J})$ is defined as a pair $(\varphi, \Phi)$ where $\Phi:{\cal{U}}\rightarrow {\cal{V}}$ is a unit-preserving $C^{*}$-morphism and $\varphi:I\rightarrow J$ is a map such that
\begin{enumerate}
\item $\varphi$ is monic, i.e., if $\alpha_1 \le \alpha_2$ in $I$ then $\varphi(\alpha_1)\le\varphi(\alpha_2)$ in $J$,
\item $\varphi$ preserves orthogonality, i.e., if $\alpha_1\perp \alpha_2$ in $I$, then $\varphi(\alpha_1)\perp \varphi(\alpha_2)$,
\item $\Phi({\cal{U}})\subset {\cal{V}}_{\varphi(\alpha)}$ for all $\alpha\in I$.
\end{enumerate}
\end{Definition}

In the remainder of this section we discuss a functor from {\rm GENHYP} to {\rm
  QUASILOCALALG}. Let $(M, G^+, G^-)$ be an object in  {\rm GENHYP} and
$$I=:\{O\subset M| O \mbox{ is open, relatively compact, causally compatible, globally hyperbolic}\}\cup\{\emptyset,M\}.$$

The relation $O \perp O'$ means that $O$ and $O'$ are causally independent, i.e., there is no causal curve connecting a point in $\overline{O}$ to a point in $\overline{O'}$.

\begin{rem}
  The proof that the set $I$ is a directed set with orthogonality
  relation requires results from causality theory in a low regularity
  setting \cite{cg,saemann,KSSV} to obtain Lemma A.5.11 in \cite{bgp}
  and the existence of smooth time functions \cite{BernardSuhr}
  \cite{MinguzziAddenda} to obtain Proposition A.5.13 in
  \cite{bgp}. Properties 1 and 2 follow upon taking $\alpha=M$,
  $\beta=\emptyset$. Property 3 follows from the observation that
  $O\subset O' $ implies $J(O)\subset J(O')$, and Property 4 is
  implied by Lemma 4.4.8 in \cite{bgp} with the appropriate
  modifications of Lemma A.5.11 and Proposition A.5.13 therein.
\end{rem}

For any non-empty set $O\in I$ take the restriction of the operator $\square_g$ to the region $O$. Due to causal compatibility of $O\subset M$ the restriction of Green operators $G^{+},G^{-}$ to the region $O$ yield Green operators $G_{O}^{+},G_{O}^{-}$. Therefore, we get an object $(O, G_{O}^{+},G_{O}^{-})$ for each $O\neq\emptyset\in I$. For $\emptyset\neq O_1\subset O_2$ the inclusion induces a morphism $\iota_{O_2,O_1}$ in the category {\rm GENHYP}. This morphism is given by the embedding $O_1\rightarrow O_2$. Let $\alpha_{O_2, O_1}$ denote the morphism $\text{CCR}\circ \text{SYMPL} (\iota_{O_2,O_1})$ in $C^{*}$-ALG and  recall that $\alpha_{O_2, O_1}$ is an injective unit preserving $*$-morphism. 

We set for $\emptyset\neq O\in I$,
$$(V_O, \omega_O):={\rm SYMPL}(O, G_{O}^{+},G_{O}^{-}),$$
and for $O\in I , O\neq\emptyset,M$,
$$ {\cal{U}}_{O}:=\alpha_{M,O}(\text{CCR}(V_O,\omega_O)),$$
for $O=M$ define 
$$ {\cal{U}}_{M}:=C^{*}\left({\cal{U}}_{O\in I , O\neq\emptyset,M}\right)$$
which is the algebra of $\text{CCR}(V_M,\omega_M)$ generated by all the ${\cal{U}}_{O}$;
for $O=\emptyset$, set ${\cal{U}}_{\emptyset}=\mathbb{C}$.

Now we assign to any morphism  in {\rm GENHYP} a
morphism between quasi-local algebras in {\rm QUASILOCALALG}: Consider a  morphism $\iota:(M,G^{+},G^{-})\rightarrow (N,\tilde{G}^{+},\tilde{G}^{-})$ in {\rm GENHYP}. Let $I_1,I_2$ denote the index sets associated to $M,N$ respectively. We define a map $\varphi:I_1\rightarrow I_2$ by $M\rightarrow N$ and $O_1\rightarrow \iota(O_1)$ if $O_1\neq M$. Since $\iota$ is an embedding such that $\iota(M)\subset N$ is causally compatible, the map $\varphi$ is monotonic and preserves causal independence. Therefore, $(\varphi,\Phi)$ with $\Phi=\text{CCR}\circ \text{SYMPL}(\iota)$  is the required morphism. 
To be precise we have the following result.

\begin{Theorem}\label{quant}
The assignment $ (M,G^{+},G^{-})\rightarrow ( {\cal{U}}_{M}, \{{\cal{U}}_{O}\}_{O\in I})$ and $\iota\rightarrow (\varphi, \Phi)$ yields a functor {\rm QUANT} from {\rm GENHYP} to {\rm QUASILOCALALG}.
\end{Theorem}

\begin{proof}
A detailed proof can be found in \cite[Lemma 4.4.10, Theorem 4.4.11 and Lemma 4.4.13]{bgp}.
\end{proof}

\begin{rem}\label{axioms}
In the low regularity setting the proof above requires one to consider elements  $\phi\in H^1_{\text{comp}}(O)$ rather than  $\phi\in \mathcal D(O)$ in Lemma 4.4.10 and  the  low regularity quotient $H^1_{\text{comp}}(M)/\ker(G)$ instead of $\mathcal D(M)/\ker (G)$ with  $C^{1,1}$ causality theory in  Lemma 4.4.13 in \cite{bgp}.
\end{rem}

\subsection{The Haag-Kastler axioms }

In this subsection we show that the functor $\rm QUANT$  given by Theorem \ref{quant} satisfies the Haag-Kastler axioms.

\begin{Theorem}\label{quantum}
The functor ${\rm QUANT}:{\rm GENHYP}\rightarrow{\rm QUASILOCALALG}$ satisfies the Haag-Kastler axioms, i.e., for every object $(M,G^{+},G^{-})$ in {\rm GENHYP} the corresponding quasi-local $C^{*}$-algebra $({\cal{U}}_{M},\{\cal{U}_{O}\}_{O\in I})$ satisfies:

\begin{enumerate}
\item If $O_1\subset O_2$ then ${ {\cal{U}}_{O_1}}\subset {{\cal{U}}_{O_2}}$ for all $O_1,O_2 \in I$. 
\item ${\cal{U}}_{M}=\overline{\bigcup_{O\in I,  O\neq M,\emptyset}{\cal{U}_{O}}}$.
\item ${\cal{U}}_{M}$ is simple.
\item The ${\cal{U}}_{O}$'s have a common unit 1.
\item For all $O_1,O_2\in I$ with $J(\overline{O_1})\cap \overline{O_2}=\emptyset$ the subalgebras ${\cal{U}}_{O_1},{\cal{U}}_{O_2}$ commute.
\item (Time-slice axiom) Let $O_1\subset O_2$ be nonempty element of $I$ admitting a common  smooth spacelike Cauchy hypersurface, then ${\cal{U}}_{O_1}={\cal{U}}_{O_2}$.
\item Let $O_1,O_2\in I$ and let the Cauchy development $D(O_2)$ be relatively compact in $M$. If $O_1\subset D(O_2)$, then $ {{\cal{U}}_{O_1}}\subset {{\cal{U}}_{O_2}}$.
\end{enumerate}
\end{Theorem}
 
The proof will be based on the  following two lemmas. 

\begin{Lemma}\label{timefun}
Let $O$ be a causally compatible globally hyperbolic open subset of a globally hyperbolic manifold $M$. Assume there exists a  smooth spacelike Cauchy hypersurface $\Sigma$ of $O$ which is also a Cauchy hypersurface of $M$, let $h$ be a  smooth Cauchy time-function on $O$ and $K\subset M$ be compact. Assume that there exists $t\in \mathbb{R}$ with $K\subset I^{+}(h^{-1}(t))$.
Then there is a smooth function $\rho:M\rightarrow [0,1]$ such that 

\begin{enumerate}
\item $\rho=1$ on a neighbourhood of $K$,
\item $\supp(\rho)\cap J^{-}(K)\subset M$ is compact, and
\item $\overline{\{x\in M|0\le\rho(x)\le 1\}}\cap  J^{-}(K)$ is compact and contained in $O$.
\end{enumerate}
\end{Lemma}

The proof of Lemma \ref{timefun} in the $C^{1,1}$ setting can be
carried out following that of \cite{bgp} with suitable modifications
using results of low regularity causality theory \cite{cg, saemann,
  KSSV}. In particular, the proof uses the facts that, the causal
relation is closed, that if $S,S_t$ are Cauchy hypersurfaces of $O$
and $S$ is also a Cauchy hypersurface of $M$, then $S_t$ is a Cauchy
hypersurface of $M$ and the existence of Cauchy hypersurfaces in
globally hyperbolic spacetimes.

\begin{Lemma}\label{timeaxiom}
Let $(M,G^{+},G^{-})$ be an object of {\rm GENHYP} and $O$ be a causally compatible globally hyperbolic open subset of $M$. Assume that there exists a Cauchy hypersurface $\Sigma$ which is also a Cauchy hypersurface of $M$. Let $\varphi \in U_0$,
then there exist $\chi\in V_{0}$ and $\psi\in U_0$ such that $\supp(\psi)\subset O$ and $\varphi=\psi+\square_g\chi$.
\end{Lemma}

\begin{proof}
Let $h$ be a Cauchy time function on $O$. Fix $t_- \le t_+$ in the range of $h$. Then the subsets $\Sigma_{-}=h^{-1}(t_{-})$,$\Sigma_{+}=h^{-1}(t_{+})$ are Cauchy hypersurfaces of $M$. Hence every inextendable curve timelike curve in $M$ meets $\Sigma_{-},\Sigma_{+}$. Since $t_{-}\le t_{+}$, the set $\{I^{+}(\Sigma_{-}),I^{-}(\Sigma_{+})\}$ is a finite open cover of $M$. Let $\{f_{+},f_{-}\}$ be a smooth partition of unity subordinated to this cover.  In particular, $\supp(f_{\pm})\subset I^{\pm}(\Sigma_{\mp})$. Set $K_{\pm}:=\supp(f_{\pm}\varphi)=\supp(\varphi)\cap\supp(f_{\pm})$. Then $K_{\pm}$ is a compact subset of $M$ satisfying $K_{\pm}\subset I^{\pm}(\Sigma_{\mp})$. Applying Lemma \ref{timefun} we obtain two smooth functions $\rho_{\pm}:M\rightarrow [0,1]$ satisfying
\begin{enumerate}
\item $\rho_{\pm}=1$ on a neighbourhood of $K_{\pm}$,
\item $\supp(\rho_{\pm})\cap J^{\mp}(K_{\pm})\subset M$ is compact, and
\item $\overline{\{x\in M|0\le\rho_{\pm}(x)\le 1\}}\cap  J^{\mp}(K_{\pm})$ is compact and contained in $O$.
\end{enumerate}

Set $\chi_{\pm}:=\rho_{\pm}G^{\mp}(f_{\pm}\varphi),\chi:=\chi_{+}-\chi_{-}$ and $\psi:=\varphi-\square_g\chi$. Since $\supp(G^{\mp}(f_{\pm}\varphi))\subset J^{\mp}(K_{\mp})$, the support of $\chi_{\pm}$ is contained in $\supp(\rho_{\pm})\cap J^{\mp}(K_{\pm})$ which is compact by the second property of $\rho_{\pm}$.
Since $\rho_\pm$ and $f_\pm$ are smooth by construction, we have  $\chi_{\pm}\in H^2_{\text{comp}}(M)$. Moreover,
\begin{align*}
\square_g\chi_{\pm}&= G^{\mp}(f_{\pm}\varphi)\square_g\rho^{\pm}+g^{\alpha\beta}\partial_\alpha\rho^{\pm}\partial_\beta G^{\mp}(f_{\pm}\varphi)+\rho^{\pm} \square_g G^{\mp}(f_{\pm}\varphi)\\
&=G^{\mp}(f_{\pm}\varphi)\square_g\rho^{\pm}+g^{\alpha\beta}\partial_\alpha\rho^{\pm}\partial_\beta G^{\mp}(f_{\pm}\varphi)+\rho^{\pm} (f_{\pm}\varphi),
\end{align*}
which implies $\square_g\chi_{\pm}\in H^1_{\text{loc}}(M)$. Notice that $\square_g \rho^\pm$ is not smooth but $C^{0,1}$.

Now $\psi$ is the difference of $H^1_{\text{loc}}(M)$ functions so  it remains to show that $\supp(\psi)$ is compact and contained in $O$. By the first property of $\rho_{\pm}$, one has $\chi_{\pm}:=G^{\mp}(f_{\pm}\varphi)$ in a neighbourhood of $K_{\pm}$. Moreover, $f_{\pm}\varphi=0$ on $\{\rho_{\pm}=0\}$. Hence, $\square_g\chi_{\pm}=f_{\pm}\varphi$ on $\{\rho_{\pm}=0\}\cup\{\rho_{\pm}=1\}$. Therefore, $f_{\pm}\varphi-\square_g \chi_{\pm}$ vanishes outside $\overline{\{x\in M|0\le\rho_{\pm}(x)\le 1\}}$, i.e., $\supp(f_{\pm}-\square_g \chi_{\pm})\subset \overline{\{0\le\rho_{\pm}(x)\le 1\}}$. By the definitions of $\chi_{\pm},f_{\pm}$ one also has $\supp(f_{\pm}\varphi-\square_g \chi_{\pm})\subset J^{\mp}(K_{\pm})$, hence 
 $\supp(f_{\pm}\varphi-\square_g \chi_{\pm})\subset J^{\mp}(K_{\pm})\cap \overline{\{0\le\rho_{\pm}(x)\le 1\}}$ which is compact and contained in $O$ by the third property of $\rho_{\pm}$.  Therefore, $\psi\in U_0$ with  $\supp(\psi)\subset O$. Moreover, $\varphi-\psi=\square_g\chi\in U_0$ which gives $\chi\in V_0$. 
\end{proof}

\begin{proof}[Proof of Theorem \ref{quantum}]

The first, fourth and fifth axiom follow from the
  definition of the quasi-local $C^{*}$-algebra, the definition of the
  set $I$ and \cite[Lemma 4.4.10]{bgp}. The second axiom follows from
  \cite[Lemma 4.4.13 ]{bgp} and the third axiom follows from
  \cite[Remark 4.5.3]{bgp}.  Remark \ref{axioms} mentions the
  necessary modifications of those Lemmas in the $C^{1,1}$ setting.

It therefore remains to prove the time-slice axiom. 
Let $O_1\subset O_2$ be nonempty casually compatible globally hyperbolic subsets of $M$ admitting a common smooth spacelike Cauchy hypersurface $\Sigma$.
Let $[\phi]\in V(O_2)$. Then Lemma \ref{timeaxiom} applied to $M:=O_2$ and $O=O_1$ yields $\chi\in V_{0},\psi\in U_{0}$ such that 
$\phi=\psi_{\text{ext}}+\square_g \chi$ with $\supp(\psi)\subset O_1$. Since, $\square_g\chi\in \ker(G_{O_2})$ we have $[\phi]=[\psi_{\text{ext}}]$, that is, $[\phi]$ is the image of the symplectic linear map $V(O_1) \rightarrow V(O_2)$ induced by the inclusion $\iota: O_1\rightarrow O_2$. Therefore, the map is surjective, and hence an isomorphism of symplectic topological vector spaces. This isomorphism functorially induces an isomorphism of $C^*$-algebras, hence ${\cal{U}}_{O_1}={\cal{U}}_{O_2}$. This proves the time-slice axiom.  
Finally, the seventh axiom can be deduced from the first and the sixth axiom \cite[Theorem 4.5.1]{bgp}.
\end{proof}

\section{Discussion}

In this paper we have constructed Green operators for spacetime metrics of regularity $C^{1,1}$.
The function spaces for the domain and range of the Green operators play a fundamental role in low regularity spacetimes and 
our choices for these spaces were motivated by the following two requirements: Global well posedness of the Cauchy problem and employing Sobolev spaces, such as $H^k_{\text{loc}}(M)$ and $H^k_{\text{comp}}(M)$ ($k \in \N_0$), that do not depend on a Riemannian background metric.  
We have shown that the quotient space $V(M) = U_0/\square_g V_0$ can be used to construct quasi-local $C^*$-algebras that satisfy the Haag-Kastler axioms, so that in a quantum theoretic setting the self-adjoint elements in these $C^*$-algebras can be associated with the observables of the theory. 

\subsection{Topological Issues}

Let us describe the quotient vector space $U_0/\square_g V_0$ in some
more detail for the \emph{globally hyperbolic case}, where we have
$\text{ker}(G)=\text{im}(\square_g)$ as a consequence of the spectral
sequence given in Theorem \ref{exact}, thus $U_0/\square_g V_0 = U_0 /
\text{ker}(G)$ in this case.  Recall that $G$ is a linear map $U_0 \to
V_{\text{sc}}$ and let $G_0$ denote the associated map from the
quotient $U_0 / \text{ker}(G)$ to $\text{im}(G) \subseteq
V_{\text{sc}}$, defined by $G_0 (\phi + \text{ker}(G)) := G \phi$ for
every $\phi \in U_0$. Therefore, $G_0$ is linear and bijective by
construction and we arrive at the following chain of (algebraic)
isomorphisms of vector spaces
\begin{equation}\label{isochain}
   U_0/\square_gV_0 = U_0 / \text{ker}(G) \cong \text{im}(G)=\text{ker}(\square_g)\subseteq V_{\text{sc}}.
\end{equation}
Recall that the analogue of \eqref{isochain} in the smooth globally hyperbolic case, as discussed in \cite{bgp}, is
$$
  \mathcal{D}(M)/\square_g \mathcal{D}(M) = \mathcal{D}(M) / \text{ker}(G) \cong \text{im}(G) = \text{ker}(\square_g)
	\subseteq C^\infty_{\text{sc}},
$$
showing also that the quotient is isomorphic to the space of solutions to the homogeneous wave equation. 

The question arises whether the isomorphism in the middle part of \eqref{isochain}, obtained via the factored map $G_0$, is topological, where the quotient $U_0 / \text{ker}(G)$ is equipped with the finest topology such that the canonical surjection $\pi \colon U_0 \to U_0 / \text{ker}(G)$, $\phi \mapsto \phi + \text{ker}(G)$ is continuous. Note that by continuity of $G$ we have that $\text{ker}(G)$ is closed in the normed space $U_0$, hence $U_0 / \text{ker}(G)$ is a normed space (in particular, Hausdorff). Furthermore, $G_0$ is continuous by construction and the continuity of $G$, thus it remains to be checked whether the inverse of $G_0$ is continuous, or, equivalently, whether $G_0$ is an open map. 

\begin{rem} We note that by \cite[Chapter III, Proposition 1.2]{Schaefer}, the factored map $G_0$ is a topological isomorphism if and only if $G$ is open as a map from $U_0$ to $\text{im}(G)$ (with the relative topology on the latter). In case of Fr\'{e}chet spaces such a property for $G$ could be deduced conveniently via an open mapping principle or from a closed image criterion, but observe that neither $U_0 = H^1_{\text{comp}}(M)$ nor $\text{im}(G)$ is complete (with respect to the metric inherited from the Banach space $H^1(M)$ and the Fr\'{e}chet space $H^2_{\text{loc}}(M)$, respectively).
\end{rem}

We choose a finer topology $\sigma$ on $V_{\text{sc}}$ to make $\square_g \colon (V_{\text{sc}},\sigma) \to H^1_{\text{loc}}(M)$ continuous by adding the seminorms $p_{\chi}(\phi) := \| \chi \cdot \square_g \phi \|_{H^1}$ ($\chi \in \mathcal{D}(M)$) to those on $V_{\text{sc}}$ inherited from $H^2_{\text{loc}}(M)$. Note that this has no effect on the subspace $\text{im}(G) \subseteq V_{\text{sc}}$, since $\text{im}(G) \subseteq \text{ker}(\square_g)$ in the complex of maps in Theorem \ref{exact} (even equality holds due to global hyperbolicity). In fact, $\sigma$ is precisely the coarsest topology that is finer than the $H^2_{\text{loc}}(M)$-topology on $V_{\text{sc}}$, which we denote by $\tau_2$, and renders $\square_g$ continuous as a map $V_{\text{sc}} \to H^1_{\text{loc}}(M)$, i.e., $\sigma$ is the supremum (in the lattice of topologies on $V_{\text{sc}}$) of $\tau_2$ and the initial (projective) topology $\tau_1$ with respect to $\square_g$. Therefore, we have continuity of $G \colon U_0 \to (V_{\text{sc}},\sigma)$, since $G$ is continuous $U_0 \to (V_{\text{sc}},\tau_2)$ by Corollary \ref{sol} and also continuous $U_0 \to (V_{\text{sc}},\tau_1)$ due to the obvious continuity of $\square_g \circ G = 0$ from $U_0$ into  $H^1_{\text{loc}}(M)$.

\begin{Lemma}\label{G_0}
The inverse of $G_0 \colon U_0 / \text{ker}(G) \to \text{im}(G)$, $\phi + \text{ker}(G) \mapsto G \phi$, is continuous.
\end{Lemma}

\begin{proof} We will show that $G_0^{-1}$ can be written as the composition $G_0^{-1} = \pi \circ P \circ Z$ of three continuous linear maps. The map $\pi \colon U_0 \to U_0 / \text{ker}(G)$ is the canonical surjection, which is continuous by construction. It remains to construct suitable continuous maps $P$ and $Z$ with $P \circ Z \colon \text{im}(G) \to U_0$ and such that $G_0 \circ \pi \circ P \circ Z = \text{id}_{\text{im}(G)}$ and $\pi \circ P \circ Z \circ G_0 = \text{id}_{U_0 / \text{ker}(G)}$. 

Let 
$V^\pm_{\text{sc}} := \{ \phi \in V_{\text{sc}} \mid \supp(\phi) \subseteq J^\pm(K) \text{ for some compact subset } K \subseteq M\}$ and define the subspace
$$
    W := \{ (\phi_-,\phi_+) \in V^-_{\text{sc}} \times V^+_{\text{sc}} \mid \phi_- + \phi_+ \in \text{ker} (\square_g) \} 
		\subseteq V_{\text{sc}} \times V_{\text{sc}},
$$
which we equip with the trace of the product topology stemming from $\sigma$.

\emph{Construction of $P$:} 
We consider $P \colon W \to U_0$, given by $P (\phi_-,\phi_+) := (\square_g \phi_+ - \square_g \phi_-)/2$. Note that  a priori, $P (\phi_-,\phi_+)$ is only in $H^1_{\text{loc}}(M)$ and we have to show that
$P (\phi_-,\phi_+)$ has compact support, thus belongs to $U_0 = H^1_{\text{comp}}(M)$.
To prove this, observe that $\phi = \phi_- + \phi_+ \in \text{ker}(\square_g)$ implies $\square_g \phi_- = - \square_g \phi_+$, hence $P \phi = \square_g \phi_+ = - \square_g \phi_-$. Let $K_-$ and $K_+$ be compact subsets of $M$ with $\supp (\phi_\pm) \subseteq J^\pm(K_\pm)$, then we have $\supp(P \phi) \subseteq \supp(\phi_-) \cap \supp(\phi_+) \subseteq J^-(K_-) \cap J^+(K_+)$, where $J^-(K_-) \cap J^+(K_+)$ is compact by global hyperbolicity (\cite[Lemma A.5.7]{bgp}).
The continuity of $P$ is clear by construction of the topology $\sigma$.

\emph{Construction of $Z$:} As a preparation we will first construct two continuous maps $S^\pm \colon V_{\text{sc}} \to V^\pm_{\text{sc}}$, such that $\phi = S^- \phi + S^+ \phi$ holds for every $\phi \in V_{\text{sc}}$ and, moreover,
\begin{equation}\label{claimGS}
   G \,\square_g S^\pm \phi = \pm \phi \qquad \forall \phi \in \text{im}(G).
\end{equation}
Choose a smooth spacelike Cauchy surface $\Sigma \subseteq M$ and let $t \colon M \to \R$  be a smooth temporal function such that $\Sigma = t^{-1}(0)$ (compare the earlier discussion on causality in $C^{1,1}$-spacetimes). We obtain an open covering of $M$ by the two sets $O_- := \{ x \in M \mid t(x) < 1 \}$ and $O_+ := \{ x \in M \mid t(x)  > -1 \}$ and choose a subordinate partition of unity $\chi_-, \chi_+ \in C^\infty(M)$, i.e., $\supp(\chi_\pm) \subseteq O_\pm$ and $\chi_- + \chi_+ = 1$. We define $S^\pm \phi := \chi_\pm \phi$, then the relation $\phi = S^- \phi + S^+ \phi$ holds by construction and the continuity of $S^\pm$ is clear from continuity of multiplication by fixed smooth functions with respect to (localised) Sobolev norms. It remains to show that $S^\pm \in V^\pm_{\text{sc}}$ for every $\phi \in V_{\text{sc}}$ and  Equation \eqref{claimGS} is true.

Let $\phi \in V_{\text{sc}}$ and $K \subseteq M$ be compact such that
$\supp(\phi) \subseteq J^-(K) \cup J^+(K)$. Then $\supp(\chi_+ \phi)
\subseteq O_+ \cap (J^-(K) \cup J^+(K)) \subseteq (O_+ \cap J^-(K))
\cup J^+(K)$.  Note that $O_+ \cap J^-(K)$ is relatively compact by
\cite[Corollary A.5.4]{bgp}, since $O_+ \subseteq J^+(\Sigma_{-})$
holds with $\Sigma_{-} := t^{-1}(- 1)$ (note that the time function is
strictly increasing along causal curves).  Therefore, with some
compact set $K_+$ containing $K$ as well as $O_+ \cap J^-(K)$ we
obtain $\supp(\chi_+ \phi) \subseteq J^+(K_+)$, thus $S^+ \phi \in
V^+_{\text{sc}}$. The reasoning for $S^- \phi \in V^-_{\text{sc}}$ is
analogous.

For the proof of \eqref{claimGS} we start by noting that $\phi \in \text{im}(G) = \text{ker} (\square_g)$ implies $0 = \square_g \phi = \square_g S^- \phi + \square_g S^+ \phi$, so that the part with $S^-$ in \eqref{claimGS} follows once the equation for $S^+$ is shown. Recall that we have $\square_g S^+ \phi = - \square_g S^- \phi \in H^1_{\text{comp}}(M)$ from the reasoning in the construction of $P$ above. Moreover, for every test function $\psi$ on $M$ we have that $\supp(G^\mp \psi) \cap \supp(S^\pm \phi) \subseteq J^\mp(\supp(\psi)) \cap J^\pm(K)$ for some compact set $K$, hence global hyperbolicity guarantees that the supports of $G^- \psi$ and $S^+ \phi$ as well as those of $G^+ \psi$ and $S^- \phi$ always have compact intersection. To summarise, we may apply Theorem \ref{greenadj} and Lemma \ref{selfadj} to obtain the following chain of weak equalities
\begin{multline*}
   (\psi, G^\pm \square_g S^+ \phi)_{L^2(M,g)} = ( G^{\mp} \psi, \square_g S^+ \phi)_{L^2(M,g)} =  
	(G^{\mp} \psi, \pm \square_g S^\pm \phi)_{L^2(M,g)}\\ = ( \square_g G^{\mp} \psi, \pm S^\pm \phi)_{L^2(M,g)} = (\psi, \pm S^\pm \phi)_{L^2(M,g)},
\end{multline*} 
which implies $G^\pm \square_g S^+ \phi = \pm S^\pm \phi$ and therefore $G \,\square_g S^+ \phi = G^+ \square_g S^+ \phi - G^- \square_g S^+ \phi = S^+ \phi - (-  S^- \phi) = S^+ \phi + S^- \phi = \phi$. Thus, Equation \ref{claimGS} is proved and concludes the preparatory construction of $S^\pm$.  

Finally, we turn to the definition of the map $Z$. 
Observe that $\phi \in \text{im}(G) = \text{ker}(\square_g) \subseteq V_{\text{sc}}$ implies $(S^- \phi, S^+ \phi) \in W$, which allows  to set $Z \phi := (S^- \phi, S^+ \phi)$ for every $\phi \in \text{im}(G)$  and obtain a continuous linear map $Z \colon \text{im}(G) \to W$.

\medskip

We complete the proof by showing that $\pi \circ P \circ Z$ is the inverse of $G_0$.

$\bullet$ The relation $G_0 \circ \pi \circ P \circ Z = \text{id}_{\text{im}(G)}$ holds, since for every $\phi \in \text{im}(G)$ we have 
\begin{multline*}
  G_0(\pi(P(Z \phi))) = G_0(\pi(P(S^-\phi,S^+\phi))) = \frac{1}{2} G_0(\pi(\square_g S^+ \phi - \square_g S^- \phi))\\ = 
	\frac{1}{2} G_0(\square_g S^+ \phi - \square_g S^- \phi + \text{ker}(G)) = 
	\frac{1}{2} (G \,\square_g S^+ \phi - G\,\square_g S^- \phi),
\end{multline*}
where we may apply \eqref{claimGS} to rewrite the last term as $\frac{1}{2}(\phi + \phi) = \phi$.

$\bullet$ We finally show that the equation $\pi \circ P \circ Z \circ G_0 = \text{id}_{U_0 / \text{ker}(G)}$ is true. Let $f \in U_0$, then
\begin{multline*}
  \pi(P(Z(G_0(f + \text{ker}(G))))) = \pi(P(Z(Gf))) = \pi(P(S^- Gf, S^+ Gf))\\ = \pi(\frac{1}{2}(\square_g S^+ Gf - \square_g S^- Gf))
	= \frac{1}{2}(\square_g S^+ Gf - \square_g S^- Gf) + \text{ker}(G)
\end{multline*}
and in the last term we may replace $f_1 := \frac{1}{2}(\square_g S^+ Gf - \square_g S^- Gf)$  by  $f$, since thanks to \eqref{claimGS} the difference is in the kernel: $G(f_1 - f) = \frac{1}{2}(Gf + Gf) - Gf = Gf - Gf = 0$.
\end{proof}

\begin{Proposition} For a globally hyperbolic $C^{1,1}$ spacetime $(M,g)$, we obtain a topological isomorphism $U_0 / \text{ker}(G) \cong \text{im}(G)$ according to \eqref{isochain}, where $V_{\text{sc}}$ carries the topology $\sigma$.
\end{Proposition}

\begin{rem}  We are not using an inductive limit construction for the topology on $V_{\text{sc}}$ as, e.g., in \cite{barwafo}, because we preferred to stay with questions of convergence and continuity in the simpler realm of local Sobolev norms. Moreover, in the above context, we would otherwise not have a topological isomorphism of $\text{im}(G)$ with $U_0 / \text{ker}(G)$, since we decided coherently that $U_0$ should inherit the norm topology from $H^1(M)$, thus rendering $U_0 = H^1_{\text{comp}}(M)$ normed, but incomplete. However, the basic constructions of quantisation  for the associated symplectic (quotient) vector spaces do not require completeness.
\end{rem}

\subsection{An equivalent symplectic structure }

An analogous construction of the {\rm CCR} representation can be
achieved using a symplectic structure on the vector space of solutions
to the homogeneous problem parametrised by their initial data
\cite{wald}.  In that context, one defines a symplectic structure
$\Xi$ on $\text{ker}(\square_g)$ given by $\Xi(\phi,\psi)=\int_\Sigma
(u_1v_0-v_1u_0) \mu_h$ where $(u_0,u_1),(v_0,v_1) $ are compactly
supported smooth initial data induced by the smooth solutions
$\phi,\psi$ respectively on the Cauchy hypersurface
$\Sigma$. Moreover, the symplectic structure and the Weyl system
generated by it is independent of the chosen Cauchy hypersurface and
it is isomorphic to the Weyl system generated by $(U_0 /
\text{ker}(G), \omega)$ \cite{wald,dim80}.

In the $C^{1,1}$ setting the above construction remains true with suitable modifications.

To be precise, using Theorem \ref{exact} we know that $\ker(\square_g)= \text{im}(G)$. Moreover, for any smooth spacelike Cauchy hypersurface $\Sigma$, if $\phi=G(f)$ then $\phi|_\Sigma\in H^2_{\text{comp}}(\Sigma)$ and  $\nabla_n \phi|_\Sigma\in H^1_{\text{comp}}(\Sigma)$. This follows from the observation that $\phi \in V_{\text{sc}}$ and is the difference of two solutions to the Cauchy problem with zero initial data, which by Theorem \ref{globalexist} belong to the space $C^0(\R, H^2(\Sigma_t))\cap C^1(\R, H^1(\Sigma_t))$. Therefore, given any smooth spacelike Cauchy hypersurface $\Sigma$, we define for $\phi, \psi \in \ker\square_g$ with  $u_0 := \phi|_\Sigma$, $u_1 := \nabla_n \phi|_\Sigma$, $v_0 := \psi|_\Sigma$, $v_1 := \nabla_n \psi|_\Sigma$, hence $(u_0,u_1), (v_0,v_1) \in H^2_{\text{comp}}(\Sigma)\times H_{\text{comp}}^1(\Sigma)$, the skew-symmetric bilinear form
$$
  \Xi_{\Sigma}(\phi,\psi)=\int_\Sigma (u_1v_0-v_1u_0) \mu_h.
$$
  It follows from linearity, the uniqueness of solutions to the Cauchy problem, and direct computations that  $\Xi$ is symplectic where to show non-degeneracy one tests with elements  of the form $(0,u_1)$ and $(v_0, 0)$, i.e., with  $u_0 = 0$ and $v_1=0$ and employ uniqueness in the Cauchy problem (cf.\ \cite{dim80}).

  We show that $\Xi_{\Sigma}$ does not depend on $\Sigma$: This
  follows from the divergence theorem in a region bounded by two
  Cauchy hypersurfaces $\Sigma_1$, $\Sigma_2$ and the conservation of
  the current
  $j^{\mu}(\phi,\psi)=g^{\mu\nu}\left(\phi\nabla_{\nu}\psi-\psi\nabla_{\nu}\phi\right)$. Explicitly we have for any $\phi,\psi\in \ker(\square_g)$
\begin{equation*}
\int_{J^-(\Sigma_1)\cap J^+(\Sigma_2)} \mathrm{div}( j^\mu(\phi,\psi)) \nu_g=0
\end{equation*}
and 
\begin{equation*}
\int_{J^-(\Sigma_1)\cap J^+(\Sigma_2)} \mathrm{div}( j^\mu(\phi,\psi)) \nu_g=\int_{\Sigma_2} j^\mu(\phi,\psi)n_\mu \mu_{h_1}-\int_{\Sigma_1} j^\mu(\phi,\psi)n_\mu \mu_{h_2}=0.
\end{equation*}

Therefore,
\begin{equation*}
\Xi_{\Sigma_1}(\phi,\psi)=\Xi_{\Sigma_2}(\phi,\psi),
\end{equation*}
so we will drop the $\Sigma$ from the notation of $\Xi$. Notice that the $H^2_{\text{loc}}$ regularity is required in order to make sense of the divergence of the current.


Finally, we  show that the linear bijective factor map $G_0$ of $G$, as defined  before \eqref{isochain}, provides a symplectic map  from $(U_0 / \text{ker}(G), \omega)$ to $(\ker(\square_g),\Xi)$.

 \begin{Proposition}\label{data}
Let the symplectic vector spaces $(\ker(\square_g),\Xi)$, $(U_0 / \ker(G),{\omega})$ and the factor map $G_0$ be defined as above. Then we have for every $f, f' \in U_0$  with $\phi = G_0([f']) = G f', \psi = G_0([f]) = G f \in \ker \square_g$,
\begin{equation*}
  \Xi(\phi,\psi)=\omega([f'],[f]).
\end{equation*}
 \end{Proposition}

\begin{proof} Without loss of generality we may consider $M \cong \R \times \Sigma$ and suppose that  $\supp(f) \subset (t_{1},t_{2})\times\Sigma$ for some real $t_1 < t_2$. Then we have for every $\phi \in \ker(\square_g)$ upon integrating by parts twice,
\begin{multline*}
  \int_{(t_{1},t_{2})\times\Sigma}\phi\square_{g}G^{+}(f)\nu_{g} =
\int_{(t_{1},t_{2})\times\Sigma}\square_{g}\phi G^{+}(f)\nu_{g} \\ - \int_{\Sigma_{t_{2}}}\left(\phi \nabla_n G^{+}(f)-G^{+}(f)\nabla_n\phi \right)\mu_{h_2}+\int_{\Sigma_{t_{1}}}\left(\phi \nabla_n G^{+}(f)-G^{+}(f)\nabla_n\phi \right)\mu_{h_1}.
\end{multline*} 
Using the fact that $\square_{g}\phi=0$ and that $\Sigma_{t_{1}}$ is disjoint\footnote{Because $\supp(G^{+} f) \subseteq J^+(\supp(f)) \subseteq (t_1 + \eps, \infty) \times \Sigma$ for some $\eps > 0$. from $\supp(G^{+} f)$}  we obtain
\begin{eqnarray*}
 \int_{(t_{1},t_{2})\times\Sigma}\phi\square_{g}G^{+}(f)\nu_{g}=-\int_{\Sigma_{t_{2}}}\left(\phi \nabla_n G^{+}(f)-G^{+}(f)\nabla_n\phi \right)\mu_{h_2}.
\end{eqnarray*} 
Similarly, from the causal properties again  we have that $\Sigma_{t_{2}}$ and $\supp(G^{-}(f))$ are disjoint. Therefore $G^{+} f |_{\Sigma_{t_{2}}}=G f |_{\Sigma_{t_{2}}}$ and $\nabla_n G^{+} f |_{\Sigma_{t_{2}}}=\nabla_n G f |_{\Sigma_{t_{2}}}$ which gives
$$
 \int_{(t_{1},t_{2})\times\Sigma}\phi\square_{g}G^{+}(f)\nu_{g}=-\int_{\Sigma_{t_{2}}}\left(\phi \nabla_n G(f)-G(f)\nabla_n\phi \right)\mu_{h_2}.
$$
Recalling that  $\psi = G(f)$  and  $t_1 < t < t_2$ in $\supp(f)$  we obtain
$$
 \Xi(\phi,\psi) = \int_{(t_{1},t_{2})\times\Sigma}\phi f\nu_{g} = \int_{M}\phi f\nu_{g}.
$$
Here, we use also the assumption $G(f')= \phi$ to proceed with
$$
  \int_{M}\phi f\nu_{g}=\int_{M}G(f') f\nu_{g}=\tilde{\omega}(f',f) = \omega([f'],[f]).
$$
\end{proof}

We have established a symplectomorphism between the spaces
$(\text{ker}(\square_g),\Xi)$ and $(U_0 /
\text{ker}(G),{\omega})$. This implies that the functor {\rm CCR} will
give isomorphic $C^*$-algebras in the quantisation. Therefore, the
result shows that one can use either the elements of $U_0 / \ker(G)$
or those of $\ker(\square_g)$ to construct the algebra of quantum
observables.

\subsection{The physical quantum states}

Finally, in order to construct a full quantum field theory in a low
regularity spacetime, a suitable choice of quantum states must be
made.  Usually, quantum states $\Lambda$ are given by certain positive
linear functionals on the quasi-local $C^{*}$-algebra.  A common
candidate for the physical quantum states in the smooth case are the
quasi-free states that satisfy the microlocal spectrum condition. as
described below.

To be precise, given a real scalar product $\mu \colon
\ker(\square_g)\times\ker(\square_g) \to \R$ satisfying $|\Xi(\phi,
\psi)|^{2}\le\frac{1}{4} \mu(\phi, \psi)\mu(\phi, \psi)$ for all
$\phi,\psi \in \ker(\square_g)$, we define a quasi-free state by
$\Lambda_{\mu}(W(\phi))=e^{\frac{1}{2}\mu(\phi,\phi)}$ \cite{wald}.

To specify the microlocal spectrum condition, we need to define
appropriate subsets of $T^*(M \times M) \setminus 0$, i.e., the
cotangent bundle with the zero section removed, and the two-point
function of the state $\Lambda_\mu$, which is a distribution on $M
\times M$.  Let
\begin{multline*}
  C=\big\{(x_{1},\eta,x_{2},\tilde{\eta})\in T^{*}(M\times M)\backslash0;  
  g^{ab}(x_{1})\eta_{a}{\eta}_{b}=0,g^{ab}(x_{2})\tilde{\eta}_{a}\tilde{\eta}_{b}=0, (x_{1},\eta)\sim(x_{2},\tilde{\eta})\big\}, \\
  \text{ and } C^{+}=\left\{(x_{1},\eta,x_{2},\tilde{\eta})\in C; \eta^{0}\ge0,\tilde{\eta}^{0}\ge0\right\},
\end{multline*}
\noindent
where $(x_{1},\eta)\sim(x_{2},\tilde{\eta})$ means that $\eta$, $\tilde{\eta}$ are cotangent to the null geodesic $\gamma$ at $x_{1}$, $x_{2}$ respectively, and parallel transports of each other along $\gamma$.
 The value of the two point function of a state $\Lambda_\mu$ acting on the elements of the algebra defined by $\phi$ and $\psi$ is 
 $$
\langle \Lambda_2, \phi \otimes \psi \rangle :=-\frac{\partial^2}{\partial s\partial t}\Lambda_{\mu}(W(t\phi)W(s\psi))|_{s=t=0}=-\frac{\partial^2}{\partial s\partial t}\left(\Lambda_{\mu}[ W(s\phi+t\phi)]e^{\frac{ist\Xi(\phi,\psi)}{2}}\right)|_{s=t=0}.
$$
Using the isomorphism between $\ker(\square_g)$ and $V(M)$ the two-point function can be seen to induce a bidistribution on spacetime, i.e.,  $\Lambda_2\in \mathcal{D}'(M\times M)$.

 \begin{Definition}
 A quasi-free state $\Lambda_{H}$ on the algebra of observables satisfies the microlocal spectrum condition if its two point function $\Lambda_{2H}$ is a distribution $\mathcal{D}'(M\times M)$ and satisfies the following wavefront set condition
$$
  WF'(\Lambda_{2H})=C^{+},
$$
 where $WF'(\Lambda_{2H}):= \{(x_1, \eta; x_2, -\tilde{\eta}) \in T^{*}(M\times M); (x_1, \eta; x_2, \tilde{\eta}) \in WF(\omega_{2H})\}.$
 
 \end{Definition}
 The states that satisfy the microlocal spectrum condition are called Hadamard states and their class includes ground states and KMS states (\cite{fewster, rad, kay}).

 In the low regularity setting we require a generalisation of Hadamard states. A larger class of states, called adiabatic states of order $N$ and characterised in terms of their Sobolev-wavefront set, has been obtained by Junker and Schrohe \cite{schrohe}. These states are natural candidates to replace the Hadamard states in spacetimes with limited regularity.  In particular,  quantum ground states have been constructed in static spacetimes using semigroup techniques \cite{ds} and they can be described as  adiabatic states \cite{ss}. We briefly recall the definition of this class of states and of the Sobolev wavefront set.
 
 \begin{Definition}
 A quasi-free state $\Lambda_{N}$ on the algebra  of observables
is called an adiabatic state of order $N\in\mathbb{R}$ if its two-point function $\Lambda_{2N}$ is a
bidistribution that satisfies for every $s\le N +\frac{3}{2}$ the  $H^s$-wavefront set condition 
$$
  {WF^{s}}'(\Lambda_{2N})\subset C^{+},
$$
where $WF^{s}$ denotes the refinement of the notion of the wavefront set in terms of Sobolev regularity (\cite{fourier2}), i.e.,$(x,\xi)\not\in WF^{s}(u)$ if and only if $u=u_1+u_2$ with  $u_1\in H^s$ and $(x,\xi)\not\in WF(u_2)$.
  \end{Definition}

\begin{appendix}
\section{Regularisation methods and generalised functions} 

In this section we gather a minimum of notions required from the theory of Colombeau generalised functions and regularisation methods for Lorentzian metrics. For a comprehensive introduction to the theory of Colombeau algebras we refer to \cite{Col,GKOS:01}, the details about the approximation results for Lorentzian metrics can be found in \cite{KSV,cg,ladder}.

Let $E$ be a locally convex topological vector space whose topology is given by the family of seminorms $\{ p_j \}_{j\in J}$. The elements of 
$$ \mathcal M_E:=\{(u_{\e})_{\e}\in E^{(0,1]}:\forall j\in J \: \exists N \in \mathbb N_0\:\:\: p_j(u_{\e})=O(\e^{-N})\:\:\mathrm{as}\:\:\e\rightarrow 0 \}
$$
and
$$ \mathcal N_E:=\{(u_{\e})_{\e}\in E^{(0,1]}:\forall j\in J \: \forall q \in \mathbb N_0\:\:\: p_j(u_{\e})=O(\e^{q})\:\:\mathrm{as}\:\:\e\rightarrow 0 \}
$$
are called $E$-\emph{moderate} and $E$-\emph{negligible}, respectively. Defining operations componentwise turns $\mathcal N_E$ into a vector subspace of $\mathcal M_E$. We define the \emph{generalised functions based on $E$} as the quotient $\mathcal G_E:=\mathcal M_E / \mathcal N_E$. If $E$ is a differential algebra, then $\mathcal N_E$ is an ideal in $\mathcal M_E$ and $\mathcal G_E$ is a differential algebra as well, called the Colombeau algebra based on $E$.

Let $\Omega$ be an open subset of $\mathbb R^n$. By choosing $E= C^{\infty}(\Omega)$ with the topology of uniform convergence of all derivatives one obtains the standard Colombeau algebra $\mathcal G_{C^{\infty}(\Omega)}=\mathcal G(\Omega)$; here we will mainly use $E=H^{\infty}(\Omega)=\{h\in  C^{\infty}(\overline{\Omega}):\partial^{\alpha}h\in L^2(\Omega)\:\forall \alpha\in\mathbb N^{n}_0 \}$  with the family of semi-norms
$$ \norm{h}{H^k}=\big(\sum_{|\alpha|\leq k} \norm{\partial^{\alpha}h}{L^2}^2 \big)^{1/2}\quad(k\in \mathbb N_0)
$$
or $E=W^{\infty,\infty}(\Omega)=\{h\in C^{\infty}(\overline{\Omega}):\partial^{\alpha}h\in L^{\infty}(\Omega)\:\forall \alpha\in\mathbb N^{n} \}$
with the family of semi-norms
$$ \norm{h}{W^{k,\infty}}=\max\limits_{|\alpha|\leq k} \norm{\partial^{\alpha}h}{L^{\infty}}\quad(k\in \mathbb N_0).
$$
We employ the notation
$$\mathcal G_{L^2}(\Omega):=\mathcal G_{H^{\infty}(\Omega)}\quad\mathrm{and}\quad
\mathcal G_{L^{\infty}}(\Omega):=\mathcal G_{W^{\infty,\infty}(\Omega)}.
$$

Colombeau algebras contain the distributions as a linear subspace, though not every element of a Colombeau algebra is a regularisation of a distribution.  Their elements are equivalence classes of nets of smooth functions, $\mathcal G(\Omega)\ni u=[(u_{\e})_{\e}]$. We say that a Colombeau function $u$ is \emph{associated with a distribution} $w\in \mathcal D'(\Omega)$ if  some (and hence every) representative $(u_{\e})_{\e}$ converges to $w$ in $\mathcal D'(\Omega)$. The distribution $w$ represents the macroscopic behaviour of $u$ and is called the \emph{distributional shadow} of $u$. 

A generalised function $u\in\mathcal G(\Omega)$ is said to be of \emph{$L^{\infty}$-log-type} if 
$$\norm{u_{\varepsilon}}{L^{\infty}(\Omega)}=O(\mathrm{log}(1/ \varepsilon))\quad \mathrm{as} \quad \varepsilon \rightarrow 0.$$ 
Logarithmic growth conditions on the coefficients of a differential equation are typical in statements on existence and uniqueness of generalised solutions. These results are usually derived from a detailed analysis of regularisation techniques and Colombeau solutions often lead to very weak solutions in the sense of \cite{Ruzhansky}.

A related methodology of regularisation is used in the approximation results of \cite{KSV} which show how to approximate a globally hyperbolic $C^{1,1}$ metric
by a smooth family of globally hyperbolic metrics while controlling
the causal structure. We recall from  \cite[Sec.\ 3.8.2]{ladder}, \cite[Sec.\ 1.2]{cg} 
that for two Lorentzian metrics $g_1$,
$g_2$, we say that $g_2$ has \emph{strictly wider light cones} than $g_1$, denoted by 
\begin{equation}
 g_1\prec g_2, \text{ if for any tangent vector } X\not=0,\ g_1(X,X)\le 0 \text{ implies that } g_2(X,X)<0.
\end{equation}
Thus any $g_1$-causal vector is $g_2$-timelike.
The key result is \cite[Prop.\ 1.2]{cg}, which we give here in the strengthened  version of  \cite[Prop.\ 2.3]{KSV}. 

\begin{Proposition}\label{CGapprox} Let $(M,g)$ be a $C^0$-spacetime 
and let $h$ be some smooth
background Riemannian metric on $M$. Then for any $\eps>0$, there exist smooth
Lorentzian metrics $\check g_\eps$ and $\hat g_\eps$ on $M$ such that $\check g_\eps
\prec g \prec \hat g_\eps$ and $d_h(\check g_\eps,g) + d_h(\hat g_\eps,g)<\eps$,
where  
\begin{equation}\label{CGdh}
d_h(g_1,g_2) := \sup_{p\in M,0\not=X,Y\in T_pM} \frac{|g_1(X,Y)-g_2(X,Y)|}{\|X\|_h
\|Y\|_h}.
\end{equation}
Moreover, $\hat g_\eps(p)$ and $\check g_\eps(p)$ depend smoothly on $(\eps,p)\in \R^+\times M$, and if
$g\in C^{1,1}$ then letting $g_\eps$ be either $\check g_\eps$ or $\hat g_\eps$,
we additionally have 
\begin{itemize}
\item[(i)] 
  For any compact subset $K \Subset M$ there exists a sequence
  $\eps_j\searrow 0$ such that $\hat g_{\eps_{j+1}}\prec \hat
  g_{\eps_{j}}$ on $K$ (resp.\ $\check g_{\eps_{j}}\prec \check
  g_{\eps_{j+1}}$ on $K$) for all $j\in \N_0$.

\item[(ii)] If $g'$ is a continuous Lorentzian metric with $g\prec g'$
  (resp.\ $g'\prec g$) then $\hat g_\eps$ (resp.\ $\gec$) can be
  chosen such that $g\prec \hat g_\eps \prec g'$ (resp.\ $g'\prec \gec
  \prec g$) for all $\eps$.

\item[(iii)] There exist sequences of smooth Lorentzian metrics
  $\check g_j\prec g \prec \hat g_{j}$ ($j\in \N$) such that
  $d_h(\check g_j,g) + d_h(\hat g_j,g)<1/j$ and $\check g_j \prec
  \check g_{j+1}$ as well as $\hat g_{j+1}\prec \hat g_{j}$ for all
  $j\in \N$.

\item[(iv)] If $g$ is $C^{1,1}$ and globally hyperbolic then the $\hat
  g_\eps$ (and $\gec$) can be chosen to be globally hyperbolic.

\item[(v)] If $g$ is $C^{1,1}$ then the regularisations can in
  addition be chosen such that they converge to $g$ in the
  $C^1$-topology and such that their second derivatives are bounded,
  uniformly in $\eps$ on compact sets.
\end{itemize}
\end{Proposition}

\begin{rem}\label{ghstab}

  In our application the main point we will need compared to
  \cite[Sec.\ 1.2]{cg} is property (iv) which guarantees that for
  globally hyperbolic metrics there exist approximations with strictly
  narrower (wider) lightcones that are themselves globally hyperbolic.
  Extending methods of \cite{Ger70}, it was shown in \cite{BM11} that
  global hyperbolicity is stable in the interval
  topology. Consequently, if $g$ is a smooth, globally hyperbolic
  Lorentzian metric, then there exists some smooth globally hyperbolic
  metric $g'\prec g$ (resp. $g'\succ g$).  Constructing $\hat g_\eps$
  resp.\ $\hat g_j$ as in (ii) then automatically gives globally
  hyperbolic metrics (cf.\ \cite[Sec.\ II]{BM11}).
\end{rem}

\section{Function Spaces}


The (real) Hilbert space $L^2(M,g)$ is used in the section on Green operators to formulate adjointness properties and is defined as follows: Recall that for any Lorentzian manifold $(M,g)$ we have a unique positive density $\mu_g$ on $M$ \cite[Proposition 2.1.15]{Waldmann}, which has the local coordinate expression $\sqrt{|\det (g_{ij})|} \, |dx^0 \wedge \ldots \wedge dx^n|$; in case of a  Lipschitz continuous metric $g$ the density $\mu_g$ is continuous and induces a positive Borel measure on $M$, which we employ to define the corresponding $L^2$ space and denote it by $L^2(M,g)$. If $M$ is orientable, then we have a global volume form $\nu_g$ on $M$ \cite[Ch. 7]{nei} from which the density $\mu_g$ can be obtained.

We consider the (real) Sobolev spaces $H^{m}(M)$ for a nonnegative integer $m$ to be defined with respect to some chosen smooth Riemannian background metric on $M$ as described in \cite{heb}, i.e., by completion of the space of (real) smooth functions whose covariant derivatives up to order $m$ are square integrable with respect to the positive Borel measure on $M$ associated with the Riemannian metric (cf.\ \cite[$\S$III.3]{cha}; it can be written in terms of a global Riemannian volume form, if $M$ is orientable). Recall (\cite[Theorems 2.7 and 2.8]{heb}) that the space $\mathcal{D}(M)$ (of smooth compactly supported test functions) is dense in $H^1(M)$, if $M$ is complete with respect to the chosen Riemannian metric, and also in $H^m(M)$ for $m \geq 2$, if, in addition, Riemannian curvature bounds hold as well. 

On a compact manifold $M$, the definition of $H^m(M)$ is independent of the chosen Riemannian background metric 
(\cite[Proposition 2.3]{heb}) and, similarly, one concludes that Sobolev norms induced by two different Riemannian metrics on functions with support contained in a fixed compact subset are equivalent. This observation guarantees that the following two spaces are independent of the chosen background metric, namely the compactly supported Sobolev functions $H^m_{\text{comp}}(M) := \{ f \in H^m(M) \mid \text{supp}(f) \text{ is compact in } M \}$ and the local space $H^m_{\text{loc}}(M) := \{ f \colon M \to \R \text{ measurable} \mid \forall \varphi \in \mathcal{D}(M) \colon \varphi f \in H^m(M) \}$ (in fact, $\varphi f \in H^m_{\text{comp}}(M)$ in the latter case). 

In the context of the function space topologies for the current paper,
we simply consider $H^m_{\text{comp}}(M)$ as a subspace of the Banach
space $H^m(M)$, hence it is normed and not complete. One could equip
$H^m_{\text{comp}}(M)$ with a complete (non-metrizable) locally convex
vector space topology, e.g., as in \cite{barwafo} or \cite[Part II,
Chapter 31]{Treves}, via a strict inductive limit construction which
turns it into a so-called (LF)-space, but we prefer to formulate our
results more directly in terms of the inherited Sobolev norm.

For $H^m_{\text{loc}}(M)$ we have the family of semi-norms $f \mapsto
\| \varphi f \|_{H^m(M)}$, parametrised by $\varphi \in
\mathcal{D}(M)$, which provides us with a Fr\'{e}chet space topology
on $H^m_{\text{loc}}(M)$ (cf.\ \cite[Part II, Chapter 31]{Treves} or
\cite{barwafo}). We clearly have $H^m(M) \subseteq
H^m_{\text{loc}}(M)$ (with continuous embedding).

If $K$ is a compact subset of $M$ and $f \in H^m_{\text{loc}}(M)$ we
occasionally abuse the notation and write $\| f\|_{H^m(K)}$ to mean
the value obtained when the integrals defining $\| \varphi f
\|_{H^m(M)}$ are only evaluated on $K$ and $\varphi \in
\mathcal{D}(M)$ is a cut-off such that $\varphi = 1$ on $K$. (No
cut-off is required if $f \in H^m(M)$.)

In case of $M = (0,T)\times \Sigma$ we may choose the Riemannian
background metric in the form $dt \otimes dt + \gamma$, where $\gamma$
is a Riemannian metric on $\Sigma$. We will then often consider a
function $ v\in L^2((0,T)\times \Sigma)$ as a map $t\mapsto {v}(t)$
from the interval into the Hilbert space $L^2(\Sigma)$ in the sense
that ${v}(t)(x)= v(t,x)$ holds pointwise for continuous $v$. Thanks to
Fubini's theorem, we may then write
\begin{equation}
\| v\|_{L^2((0,T)\times \Sigma)}^2=\int_0^T \| v(t) \|_{L^2(\Sigma)}^2 dt. 
\end{equation}
For general constructions with measurable functions valued in Banach spaces we refer to \cite{leo,kab}; in particular we will make use of
the isomorphism $L^2((0,T)\times \Sigma)\cong L^2((0,T),L^2(\Sigma))$ \cite[Theorem 8.28]{leo}. If $v$ is differentiable and interpreted as
a function $t \mapsto {v}(t)$, we will occasionally denote the partial derivative $\partial_t v$ by ${\dot v}$ and write
$\partial_t$ for the corresponding vector field on $(0,T)\times \Sigma$.  The space $C^k([0, T]; H^m(\Sigma))$ ($k$ a nonnegative integer) consists of all $k$ times continuously (strongly) differentiable functions (if $k = 0$, simply continuous functions) $v \colon [0,T]  \to  H^m(\Sigma)$ with finite norm  
\begin{equation}
 \|{v}\|_{C^k([0,T],H^m(\Sigma))}:=\max_{0\leq j \leq k}\sup_{t\in[0,T]}\| \partial_{t}^j {v}(t)\|_{H^m(\Sigma)} < \infty. 
\end{equation}
We have $C^k([0, T]; H^m(\Sigma)) \subseteq L^2((0,T),H^m(\Sigma))$ (with continuous embedding).

In place of a bounded time interval we will occasionally consider the basic spacetime to be $\R \times \Sigma$
and deal with function spaces of Bochner measurable maps from $\R$ to some of the Sobolev-type Hilbert spaces (cf.\ \cite[Chapter 8]{leo}), in particular, $L^2(\R,H^1(\Sigma))$. We will then use the notation $L^2_{\text{loc}}(\R,H^1(\Sigma))$ for the set of all Bochner-measurable functions $v \colon \R \to H^1(\Sigma)$ such that for every compact subinterval $I \subset \R$ the restriction $v |_I$ belongs to $L^2(I,H^1(\Sigma))$.

In looking at energy estimates on $\R \times \Sigma$ we will also need versions of the Sobolev norms where the derivatives are taken in both the space and time directions but the integration and volume form are confined to the $t=\tau$ level hypersurfaces $S_\tau := \{\tau\} \times \Sigma$. These norms will be denoted  by
\begin{equation}\label{restricted}
 \| u \|_{\tilde H^m(S_\tau)}=\left(\sum_{j=1}^m \int_{S_\tau}(u^2+(\partial^j_t u)^2+|\tnab^ju|^2) d \mu_\tau \right)^{\frac{1}{2}},
\end{equation} 
where $\tnab$ is the covariant derivative with respect to the spatial background metric $\gamma$ and $\mu_\tau$ is the Riemannian measure on $S_\tau$ which is just that given by the spatial metric $\gamma$.

Finally let us adapt the basic function space structures to the
situation of a general globally hyperbolic $C^{1,1}$ spacetime $(M,g)$
with Cauchy hypersurface $\Sigma$, where we suppose that---according
to the discussion in the subsection on $C^{1,1}$ causality theory---we
have chosen a smooth temporal function $t \colon M \to \Real$ such
$\Sigma = t^{-1}(0)$ and a corresponding diffeomorphism $\Phi \colon M
\to \R \times \Sigma$. For $\tau \in \R$ denote the corresponding
level surface by $\Sigma_\tau := t^{-1}(\tau) = \Phi^{-1}(\{ \tau \}
\times \Sigma)$, hence $\Sigma_0 = \Sigma$, and consider again a
background Riemannian metric of the form $ h=dt\otimes dt+\gamma$ on
the product manifold $\R \times \Sigma$, which in turn provides us
with the convenient background metric $\Phi^* h$ on $M$. In the
sequel, all Sobolev spaces on submanifolds of $M$ or $\R \times
\Sigma$ will be considered to be defined via Riemannian metrics
induced by $\Phi^* h$ or $h$, respectively.  Let $\Phi_\tau$ denote
the induced diffeomorphism $\Sigma_\tau \to \Sigma$, i.e., $\Phi(x) :=
(\tau, \Phi_\tau(x))$ for every $x \in \Sigma_\tau$.

We will commit another abuse of notation and a somewhat naive simplification in defining now the spaces  $C^k(I, H^m(\Sigma_t))$  for the case of a compact interval $I = [0,T]$ or for $I = \R$ without using the full theory of more sophisticated constructions in terms of sections, e.g., as in \cite{barwafo}. Let $B_m(I)$ denote the set of all maps $u \colon I \to \bigcup_{\tau \in I} H^m(\Sigma_\tau)$ such that $u(\tau) \in H^m(\Sigma_\tau)$ for every $\tau \in I$. Then we have that $u \in B_m(I)$ implies (equivalently) $u(\tau) \circ \Phi_\tau^{-1} \in H^m(\Sigma)$ for every $\tau \in I$. We define $C^k(I, H^m(\Sigma_t))$ to be the subset of those elements $u \in B_m(I)$ such that the map $\tau \mapsto u(\tau) \circ \Phi_\tau^{-1}$ belongs to $C^k(I, H^m(\Sigma))$. 

For elements $u \in C^k(I, H^m(\Sigma_t))$ we can then also define the norms over spatial domains, but involving derivatives in space and time directions, such as $\| u \|_{\tilde H^m(\Sigma_\tau)}$ via the corresponding $ \| . \|_{\tilde H^m(S_\tau)}$-norm  evaluated for the associated map $\tau \mapsto u(\tau) \circ \Phi_\tau^{-1}$ in $C^k(I, H^m(\Sigma))$.

Note that  the definition of the spaces  $C^k(I, H^m(\Sigma_t))$ depends on the splitting $M \cong \R \times \Sigma$ and on the
choice of temporal function. However, the reasoning in the main text tries to use the temporal function only in intermediate calculations and afterwards gives formulations of results essentially in ``pure'' spacetime terms without recourse to the splitting.

\end{appendix}


\bibliographystyle{abbrv}
\bibliography{hsssv}

\end{document}